%% file: main.tex
\documentclass[prd]{article}
\usepackage[utf8]{inputenc}
\usepackage{amsfonts}
\usepackage[margin=1in]{geometry}
\usepackage{hyperref,amsthm,enumerate,
            imakeidx, xparse, mathtools,
            xcolor, color, colortbl, 
            amssymb, tensor, 
            soul, graphicx ,titlesec,  appendix, tikz,
            amsmath,scalerel, comment, float, multirow, multicol, caption, subcaption} %boisik}%, musixtex}
\usepackage[most]{tcolorbox}
\usepackage{stackengine}
\usetikzlibrary{arrows,decorations.markings}
\usetikzlibrary{shapes.misc}
\usetikzlibrary{arrows.meta}
\usetikzlibrary{angles,quotes}
\usepackage{tkz-euclide}
\usetikzlibrary{intersections}
\usetikzlibrary{calc}
\usetikzlibrary{backgrounds}
\usepackage{cite}
\usepackage{booktabs}
\usepackage{pgfplots}
\hypersetup{
pdfstartview = {FitH},
}
\hypersetup{
	colorlinks=true,         
	linkcolor=brown,          
	citecolor=red,        
	urlcolor=blue            
}

\newcommand{\be}{\begin{equation}}
\newcommand{\ee}{\end{equation}}
\newcommand{\beq}{\begin{equation}}
\newcommand{\eeq}{\end{equation}}

\newtheorem{definition}{Definition}

\newtheorem{proposition}{Proposition}

\newcommand{\cA}{{\cal A}}

\newcommand{\cC}{{\cal C}}
\newcommand{\cD}{{\cal D}}

\newcommand{\cF}{{\cal F}}

\newcommand{\cK}{{\cal K}}
\newcommand{\cL}{{\cal L}}
\newcommand{\cM}{{\cal M}}

\newcommand{\cO}{{\cal O}}
\newcommand{\cP}{{\cal P}}
\newcommand{\cS}{{\cal S}}

\newcommand{\cU}{{\cal U}}
\newcommand{\cV}{{\cal V}}
\newcommand{\cZ}{{\cal Z}}

\def\g{\mathfrak g}

\def\sl{\mathfrak{sl}}
\def\an{\mathfrak{an}}
\def\su{\mathfrak{su}}

\def\bbR{\mathbb{R}}

\def\bbC{\mathbb{C}}
\def\bbK{\mathbb{K}}

\def\SU{\mathrm{SU}}

\def\AN{\mathrm{AN}}

\def\dd{\mathrm{d}}
\def\dh{\mathrm{d} h}

\def\dx{\mathrm{d} x}

\def\id{\textrm{id}}
\def\idd{\textrm{id}^{\ot 3}}

\def\coV{\overline \cV}
\def\coK{\overline \cK}

\def\rhd{\triangleright}
\def\lhd{\triangleleft}
\def\la{\langle}
\def\ra{\rangle}
\def\ot{\otimes}
\def\op{\oplus}

\newcommand\blackbowtie{\mathrel{\scalerel*{\blacktriangleright\joinrel\blacktriangleleft}{x}}}

\newcommand\semicoprodl{\mathrel{\scalerel*{\mathrel{>}\joinrel\blacktriangleleft}{x}}}
\def\uDe{\underline{\Delta}}
\def\um{\underline{m}}
\def\sDe{\Delta_{\cV}}
\def\dela{\hat{\delta}_{A}}
\def\delh{\delta_{H}}

\def\demi{\frac{1}{2}}
\def\mone{^{-1}}

\definecolor{darkgreen}{rgb}{0.0, 0.5, 0.0}
\definecolor{darkorange}{rgb}{1.0, 0.5, 0.0}
\definecolor{darkbrown}{rgb}{0.6, 0.3, 0.0}
\definecolor{cargreen}{rgb}{0.0, 0.8, 0.6}
\definecolor{darkpink}{rgb}{0.9, 0.3, 0.5}
\definecolor{gold}{rgb}{0.85, 0.65, 0.1}
\definecolor{bluette}{rgb}{0.3, 0.25, 0.55}
\definecolor{darkred}{rgb}{0.75, 0.0, 0.0}
\definecolor{darkblue}{rgb}{0.0, 0.3, 0.65}

\usepackage{authblk}

\makeatletter
\renewenvironment{abstract}{%
    \if@twocolumn
      \section*{\abstractname}%
    \else %% <- here I've removed \small
      \begin{center}%
        {\bfseries\sffamily\abstractname\vspace{\z@}}%  %% <- here I've added \Large
      \end{center}%
      \quotation
    \fi}
    {\if@twocolumn\else\endquotation\fi}
\makeatother

\title{Group field theory on quantum groups}

\begin{document}

\author[1]{\sffamily Florian Girelli\thanks{florian.girelli@uwaterloo.ca}}
\author[2,1]{\sffamily Matteo Laudonio\thanks{matteo.laudonio@u-bordeaux.fr}}

\affil[1]{\small Department of Applied Mathematics, University of Waterloo, 200 University Avenue West, Waterloo, Ontario, Canada, N2L 3G1}
\affil[2]{\small Univ. Bordeaux, LABRI, 351 Cours de la Libération, 33400 Talence, France}
\maketitle
\begin{abstract}
 We introduce the framework of Hopf algebra field theory (HAFT) which generalizes the notion of group field theory  to the  quantum group  (Hopf algebra) case. We focus in particular on the 3d case and show how the HAFT we considered is topological. The highlight of the construction is the notion of  plane-wave  which  leads, in the specific example of $\SU_q(2)$ with $q$ real, to a discretization of the Euclidian $BF$ action with a negative cosmological constant. This work can be viewed as the generalization of the seminal work by Baratin and Oriti \cite{Baratin:2010wi}: we have a  (non-commutative) formulation of simplicial 3d gravity in the presence of a non-zero cosmological constant.        
\end{abstract}

\tableofcontents

\section{Introduction}
\label{Sec_HAFT_Intro}

The Turaev-Viro-Barrett-Westbury (TV) model is a topological state-sum model defined in terms  of the category of representations of a (weak) Hopf algebra/quantum group (or even a spherical category \cite{Turaev:1992hq, BARRETT1999-1, barrett1993}). Given a 3d manifold $\cM$, it consists in decorating the edges of a triangulation of $\cM$ in terms of the representation labels. The key ingredient is the (q-deformed) $6j$-symbol.
From a field theory perspective, it is related to the Chern-Simons theory or 3d gravity with a non-zero cosmological constant. Indeed, it can be seen as providing a discretization of the amplitude of the $BF$ action (with gauge symmetry given by the Lie algebra $\g$) with a cosmological constant.
\be
\label{amplitude0}
    \cZ_{BF} = \int [\dd B] [\dd A] \, e^{i\int_{\cM} \, B^I \wedge F_I + \frac{\Lambda}{6} \epsilon_{IJK} B^I\wedge B^J \wedge B^K} \,,
\ee
where $B$ is a 1-form with value in the Lie algebra $\g^*$, $A$ is the 1-form connection with value in the Lie algebra\footnote{The Lie algebras $\g$ and $\g^*$ are dual to each other, and equipped with a co-cycle structure inherit form the classical Drinfeld double.} $\g$ and $F(A)$ is the associated curvature. As a 3d quantum gravity model, the TV model is a spinfoam and it provides the full information about the 3d quantum gravity regime (with no matter fields, but particles, as topological defects can be introduced). 

The TV model has no divergence when one deals with, for example, the $q$ root of unity deformation of $\SU(2)$ (a quasi-Hopf algebra), which corresponds in the (3d) gravity language to the Euclidian case with a positive cosmological constant. In the other cases of signature and sign of the cosmological constant, it is divergent, just as the Ponzano-Regge (PR) model\cite{Ponzano_Regge_1969, FREIDEL2000237}. The PR model corresponds to the state sum model based on the representations of the gauge group of gravity with a zero cosmological constant.

The appearance of a quantum group/deformed symmetry structure can be derived from the Hamiltonian formulation of Chern-Simons or gravity, through some discretization procedure of the fields \cite{Alekseev:1994au, Alekseev:1994pa, Dupuis:2020ndx}. In particular from the gravity picture, one has to choose the proper variables to actually have the gauge transformations depending on the cosmological constant, which fits the quantum theory defined in terms a quantum group symmetry as gauge symmetry \cite{Dupuis:2020ndx}, with the deformation parameter depending on the cosmological constant. 

One can also derive the PR model from a direct discretization of the amplitude \eqref{amplitude0}, with $\Lambda=0$. Indeed, when the cosmological constant is zero, this amplitude can be essentially seen as a plane-wave and therefore as a Dirac delta function implementing that the curvature has to be zero (the $B$ field can be seen as some kind of Lagrange multiplier to implement the zero curvature condition). 
\begin{align}
    \label{amplitude1}
    \cZ^{\Lambda=0}_{BF} &= \int [dB] [dA] \, e^{i\int_{\cM} \, B^I \wedge F_I} \sim     
    \int [dA] \, \delta (F) \\
    &\rightarrow\,\,\,
    \cZ^{\Lambda=0}_{d} = \int [dX] [dg] \, \prod_e \, e^{i \la X_e \,,\, g_e\ra} = \int [dg] \, \prod_e \, \delta(g_e) \,, \label{amplitude11}
\end{align}
where $X_e \in \bbR^3$ is associated to the edge $e$ of a triangulation and $g_e$ is the closed holonomy on the loop spanning the face $e^*$ dual to the edge $e$. The notation $\la X_e \,,\, g_e\ra$ means that we consider a specific coordinate choice $k(g_e)$ for $g_e$ (bearing in mind that multiple coordinate patches are typically needed) so that 
\be
    \label{pw0}
    e^{i\la X_e \,,\, g_e\ra} \equiv e^{i X^a k_a(g_e)} \,,
\ee
and we recover indeed a plane-wave. We note that as a function of $X$, the plane-wave is a non-commutative function with star product $\star$ (which depends on the choice of coordinates), since we must have 
\be
    \label{pw1}
    e^{i\la X_e \,,\, g_e\ra} = e^{i\la X_e \,,\, \prod_{l\in \partial e^*} \, g_l\ra} = \star_l \, \sigma_l 
    \,\,,\quad 
    \sigma_{l\in \partial e^*} \equiv e^{i\la X_e \,,\, g_l\ra} \,,  
\ee
where the $g_l$ is the holonomy associated to the link $l$ of the dual complex \cite{Baratin:2010wi, Baratin:2010nn, Guedes:2013vi}.  
No analogue of such discretization really exists when the cosmological constant is not zero. Asymptotic analysis of the TV amplitude shows that one can recover the Regge action with a cosmological constant \cite{Mizoguchi:1991hk}. One can nevertheless wonder whether a similar derivation as the PR model could exist for the TV model and in particular whether there would exist a generalized notion of plane-wave which would be used to encode a discretized $BF$ action. One should notice that the $BF$ amplitude with a non-zero cosmological constant should be seen as some modified plane-wave. Indeed, similarly to \eqref{amplitude1}, we have   
\be
    \int \dd^3 B \, e^{i \int \, \big(B \cdot F + \frac{\Lambda}{6} \, B \cdot (B \times B)\big)}
    =
    \frac{2}{|\Lambda|} \, K_0\Big(\sqrt{\frac{4 \, F \cdot (F \times F)}{\Lambda}}\Big) \,,
    \label{Bessel}
\ee
where $K_0(x)$ is the modified Bessel function of the second kind, which can be seen as a smearing of the Dirac delta function. 

\medskip 

We intend to show here that even in the TV case, there is a notion of plane-wave which  gives a discretized version of the $BF$ action with a non-zero cosmological constant. Importantly, the recovered action is actually not the standard one given in \eqref{amplitude0} since it is not defined in terms of the fields which give rise to the quantum group structure upon discretization \cite{Dupuis:2020ndx}. 

\medskip

The framework of this construction will be the group field theory (GFT) approach \cite{Oriti:2006GFT,Freidel:2005qe}, which we will upgrade to a general Hopf algebra context. The reason is that group field theories have a configuration and a momentum formulation which correspond to formulating the model in the triangulation picture or in its dual complex. The two descriptions are related by a Fourier transform and hence in terms of a notion of plane-wave. Such plane-wave will appear as the building block to characterize the discretized $BF$ action action with or without cosmological constant.   

The group field theory approach on top of conveniently generating the different triangulations (or dual complex) as Feynman diagrams, emphasizes the fact that a pair of dual Hopf algebras are actually at play, one decorating the triangulation, the other one decorating the dual complex. The pair of Hopf algebras forms the symmetries of the system, namely the (generalized) Drinfeld double and the duality is encoded in the canonical element of $H \otimes A$ \cite{Majid:1996kd}.
On a practical manner, it is already known that when the cosmological constant is zero, we recover a formulation of simplicial gravity with the metric degrees of freedom, decorating the edges of the triangulation, which are non-commutative \cite{Baratin:2010wi}. Introducing an homogeneous curvature will deform further the metric degrees of freedom and will render the decorations of the dual complex non-commutative.  

We will provide the general framework to construct group field theories for any type of Hopf algebra $H$. The combinatorics we chose will focus on generating 3d triangulations, but it can be generalized for any dimension. Our construction emphasizes that the main structure at play is actually not the Hopf algebra $H$, but rather its generalized Drinfeld double $H \bowtie A$. Our work shares some similarities in the formulation with Krasnov's construction which also sets up the GFT in the Hopf algebra language to introduce matter (or topological excitations) \cite{Krasnov:2005GFT3dMatterDrinfeldDouble}.

In section \ref{Sec_HopfAlgebra_Intro}, we introduce the notations and conventions we will use extensively in the paper. We introduce in particular the notion Fourier transform and hence the notion of plane-wave that we will use to characterise the duality between the Hopf algebras which constitute the Drinfeld double. 

In section \ref{Sec_HopfAlgebraFieldTheory}, we introduce the main ingredients  of a Hopf algebra field theory (HAFT), which generalizes the notion of group field theory. We consider a specific action and show that the Feynman diagrams can be related by the Pachner moves, so that their amplitude are actually topological invariants. 

In section \ref{Sec_Examples}, we construct the main example, namely the case of the Drinfeld double of $\SU_q(2)$, with $q$ real. We recall how the HAFT Feynman diagrams in this specific example can be related to TV model, as Boulatov originally discussed. We then discuss how the plane-wave we have can be related to the discretization of the $BF$ action with a cosmological constant.

We have several appendices which contain the proofs of many of the propositions given in the main text. We have also provided two more examples, the case of the finite group, which is relevant to discuss spinfoams with finite groups \cite{Bahr:2011yc} and the undeformed $\SU(2)$ case.

\section{Conventions and notations }
\label{Sec_HopfAlgebra_Intro}

In this section we recall the key Hopf algebra properties that we need and introduce the relevant notations and conventions. The expert reader can skip this section.

\subsection{Hopf algebra notations}
We note as usual  $\Delta$ the co-product, $S$ the antipode, $\varepsilon$ the co-unit, $m$ (or sometimes for short $\cdot$) the multiplication, $\eta$ (or, for simplicity, $1$) the unit and $\tau$ the permutation map. Use use the symbol $\circ$ for the composition of maps and $\id^n$ for the tensor product of $n$ identity maps.  
A bi-algebra $(A,m,\eta,\Delta,\epsilon)$ over the field $\mathbb{K}$ is a vector space over $\mathbb{K}$ which is both an algebra and a co-algebra in a compatible way. An Hopf algebra $(A,m,\eta,\Delta,\epsilon,S)$ over $\mathbb{K}$ is a bi-algebra over $\mathbb{K}$ equipped with an antipode $S$. \\
Given the Hopf algebra $A$, the co-unit and antipode satisfy the axioms
\begin{align}
    &
    (\id \ot \varepsilon) \circ \Delta
    =
    (\varepsilon \ot \id) \circ \Delta
    =
    \id \,,
    \label{HopfAlg_Axiom_Co-unit} 
    \\
    &
    m \circ (S \ot \id) \circ \Delta = 
    m \circ (\id \ot S) \circ \Delta = 
    \eta \, \varepsilon \,.
    \label{HopfAlg_Axiom_Antipode} 
\end{align}
Given the Hopf algebra $A$, we denote $A^{n} = \bigotimes_i^n A_i$ its associated tensor product and consider the following maps.
\begin{align}
    \Delta^{n} 
    & : A \rightarrow A^{n}
    \nonumber \\
    & \text{with} \quad  
    \Delta^{n} \equiv (\id \ot \dots \ot \id \ot \Delta) \circ \dots \circ \Delta
    \,, \quad  n \geq 2 
    \label{1-d_CoProduct} \\
    m^{n} 
    & : A^{n} \rightarrow A
    \nonumber \\
    & \text{with} \quad  
    m^{n} \equiv m \circ \dots \circ (\id \ot \dots \ot \id \ot m) 
    \,, \quad  n \geq 2
    \label{1-d_Multiplication} \\
    \tau^{n} 
    & : A^{n} \rightarrow A^{n}
    \nonumber \\
    & \text{with} \quad 
    \tau^{n} (a_1 \ot \dots \ot a_n) =
    a_n \ot \dots \ot a_1
    \label{1-d_Permutation} \\
    \varepsilon^n 
    & : A^{n} \rightarrow \mathbb{K} 
    \label{1-d_CoUnit} \\
    \eta^n
    & : \mathbb{K} \rightarrow A^{n} 
    \label{1-d_Unit} 
\end{align}
We note that, since the co-product (resp. the multiplication) is co-associative (resp. associative), the position of the co-product in \eqref{1-d_CoProduct} (resp. the multiplication in \eqref{1-d_Multiplication}) does not matter.
Furthermore, the tensor product $A^{n}$ is in turn an Hopf algebra with associated maps
\begin{align}
    \Delta^{(n)} 
    & : A^{n} \rightarrow A^{n} \ot A^{n} 
    \nonumber \\
    & \text{with} \quad 
    \Delta^{(n)} (a_1 \ot \dots \ot a_n) = ({a_1}_{(1)} \ot \dots {a_n}_{(1)}) \ot ({a_1}_{(2)} \ot \dots \ot {a_n}_{(2)})
    \label{d-d_CoProduct} \\    
    m^{(n)} 
    & : A^{n} \ot A^{n} \rightarrow A^{n}
    \nonumber \\
    & \text{with} \quad 
    m^{(n)} \big( (a_1 \ot \dots \ot a_n) \ot (b_1 \ot \dots \ot b_n) \big) = m(a_1 \ot b_1) \ot \dots \ot m(a_n \ot b_n)
    \label{d-d_Multiplication} \\
    \tau^{(n)}
    & : A^{n} \ot A^{n} \rightarrow A^{n} \ot A^{n}
    \nonumber \\
    & \text{with} \quad 
    \tau^{(n)} \big( (a_1 \ot \dots \ot a_n) \ot (b_1 \ot \dots \ot b_n) \big) =
    (b_1 \ot \dots \ot b_n) \ot (a_1 \ot \dots \ot a_n)
    \label{d-d_Permutation}
    \\
    S^{(n)}
    & : A^{n} \rightarrow A^{n}
    \nonumber \\
    & \text{with} \quad 
    S^{(n)} (a_1 \ot \dots \ot a_n) =
    (Sa_1 \ot \dots \ot Sa_n)
    \label{d-d_Antipode}
\end{align}
Here we used the Sweedler notation for the co-product: $\Delta a = \sum a_{(1)} \ot a_{(2)}$. Given a tensor product of Hopf-algebras, we will later use the notation
\begin{align}
    &
    m_{ij} (1 \ot \dots \ot a_1 \ot \dots \ot a_2 \ot \dots \ot 1) \equiv
    1 \ot \dots \ot m(a_1 \ot a_2) \ot \dots \ot 1 \ot \dots \ot 1 \,,
    \label{ij_Product}
    \\
    &
    \Delta_{ij} a \equiv 1 \ot \dots \ot a_{(1)} \ot \dots \ot a_{(2)} \ot \dots \ot 1 \,, \qquad a \in A
    \,.
    \label{ij_CoProduct}
\end{align}
The map $m_{ij}$ stands for a product of two elements embedded in the tensor product of Hopf algebras, respectively living in the $i^{th}$ and $j^{th}$ places, and the result lives in the $i^{th}$ space, for $i < j$. The map $\Delta_{ij}$ stands for usual co-product embedded in a tensor product of Hopf algebras and its two components live respectively in the $i^{th}$ and $j^{th}$ copies of the algebra. We use a minus sign as notation to imply the antipode on one of the tensor spaces, for instance
\be
    \Delta_{-i \, j} a \equiv 1 \ot \dots \ot Sa_{(1)} \ot \dots \ot a_{(2)} \ot \dots \ot 1\,,
    \qquad a \in A
    \,.
    \label{ij_CoProduct_Antipode}
\ee
These notations can be  generalized to the $n$-dimensional product \eqref{1-d_Multiplication} and co-product \eqref{1-d_CoProduct}, which we denote $m^n_{i_1 \dots i_n}$ and $\Delta^n_{i_1 \dots i_n}$. Similarly, they can be generalized to the $n$-dimensional product \eqref{d-d_Multiplication} and co-product \eqref{d-d_CoProduct}, which we denote $m^{(n)}_{i_1 \dots i_n \,\,\, j_1 \dots j_n}$ and $\Delta^{(n)}_{i_1 \dots i_n \,\,\, j_1 \dots j_n}$. 
\begin{definition}[Simplex maps]
\label{Def_Underline_Maps}
    Consider the maps
    \be
        \um : A^{12} \to A^{6}
        \,\,,\quad
        \uDe : A^{6} \to A^{12}
        \,\,,\quad
        \sDe : A^{6} \to A^{16} \,,
    \ee
    that stand for a combination of products or co-products that have the same pattern of a tetrahedron. 
    The maps are defined as
    \begin{align}
        \um (a_1 \ot \cdots \ot a_{12})
        \equiv \, &
        (a_1 \cdot a_{9}) \ot
        (a_2 \cdot a_{12}) \ot
        (a_3 \cdot a_{4}) \ot
        (a_5 \cdot a_{11}) \ot
        (a_6 \cdot a_{7}) \ot
        (a_8 \cdot a_{10}) \,,
        \label{Underline_Multiplication_3d}
        \\
        \uDe (a_1 \ot \cdots \ot a_6)
        \equiv \, &
        \Delta_{1 \, 9} a_1 \ot 
        \Delta_{2 \, 12} a_2 \ot
        \Delta_{3 \, 4} a_3 \ot
        \Delta_{5 \, 11} a_4 \ot 
        \Delta_{6 \, 7} a_5 \ot 
        \Delta_{8 \, 10} a_6 \,,
        \label{Underline_CoProduct_3d}
        \\
        \sDe (a_1 \ot \cdots \ot a_6)
        \equiv \, &
        \Delta^4_{-16 \, 4 \, -1 \, 6} \, a_1 \cdot
        \Delta^4_{-15 \, 4 \, 2 \, 9} \, a_2 \cdot
        \Delta^4_{-14 \, 4 \, -3 \, 12} \, a_3 \cdot
        \Delta^4_{-11 \, 3 \, 2 \, 10} \, a_4 \cdot
        \Delta^4_{-13 \, 3 \, -1 \, 5} \, a_5 \cdot
        \Delta^4_{-8 \, -2 \, -1 \, 7} \, a_6 \,,
        \label{Underline_CoProduct_Simplex_3d}
    \end{align}
    where $\Delta_{ij}$ (and similarly its generalization $\Delta^4_{i_1 \, i_2 \, i_3 \, i_4}$) is given in \eqref{ij_CoProduct}.
\end{definition}

\subsection{Quantum double}

\begin{definition}[Matched pair bi-algebras] \cite{Majid:1994nw}
    Two bi-algebras $H$ and $A$ form a matched pair if there exist a pair of actions
    \be
        \rhd : A \ot H \to H
        \,\,,\quad
        \lhd : A \ot H \to A 
    \ee
    that satisfy the compatibility relations
    \be
        \begin{aligned}
            &
            (ab) \lhd h = (a \lhd (b_{(1)} \rhd h_{(1)})) (b_{(2)} \lhd h_{(2)})
            \,\,,\qquad
            1 \lhd h = \varepsilon (h) 
            \,,\\
            &
            a \rhd (hg) = (a_{(1)} \rhd h_{(1)}) ((a_{(2)} \lhd h_{(2)}) \rhd g)
            \,\,, \qquad
            a \rhd 1 = \varepsilon (a)
            \,,\\
            & \qquad     
            a_{(1)} \lhd h_{(1)} \ot a_{(2)} \rhd h_{(2)}
            =
            a_{(2)} \lhd h_{(2)} \ot a_{(1)} \rhd h_{(1)} \,.
        \end{aligned}
        \label{MatchedPair}
    \ee
\end{definition}
The following definition introduces a key-object of our construction. 

\begin{definition}[Skew paired bi-algebras.]
\label{Def_SkewSymmetricBi-Algebras}
    \cite{Majid:1994nw}
    Two bi-algebras $H$ and $A$ are skew paired if there exists a map called skew pairing $\sigma : A \ot H \to \bbK$ such that
    \be
        \begin{aligned}
            &
            \sigma(a \cdot b \ot h) =
            \sigma(a \ot h_{(1)}) \, \sigma(b \ot h_{(2)}) \,,\\
            &
            \sigma(a \ot h \cdot g) =
            \sigma(a_{(1)} \ot g) \, \sigma(a_{(2)} \ot h) \,,
        \end{aligned}
    \label{SkewMap_Multiplication-CoProduct}
    \ee
    and
    \be
        \sigma(1 \ot h) = \varepsilon(h) 
        \,\,,\quad 
        \sigma(a \ot 1) = \varepsilon(a) \,,
        \label{SkewMap_CoUnit}
    \ee
    for all $a,b \in A$ and $h,g \in H$. We have used the symbol $1$ as a shorthand to encode the unit in both the bi-algebras $H$ and $A$. If either $A$ has an antipode or $H$ has an inverse-antipode, then there exists a "convolution inverse" $\sigma\mone$ that satisfies \cite{Majid:1996kd}
    \be
        \sigma\mone(a \ot h) \equiv \sigma(S a \ot h) = \sigma(a \ot S\mone h) \,.
    \ee
\end{definition}
\begin{definition}[Generalized quantum double]
\label{Def_GenQuantumDouble}
    \cite{Majid:1994nw}
    Let $H$ and $A$ skew co-paired bi-algebras with skew pairing $\sigma$ which is convolution-invertible. 
    The generalized quantum double is the double cross product bi-algebra $D(H,A,\sigma) \equiv H \bowtie A$ built on $H \ot A$, equipped with mutual actions
    \be
        \begin{aligned}
            &
            a \lhd h = a_{(2)} \, \sigma\mone(a_{(1)} \ot h_{(1)}) \, \sigma(a_{(3)} \ot h_{(2)}) \,,\\
            &
            a \rhd h = h_{(2)} \, \sigma\mone(a_{(1)} \ot h_{(1)}) \, \sigma(a_{(2)} \ot h_{(3)}) \,,
        \end{aligned}
        \label{QuantumDouble_Actions}
    \ee
    that make the bi-algebras $H$ and $A$ a matched pair. 
    The generalized quantum double $D(H,A,\sigma)$ is a bi-algebra with product
    \be
        (h \ot a) \cdot (g \ot b) = 
        h (a_{(1)} \rhd g_{(1)}) \ot (a_{(2)} \lhd g_{(2)}) b = 
        h g_{(2)} \ot a_{(2)} b 
        \,\sigma\mone(a_{(1)} \ot g_{(1)}) \, \sigma(a_{(3)} \ot g_{(3)}) \,,
    \ee
    co-product $\Delta = (\Delta_H \ot \Delta_A)$, tensor product, unit and co-unit.
\end{definition}
\begin{definition}[Matched co-pair bi-algebras]
    \cite{Majid:1994nw}
    Two bi-algebras $H$ and $A$ form a matched co-pair if there exist a pair of co-actions
    \be
        \alpha : H \rightarrow H \ot A
        \,\,,\quad
        \beta : A \rightarrow H \ot A
    \ee
    that satisfy the compatibility relations
    \be
        \begin{aligned}
            &
            (\Delta \ot \id) \circ \alpha(h) 
            = 
            ((\id \ot \beta) \circ \alpha(h_{(1)})) (1 \ot \alpha(h_{(2)})) \,, \\
            &
            (\id \ot \Delta) \circ \beta(a) 
            = 
            ((\beta(a_{(1)}) \ot 1) ((\alpha \ot \id) \circ \beta(a_{(2)}))  \,, \\
            & \qquad \qquad \qquad
            \alpha(h)\beta(a) = \beta(a)\alpha(h) \,.
        \end{aligned}
    \ee
\end{definition}
\begin{definition}[Skew co-paired bi-algebras.]
    \cite{Majid:1994nw}
    Two bi-algebras $H$ and $A$ are skew co-paired if there exists an element called skew co-pairing $\sigma \in H \ot A$ such that\footnote{In  \eqref{SkewEl_id-CoProduct}, \eqref{SkewEl_CoProduct-id}, \eqref{SkewEl_id-CoProduct_inv} and \eqref{SkewEl_CoProduct-id_inv},  we used for clarity the notation $\Delta_H$ and $\Delta_A$ to encode  the co-products on $H$ and $A$. In the following, for simplicity, we will drop the indices $H$ and $A$.}
    \begin{align}
        &
        (\id \ot \Delta_A) \, \sigma = 
        \sigma_{13} \sigma_{12} \,,
        \label{SkewEl_id-CoProduct}
        \\
        &
        (\Delta_H \ot \id) \, \sigma = 
        \sigma_{13} \sigma_{23} \,.
        \label{SkewEl_CoProduct-id}
    \end{align}
    If either $A$ admits an antipode or $H$ admits an inverse-antipode, then the skew co-pairing is invertible, with inverse
    \be
        \sigma\mone \equiv (\id \ot S) \, \sigma = (S\mone \ot \id) \, \sigma \,,
        \label{InvSkewEl}
    \ee
    that satisfies the axioms
    \begin{align}
        &
        (\id \ot \Delta_A) \, \sigma^{-1} = 
        \sigma^{-1}_{12} \sigma^{-1}_{13} \,,
        \label{SkewEl_id-CoProduct_inv}
        \\
        &
        (\Delta_H \ot \id) \, \sigma^{-1} = 
        \sigma^{-1}_{23} \sigma^{-1}_{13} \,.
        \label{SkewEl_CoProduct-id_inv}
    \end{align}
\end{definition}
Using the co-unit axiom \eqref{HopfAlg_Axiom_Co-unit} and the co-pairing properties \eqref{SkewEl_id-CoProduct} and \eqref{SkewEl_CoProduct-id}, one derives the identities 
\be
    (\varepsilon \ot \id) \, \sigma =
    (\id \ot \varepsilon) \, \sigma = 1 \,,
\ee
that we call co-unit properties of the skew co-pairing element.
\begin{definition}[Dual of the generalized quantum double]
\label{Def_DualGenQuantumDouble}
    \cite{Majid:1994nw}
    Let $H$ and $A$ skew co-paired bi-algebras with invertible skew co-pairing. The dual of the generalized quantum double is the double cross co-product bi-algebra $D^*(A,H,\sigma) = A \blackbowtie H$ built on $A \ot H$, equipped with mutual co-actions $\alpha : H \to H \ot A$ and $\beta : A \to H \ot A$ given by
    \be
        \begin{aligned}
            &
            \alpha(h) = \sigma\mone (h \ot 1) \sigma \,, \\
            &
            \beta(a) = \sigma\mone (1 \ot a) \sigma \,,
        \end{aligned}
        \label{QuantumDouble_CoActions}
    \ee
    that make the bi-algebras $H$ and $A$ a matched co-pair. 
    The dual of the generalized quantum double $D^*(A,H,\sigma)$ is a bi-algebra with co-product
    \be
        \Delta (a \ot h) = 
        \big( (\id \ot \alpha \ot \id) \circ (1 \ot \Delta_H h) \big) \cdot
        \big( (\id \ot \beta \ot \id) \circ (\Delta_A a \ot 1) \big) =
        \sigma_{23}\mone (\Delta_A a \ot \Delta_H h) \sigma_{23} \,,
        \label{CoProduct_DualQuantumDouble}
    \ee
    product $(a \ot h) \cdot (b \ot g) = ab \ot hg$, tensor product, unit and co-unit.
\end{definition}
\noindent
We will later need to extend this construction in the context of tensor product of bi-algebras. 
We consider two tensor product bi-algebras $H^n = \bigotimes_i^n H_i$ and $A^n = \bigotimes_i^n A_i$ such that the sub bi-algebras $H_i$ and $A_i$ are skew paired (resp. skew co-paired), that is for each pair $H_i, A_i$ we have a map $\sigma_i : A_i \ot H_i \rightarrow \mathbb{K}$ (resp. an element $\sigma_i \in H_i \ot A_i$). Then the tensor product bi-algebras $H^n,A^n$ are skew paired by the map 
\be
    \Sigma^n : A^n \ot H^n \rightarrow \mathbb{K}
    \label{d_SkewMap}
\ee
or skew co-paired by the element 
\be
    \Sigma^n \in H^n \ot A^n \,.
    \label{d_SkewEl}
\ee
In particular, since the bi-algebras $H^n$ and $A^n$ are tensor products of independent bi-algebras, the $n$ dimensional skew co-pairing can be written in the tensor product notation as 
\be
    \Sigma^n = \sigma_{1 \, n+1} \, \sigma_{2 \, n+2} \, \cdots \, \sigma_{n \, 2n} \,.
    \label{d_SkewEl_TensorProd}
\ee
\begin{proposition}[Properties of the multidimensional skew co-pairing element]
\label{Prop_MultiDimensionalSkewEl}
    Let $H^n$ and $A^n$ be skew co-paired tensor product bi-algebras with invertible skew co-pairing $\Sigma^n$ and convolution inverse $\Sigma^{n \, -1}$, the identities below are satisfied.
    \be
        \begin{aligned}
            &
            (m^n \ot \id^{\ot n}) \, \Sigma^n = (\id \ot \tau^n \circ \Delta^n) \, \sigma
            \,,\\
            &
            (m^n \ot \id^{\ot n}) \, \Sigma^{n \, -1} = (\id \ot \Delta^n) \, \sigma\mone 
            \,,
        \end{aligned}
        \qquad
        \begin{aligned}
            &
            (\id^{\ot n} \ot m^n) \, \Sigma^n = (\Delta^n \ot \id) \, \sigma
            \,,\\
            &
            (\id^{\ot n} \ot m^n) \, \Sigma^{n \, -1} = (\tau^n \circ \Delta^n \ot \id) \, \sigma\mone 
            \,.
        \end{aligned}
        \label{n-dimensional_SkewEl_Property}
    \ee
\end{proposition}
The proof of this proposition follows directly from the definition of $\sigma$. 

\subsection{Integral and Fourier transform}
In this section we recall the definition of the Fourier transform  between Hopf algebras. Hence we work at a purely algebraic level, we are not concerned about analytical issues, which would matter if we want to  analyze the possible divergences of the theory.  \\
First we define the notion of \textit{integral} in the Hopf algebra setting. This notion differs from the usual notion of integral one uses in calculus. This latter is actually what will be called the \textit{co-integral}. 
\begin{definition}[Integral and co-integral]
    \cite{Majid:1996kd}
    A left (resp. right) integral in $A$ is an  element $\ell_L$, (resp. $\ell_R$) in $ A$ such that
    \be
        a \cdot \ell_L = \varepsilon(a) \ell_L
        \,\,,\quad
        \ell_R \cdot a = \varepsilon(a) \ell_R 
        \,,\qquad \forall a \in A \,.
        \label{Integral_LeftRightInvariance}
    \ee
    The integral $\ell$ is normalized if $\varepsilon(\ell) = 1$. \\
    A left (resp. right) co-integral on $A$ is a map $\int_A^L : A \rightarrow \mathbb{K}$ (resp. $\int_A^R : A \rightarrow \mathbb{K}$) that satisfies the left  (resp. right) invariance condition
    \be
        \left( \text{id} \ot \int_A^L \right) \, \Delta a = 1 \ot \int_A^L a \,\,,\qquad
        \left( \int_A^R \ot \text{id} \right) \, \Delta a = \int_A^R a \ot 1 
        \,,\qquad \forall a \in A \,.
        \label{CoIntegral_LeftRightInvariance}
    \ee
    The co-integral $\int_A$ is normalized if $\int_A \, 1 = 1$.
\end{definition}
\noindent 
In order to keep track of divergences, in the following we consider non-normalized co-integrals on $H$ and $A$
\be
    \int_H \, 1 \equiv V_H 
    \,\,,\quad
    \int_A \, 1 \equiv V_A \,.
\ee
Using $\sigma$, we can now define the notion of Fourier transform and its inverse. 

\begin{definition}[Fourier transform]
    Let $H$ and $A$ be skew co-paired bi-algebras. The Fourier transform from $H$ to $A$ is a map $\mathcal{F} : H \rightarrow A$ defined as\footnote{Note that we used left co-integrals for both the Fourier transform \eqref{FourierTransform} and its inverse \eqref{Inverse_FourierTransform}. As an alternative convention, we could have used right co-integrals. In this case the inverse skew co-pairing element $\sigma\mone$ would appear in the definition \eqref{FourierTransform} and the skew co-pairing element $\sigma$ in \eqref{Inverse_FourierTransform}.}
    \be
        \mathcal{F}[h] 
        \equiv \frac{1}{\sqrt{\mu}} \left( \int_{H}^L \ot \, \id \right) \,
        \big( \sigma \cdot (h \ot 1)\big) \,,
        \label{FourierTransform} 
    \ee
    with inverse map $\mathcal{F}\mone : A \rightarrow H$
    \be
        \mathcal{F}^{-1}[a]
        \equiv \frac{1}{\sqrt{\mu}} \left( \id \ot \int_{A}^L \right) \,
        \big( \sigma\mone \cdot (1 \ot a)\big) \,,
        \label{Inverse_FourierTransform} 
    \ee
    with $h \in H$, $a \in A$, and normalization factor
    \be
        \mu 
        \equiv \left( \int_{H}^L \ot \int_{A}^L \right) \, \sigma
        = \left( \int_{H}^L \ot \int_{A}^L \right) \, \sigma\mone \,.
        \label{Normalization_Factor}
    \ee
\end{definition}
\begin{proposition}
\label{Prop_Fourier}
    The Fourier transform \eqref{FourierTransform} and the inverse Fourier transform \eqref{Inverse_FourierTransform} are inverse maps in the sense that
    \be
        (\mathcal{F} \circ \mathcal{F}^{-1}) = (\mathcal{F}^{-1} \circ \mathcal{F})= \id \,.
    \ee
\end{proposition}
The analogue of the Dirac delta function will be a very important object in our construction since it allows to identify the decorations of the geometric structures. 
\begin{definition}[Delta function]
    Let $H$ and $A$ be skew co-paired bi-algebras. The Dirac delta functions of $H$ (noted $\delh$) and  $A$ (noted $\dela$), are respectively defined as the Fourier transform and the inverse Fourier transform of the unit:
    \beq
        \dela \equiv \cF[1]
        =
        \frac{1}{\sqrt{\mu}} \left( \int_{H}^L \ot \, \id \right) \, \sigma
        \,\,,\quad
        \delh \equiv \cF\mone[1] 
        =
        \frac{1}{\sqrt{\mu}} \left( \id \ot \int_{A}^L \right) \, \sigma\mone \,.
        \label{Delta}
    \eeq
    Similarly, we call opposite delta functions the elements
    \be
        \dela^{-1}
        =
        \frac{1}{\sqrt{\mu}} \left( \int_{H}^L \ot \, \id \right) \, \sigma\mone
        \,\,,\quad
        \delh^{-1}
        =
        \frac{1}{\sqrt{\mu}} \left( \id \ot \int_{A}^L \right) \, \sigma \,.
        \label{Delta_Opposite}
    \ee
\end{definition}
\begin{proposition}[Properties of the delta function]
\label{Prop_Delta}
    Let $H$ and $A$ be skew co-paired Hopf algebras.
    The delta functions  satisfy the identities below.
    \be    
        \begin{aligned}
            &
            (a \ot 1) \cdot \Delta \dela = (1 \ot Sa) \cdot \Delta \dela
            \,,\\
            &
            (h \ot 1) \cdot \Delta \delh = (1 \ot S\mone h) \cdot \Delta \delh 
            \,,
        \end{aligned}
        \qquad \quad
        \begin{aligned}
            &
            \Delta \dela\mone \cdot (a \ot 1) = \Delta \dela\mone \cdot (1 \ot Sa)
            \,,\\
            &
            \Delta \delh\mone \cdot (h \ot 1) = \Delta \delh\mone \cdot (1 \ot S\mone h) \,.
        \end{aligned}
        \label{Delta_CoProduct_Property}
    \ee
    \smallskip
    Moreover, the delta functions are normalized in the sense that
    \be
        \frac{1}{\sqrt{\mu}} \int_A^L \, \dela = 
        \frac{1}{\sqrt{\mu}} \int_A^L \, \dela\mone = 1
        \,\,,\quad
        \frac{1}{\sqrt{\mu}} \int_H^L \, \delh = 
        \frac{1}{\sqrt{\mu}} \int_H^L \, \delh\mone = 1 \,.
        \label{Delta_Normalization}
    \ee
\end{proposition}
\begin{proposition}[Delta function as integral in the Hopf algebra]
\label{Prop_DeltaIntegral}
    Consider the skew co-paired Hopf algebras $H,A$.
    The delta functions \eqref{Delta} of the Hopf algebras $H$ and $A$ are resp. left integrals of $A$ and $H$:
    \be
        a \cdot \dela = \dela \, \varepsilon(a)
        \,\,,\quad
        h \cdot \delh = \delh \, \varepsilon(h) \,.
    \label{Delta_RightInvariance}
    \ee
    The opposite delta functions \eqref{Delta_Opposite} are resp. right integrals of $A$ and $H$:
    \be
        \dela\mone \cdot a = \dela\mone \, \varepsilon(a)
        \,\,,\quad
        \delh\mone \cdot h = \delh\mone \, \varepsilon(h) \,.
    \label{OppDelta_LeftInvariance}
    \ee
\end{proposition}
\smallskip
\noindent
For the sake of clarity, here we specified whether the co-integral are left or right by using the indices $L,R$. In the following we will omit such index, implying that all the co-integrals are \textit{left} invariant ones.
\smallskip \\
\textit{Since the skew co-pairing element $\sigma$ is the kernel of the Fourier transform, we will  call it \textbf{plane wave}.}
\smallskip \\
In App. \ref{app:hopfproof} we give some proofs of the propositions and statements given in this section.

\section{Hopf algebra field theory}
\label{Sec_HopfAlgebraFieldTheory}

In this section we describe what we will call a Hopf algebra field theory (HAFT). This is a field theory based on Hopf algebras which can be seen as a generalization of ordinary group field theories \cite{Boulatov:1992vp, Oriti:2017ave}. We provide an action specified by a kinetic and an interaction term, and  we will derive the amplitude of a general cellular decomposition, expanded as a "sum" over the building blocks of a triangulation (or its dual complex).
\smallskip \\
We call three dimensional HAFT the field theory which Feynman diagrams can be seen as dual of a three dimensional triangulation. This will be the focus of our discussion. It can be generalized in a direct manner to any dimensions.

HAFT is meant to be a model about discrete geometries.
Hence, throughout its construction, we will be driven by our geometric intuition. In particular, we consider a three dimensional  triangulation and its dual complex.
We use elements of the Hopf algebra $A$ to decorate objects in the triangulation and elements of $H$ to decorate objects in the dual complex. This is analogue to what happens for the Kitaev model defined in terms of Drinfeld doubles \cite{BMCA, Buerschaper2010}, which is using a quadrangulation that is then self-dual.

We introduce the nomenclature \textit{node} and \textit{links} resp. for 0$d$ and 1$d$ objects living in the dual complex, while \textit{vertices} and \textit{edges} are resp. 0$d$ and 1$d$ objects in the triangulation. In 3d nodes and links are respectively dual to 3- and 2-simplices (tetrahedra and triangles).

\subsection{Field and closure constraint}

We note $\Phi$ the field which is an element in the tensor product of three copies of $H$, $\Phi \in H^3$. On the other hand, $\hat{\Phi}$ is the dual field and is an element in three copies of $A$, $\hat{\Phi} \in A^3$. The fields $\Phi$ and $\hat{\Phi}$ are related by a Fourier transform \eqref{FourierTransform} and its inverse \eqref{Inverse_FourierTransform}, with kernel $\Sigma$. 

Geometrically, the dual field $\hat{\Phi} \in A^3$ is associated to a triangle (2-simplex), where each of its sub-components (elements of $A$) decorates one of the edges that compose its boundary.
Dually, the field $\Phi \in H^3$ is associated to the graph dual to such triangle and its sub-components (elements of $H$) are the links, that share a single node, dual to the  edges. \\
The fields $\Phi$ and $\hat{\Phi}$ are represented in Fig. \ref{Fig_Field-DualField}. In this context, the skew co-pairing $\sigma$ encodes the information on both the triangulation and the dual complex.%: in three dimensions, each element $\sigma$ can be seen as the decoration of an edge and its dual link.
\begin{figure}
     \centering
        \input{Pics/Field} 
        \caption{The dual field $\hat{\Phi}$ is represented in black as a triangle; the field $\Phi$ is its dual graph represented in blue, made of a central node and three links.}
        \label{Fig_Field-DualField}
\end{figure}
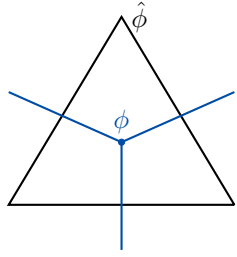
We impose that the field $\Phi$ should be invariant under some gauge symmetry, which we enforce through a projector.
\begin{definition}[Gauge projector]
    Consider the (left\footnote{We may alternatively introduce a right projector as
    \be
        (\cP_R \, \Phi) = \frac{1}{V_H} \left(\idd \ot \int_H\right) \,
        \big((\idd \ot m^3) \circ \Delta^{(3)} \Phi\big) \,,
    \ee
    defined through a right co-integral.}) projector $\cP_L : H^3 \to H^3$ whose action on the field is called gauge averaging
    \be
        (\cP_L \, \Phi) = \frac{1}{V_H} \left(\int_H \ot \idd\right) \,
        \big((m^3 \ot \idd) \circ \Delta^{(3)} \Phi\big) \,.
        \label{GaugeProj_Field}
    \ee
    We recall that $\int_H$ is the left co-integral on $H$, while $m^3$ and $\Delta^{(3)}$ are the maps resp. defined in \eqref{1-d_Multiplication} and \eqref{d-d_CoProduct}.
\end{definition}
In App. \ref{App_GaugeProjector} we prove that the operator $\cP_L$ defined above is a projector. 
The field $\Phi$ we defined is invariant under gauge averaging:
\be
    (\cP_L \, \Phi) = \Phi \,.
    \label{GaugeInvariance}
\ee
The element $(m^3 \ot \idd) \circ \Delta^{(3)} \, \Phi$ in \eqref{GaugeProj_Field} belongs to the Hopf algebra $H^{4}$: this can be understood as the graph dual to the triangle (represented by the field) plus an extra link. Such extra link is interpreted as a parallel transport in the dual complex, and \eqref{GaugeInvariance} enforces the invariance of the field under any possible translation of this type. In analogy with ordinary group field theory, we call it \textit{gauge symmetry} and  \eqref{GaugeInvariance} encodes the \textit{gauge invariance} of the field. We represent the gauge symmetry in Fig. \ref{Fig_GaugeSymmetry}.
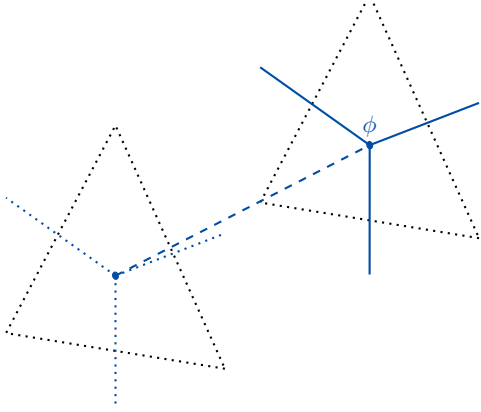
\begin{figure}
    \centering
    \input{Pics/GuageSymmetry}
    \caption{The gauge symmetry is interpreted as a translation in the dual complex. As the field $\Phi$ is given by three links dual to a triangle (in blue), the (gauge) translation is an extra variable (the dashed blue line).}
    \label{Fig_GaugeSymmetry}
\end{figure}
Dually (upon Fourier transform), the gauge symmetry translates into the closure constraint. 
\begin{proposition}[Closure constraint]
    The Fourier transform of the gauged projected field \eqref{GaugeProj_Field} gives
    \be
        \cF[(\cP_L \, \Phi)] = \hat{\cC} \cdot \hat{\Phi} \,,
        \label{ClosureConstraint_Transform}
    \ee
    where the element $\hat{\cC} \in A^3$ is  
    \be
        \hat{\cC} = \frac{\sqrt{\mu}}{V_H} \, \Delta^3 \dela\mone \,.
        \label{ClosureConstraint}
    \ee
\end{proposition}
\begin{proof}
    We show here that the Fourier transform of the projected field \eqref{GaugeProj_Field} gives the closure constraint \eqref{ClosureConstraint}.
    The Fourier transform of the gauge projected field \eqref{GaugeProj_Field} is
    \begin{align}
        \cF[(\cP_L \, \Phi)]
        & =
        \frac{1}{\sqrt{\mu^3}} \left(\int_{H^3} \ot \idd\right) \,
        \big(\Sigma \cdot ((\cP_L \, \Phi) \ot 1)
        \nonumber \\
        & =
        \frac{1}{\sqrt{\mu^3} \, V_H} \left(\int_H \ot \int_{H^3} \ot \idd\right) \,
        \big((m^3 \ot \idd \ot \idd) \, \Sigma_{23} \cdot (\Delta^{(3)}\Phi \ot 1)\big)
        \nonumber \\
        & =
        \frac{1}{\sqrt{\mu^3} \, V_H} \left(\int_H \ot \int_{H^3} \ot \idd\right) \,
        (m^3 \ot \idd \ot \idd) \,
        \big(\Sigma_{13}\mone \cdot \Sigma_{13} \cdot \Sigma_{23} \cdot (\Delta^{(3)}\Phi \ot 1)\big) \,.
    \end{align}
    Then, using the property \eqref{SkewEl_CoProduct-id} for the multi dimensional skew co-pairing element\footnote{We will often simply use $\Sigma$ as a short hand notion for $\Sigma^n$ in order to not clutter the notations.} $\Sigma$ and left invariance of the co-integral on $H$, the expression reduces to
    \begin{align}
        \cF[(\cP_L \, \Phi)]
        & =
        \frac{1}{\sqrt{\mu^3} \, V_H} \left(\int_H \ot \int_{H^3} \ot \idd\right) \,
        (m^3 \ot \idd \ot \idd) \,
        \big(\Sigma_{13}\mone \cdot (\Delta^{(3)} \ot \idd) \,(\Sigma \cdot (\Phi \ot 1)\big)
        \nonumber \\
        & =
        \frac{1}{\sqrt{\mu^3} \, V_H} \left(\int_H \ot \int_{H^3} \ot \idd\right) \,
        (m^3 \ot \idd \ot \idd) \,
        \big(\Sigma_{13}\mone \cdot \Sigma_{23} \cdot (1 \ot \Phi \ot 1)\big) \,.
    \end{align}
    Last, using the property \eqref{n-dimensional_SkewEl_Property} for the three dimensional convolution inverse $\Sigma\mone$, together with the definition of Fourier transform and opposite delta function, we get the expression of the closure constraint \eqref{ClosureConstraint}.
    \begin{align}
        \cF[(\cP_L \, \Phi)]
        & =
        \frac{1}{\sqrt{\mu^3} \, V_H} \left(\int_H \ot \int_{H^3} \ot \idd\right) \,
        \big((\id \ot \idd \ot \Delta^3) \, \sigma\mone \cdot \Sigma_{23} \cdot (1 \ot \Phi \ot 1)\big)
        \nonumber \\
        & =
        \frac{1}{V_H} \left(\int_H \ot \idd\right) \,
        \big((\id \ot \Delta^3) \, \sigma\mone \cdot (1 \ot \hat{\Phi})\big)
        =
        \frac{\sqrt{\mu}}{V_H} \, \Delta^3 \dela\mone \cdot \hat{\Phi} \,.
    \end{align}
\end{proof}
\noindent
Performing a Fourier transformation of the gauge invariance condition \eqref{GaugeInvariance} gives the identity
\be
    (\hat{\cC} \cdot \hat{\Phi}) = \hat{\Phi} \,.
\ee
$\hat{\cC}$ is called the \textit{closure constraint} as it encodes the closure of the triangle \cite{Baratin:2010nn}. Indeed, as an element of the tensor product algebra $A^3$, the closure constraint is interpreted as the combination of three edges in the triangulation. Being given as the co-product of the delta function, these edges are naturally interpreted as part of a discrete closed path. 
Therefore, \eqref{ClosureConstraint_Transform} implements the closure of the boundary (the three sub-components) of the (triangle) dual field $\hat{\Phi}$.

Note that we could also demand to implement a symmetry under permutations of the arguments $H$ in $H^3$. We will not impose this and follow instead \cite{Freidel:2009hd} where it was shown that in 3d, if the proper interaction term is chosen, there is no need to use such permutation symmetry.

\subsection{Action}

The HAFT action is composed by a kinetic plus an interaction term (with coupling constant set to $1$ for simplicity) 
\be
    \cS = \cS_{\cK} + \cS_{\cV} \,.
\ee
The interaction term is defined as the product of fields with the combinatorics of the tetrahedron. 
\begin{definition}[Interaction term]
    The interaction term of three dimensional HAFT is
    \be
        \cS_{\cV} = \int_{H^{6}} \, \big(\um \, (\cP_L \, \Phi \ot \cP_L \, \Phi \ot \cP_L \, \Phi \ot \cP_L \, \Phi)\big) \,.
        \label{InteractionTerm}
    \ee
    The map $\um$ was introduced in Def. \ref{Def_Underline_Maps} and it is explicitly given in \eqref{Underline_Multiplication_3d}.
\end{definition}
Going to the dual picture, namely the triangulation picture, allows a better interpretation of such interaction term. 
\begin{proposition}[Tetrahedron]
\label{Prop_Vertex}
    The interaction term \eqref{InteractionTerm} can be expressed as an integral operator in the two forms below
    \be
        \cS_{\cV} =
        \int_{H^{12}} \, \big(\cV \cdot (\Phi \ot \Phi \ot \Phi \ot \Phi)\big) =
        \int_{A^{12}} \, \big(\hat{\cV} \cdot (\hat{\Phi} \ot \hat{\Phi} \ot \hat{\Phi} \ot \hat{\Phi})\big) \,,
        \label{InteractionTerm_Amplitude}
    \ee
    where
    \be
        \begin{aligned}
            \cV 
            & =
            \frac{1}{\mu^{6} \, V_H^{4}} 
            \left(\int_{H^{4}} \ot \id^{12}\right) \, \sDe \delh\mone
            \,,\\
            \hat{\cV} 
            & =
            \frac{1}{\mu^{6}} \, \uDe \dela\mone \, \cdot \, (\hat{\cC} \ot \hat{\cC} \ot \hat{\cC} \ot \hat{\cC}) \,,
        \end{aligned}
        \label{VertexAmplitude}
    \ee
    are the kernels of the \textit{tetrahedron amplitudes}, resp. in the $H$ and $A$ representations. 
    The co-product $\sDe$ was introduced in Def. \ref{Def_Underline_Maps}, moreover the opposite delta functions $\delh\mone$ and $\dela\mone$ in both the amplitudes \eqref{VertexAmplitude} are six dimensional delta functions.
\end{proposition}
\noindent 
Due to the length and complexity of the proof of Prop. \ref{Prop_Vertex}, we provide it in App. \ref{App_TetrahedronAmplitude}. \\
Each dual field is associated to a triangle, by construction, the interaction term  represents the composition of four triangles, with the combinatorics of a tetrahedron specified by the co-product \eqref{Underline_CoProduct_Simplex_3d}.
The associated amplitude is thus interpreted as the boundary of a tetrahedron or the bulk and the boundary of its dual graph and it is illustrated in Fig. \ref{Fig_VertexAmplitude}. In analogy with ordinary field theory, the tetrahedron amplitude encodes the smallest information (we call it \textit{building block}) of any Feynman diagram of Hopf algebra field theory. As a consequence, the HAFT Feynman diagrams are (dual to) three dimensional triangulations.
\begin{figure}
     \centering
     \begin{subfigure}[b]{0.4\textwidth}
        \centering
        \input{Pics/Propagator}
    \caption{The propagator amplitude enforces the identification of a pair of triangles (in black) in the $A$ representation or their dual graphs (in blue) in the $H$ representation. We used arrows to emphasize that edges are identified with opposite orientations and links with the same orientation. This is reflected by the antipode in \eqref{PropagatorAmplitude}.}
    \label{Fig_PropagatorAmplitude}
    \end{subfigure}
    \hfill
    \begin{subfigure}[b]{0.55\textwidth}
        \centering
        \input{Pics/Tetrahedron}
        \caption{The tetrahedron amplitude is given by the combination of four closure constraints (in black) representing the triangles that compose the boundary of a tetrahedron.
        The graph dual to the bulk of the tetrahedron is in red and the graph dual to its boundary is in blue. The graph dual to a tetrahedron is given as a combination of six loops made of two bulk links (solid red) and two boundary links (solid blue). Each of these loops is given by one of the six co-products in $\sDe$.}
        \label{Fig_VertexAmplitude}
     \end{subfigure}
     \caption{Representation of the propagator and vertex/tetrahedron amplitude.}
\end{figure}
\smallskip \\
\begin{definition}[Kinetic term]
    The \textit{kinetic term} of three dimensional Hopf algebra field theory is
    \be
        \cS_{\cK} = \int_{H^3} \, \Phi \cdot (\tau \Phi) \,.
        \label{KineticTerm}
    \ee
\end{definition}
The associated amplitude, called propagator, thus represents the identification of the dual graph of two triangles.
As in ordinary in ordinary field theories, the propagator amplitudes are used to glue pairs of amplitudes associated to the interaction term. In this case, two tetrahedra are merged by identifying two triangles (or the associated boundary graph). We represent such identification in Fig. \ref{Fig_PropagatorAmplitude}. Such composition allows to construct arbitrary simplicial three dimensional triangulations or dual complexes (with or without boundaries). \\
In the triangulation picture, the propagator consists in identifying the edges of the triangles.
\begin{proposition}[Propagator]
\label{Prop_Propagator}
    The kinetic term \eqref{KineticTerm} can be expressed as an integral operator in two different ways, 
    \be
        \cS_{\cK} =
        \left(\int_{H^3} \ot \int_{H^3}\right) \, \big(\cK \cdot (\Phi \ot \Phi)\big) =
        \left(\int_{A^3} \ot \int_{A^3}\right) \, \big(\hat{\cK} \cdot (\hat{\Phi} \ot \hat{\Phi})\big) \,,
        \label{KineticTerm_Amplitude}
    \ee
    where
    \be
        \cK = 
        \frac{1}{\sqrt{\mu^3}} \, (S^{(3)} \ot \tau^3) \circ \Delta^{(3)} \delh 
        \,\,,\quad
        \hat{\cK} = 
        \frac{1}{\sqrt{\mu^3}} \, (\idd \ot \tau^3) \circ \Delta^{(3)} \dela\mone \,,
        \label{PropagatorAmplitude}
    \ee
    are the kernels of the \textit{propagator amplitudes}, resp. in the $H$ and $A$ representations.
\end{proposition}
\begin{proof}
    We show that the kinetic term \eqref{KineticTerm} can be expressed as in \eqref{KineticTerm_Amplitude} with kernels \eqref{PropagatorAmplitude}.
    Let us first focus on the kernel in the $A$ representation. We  use the definition of the inverse Fourier transform, the property \eqref{SkewEl_id-CoProduct_inv} and the definition of opposite delta function.
    \begin{align}
        \cS_{\cK} & =
        \int_{H^3} \, \Phi \cdot (\tau^3 \Phi)
        =
        \frac{1}{\mu^3} \left(\int_{H^3} \ot \int_{A^3} \ot \int_{A^3}\right) \,
        \big(\Sigma_{12}\mone \cdot (\id \ot \tau^3)\Sigma_{13}\mone \cdot (1 \ot \hat{\Phi} \ot \hat{\Phi})\big)
        \nonumber \\
        & =
        \frac{1}{\mu^3} \left(\int_{H^3} \ot \int_{A^3} \ot \int_{A^3}\right) \,
        \big((\id \ot \id \ot \tau^3) \circ (\id \ot \Delta^{(3)}) \Sigma\mone \cdot (1 \ot \hat{\Phi} \ot \hat{\Phi})\big)
        \nonumber \\
        & =
        \frac{1}{\sqrt{\mu^3}} \left(\int_{A^3} \ot \int_{A^3}\right) \,
        \big((\id \ot \tau^3) \circ \Delta^{(3)} \dela\mone \cdot (\hat{\Phi} \ot \hat{\Phi})\big) \,.
    \end{align}
    From this, we get the expression of $\hat{\cK}$ in \eqref{PropagatorAmplitude}. Now, to derive the kernel $\cK$ in the $H$ representation, use the definition of the Fourier transform \eqref{FourierTransform} for both the fields $\hat{\Phi}$, and the property \eqref{Delta_CoProduct_Property} of the delta function.
    \begin{align}
        \cS_{\cK} 
        & =
        \frac{1}{\sqrt{\mu^{6}}} \left(\int_{H^{6}} \ot \int_{A^{6}}\right) \,
        \big((1 \ot 1 \ot (\id \ot \tau^3) \circ \Delta^{(3)} \dela\mone) \cdot \Sigma_{13} \cdot \Sigma_{24} \cdot (\Phi \ot \Phi \ot 1 \ot 1)\big)
        \nonumber \\
        & =
        \frac{1}{\sqrt{\mu^{6}}} \left(\int_{H^{6}} \ot \int_{A^{6}}\right) \,
        (\id \ot \id \ot \id \ot \tau^3)
        \big((1 \ot 1 \ot \Delta^{(3)} \dela\mone) \cdot \Sigma_{13} \cdot (\id \ot \tau^3 \ot \id \ot \id) \, \Sigma_{24} \cdot (\Phi \ot \Phi \ot 1 \ot 1)\big)
        \nonumber \\
        & =
        \frac{1}{\sqrt{\mu^{6}}} \left(\int_{H^{6}} \ot \int_{A^{6}}\right) \,
        \big((1 \ot 1 \ot (\id \ot \tau^3) \circ \Delta^{(3)} \dela\mone) \cdot (\id \ot \tau^3 \ot \id \ot \id) \, \Sigma_{23}\mone \cdot \Sigma_{13} \cdot (\Phi \ot \Phi \ot 1 \ot 1)\big) \,.
    \end{align}
    We get rid of the delta function by using the left invariance of the co-integral on $A$, plus the normalization property \eqref{Delta_Normalization}.
    \begin{align}
        \cS_{\cK}
        & =
        \frac{1}{\sqrt{\mu^{6}}} \left(\int_{H^{6}} \ot \int_{A^{6}}\right) \,
        \big((1 \ot 1 \ot 1 \ot \tau^3 \dela\mone) \cdot (\id \ot \tau^3 \ot \id \ot \id) \Sigma_{23}\mone \cdot \Sigma_{13} \cdot (\Phi \ot \Phi \ot 1 \ot 1)\big)
        \nonumber \\
        & =
        \frac{1}{\mu^{3}} \left(\int_{H^{6}} \ot \int_{A^3}\right) \,
        \big((\id \ot \tau^3 \ot \id \ot \id) \, \Sigma_{23}\mone \cdot \Sigma_{13} \cdot (\Phi \ot \Phi \ot 1)\big) \,.
    \end{align}
    Last, using the definition of convolution inverse, its property \eqref{SkewEl_id-CoProduct} and the definition of delta function we derive the expression of the kernel $\cK$ in \eqref{PropagatorAmplitude}.
    \begin{align}
        \cS_{\cK}
        & =
        \frac{1}{\mu^{3}} \left(\int_{H^3} \ot \int_{H^3} \ot \int_{A^3}\right) \,
        \big((S^{(3)} \ot \tau^3 \ot \id) \circ (\Delta^{(3)} \ot \id) \, \Sigma\mone \cdot (\Phi \ot \Phi \ot 1)\big)
        \nonumber \\
        & =
        \frac{1}{\sqrt{\mu^{3}}} \left(\int_{H^3} \ot \int_{H^3}\right) \,
        \big((S^{(3)} \ot \tau^3) \circ \Delta^{(3)} \delh \cdot (\Phi \ot \Phi)\big) \,.
    \end{align}
\end{proof}

\subsection{Feynman diagram amplitude}

The HAFT Feynman diagrams represent decorated three dimensional  triangulations or their dual complexes. 
The strength of GFT's or HAFT's is that their Feynman diagrams can be seen as the triangulations (or their dual complexes) of a manifold $\cM$. The interaction term is chosen to encode the $d$-simplex if the manifold $\cM$ has dimension $d$. As argued before we focused here on the case $d=3$. 
The Feynman diagrams amplitudes of the HAFT in the $A$ polarization defined earlier can be expanded in terms of diagrams $\Gamma$ which correspond to the triangulations of manifolds. In particular, the amplitude of a triangulation is expressed as a combination of tetrahedron amplitudes \eqref{VertexAmplitude} glued by the proper propagator amplitudes \eqref{PropagatorAmplitude}.

Dually, in the $H$ polarization, the Feynman diagrams $\Gamma^*$ correspond to the dual complexes of the triangulation $\Gamma$. In particular, the amplitude of a general dual complex can be expressed as a combination of loops made of an arbitrary number $N$ of bulk links (part of the bulk graph of a tetrahedron). Each of such loops is the closed path spanning a face dual to an edge (shared by a number $N$ of tetrahedra) of the triangulation. \\
In line with the usual interpretation of curvature in models of discrete geometries, each of these loops can be interpreted as probing the local curvature around this edge. 

We intend now to give the expression of the Feynman diagram amplitude in the different polarizations. We will show that it can actually be expressed in terms of the plane-wave $\sigma$ (which we recall is the canonical element of the dual of the generalized quantum double $D^*(H,A,\sigma)$). As we will show in Sec. \ref{Sec_Examples}, this expression is associated to the discretization of a $BF$ theory.
\begin{proposition}[Amplitude in the triangulation]
\label{Prop_GeneratingFunction}
    Let $\Gamma$ be the triangulation of a manifold  built as the combination of $M$ tetrahedra $\tau$. As a Feynman diagram of the HAFT, it is associated to the amplitude 
    \be
        \cA_{\Gamma} = 
        \frac{1}{\mu^{6 M}} \, 
        \int_{A^{12 M}} \, \prod_{\tau} \,
        \big(\uDe^{(3)} \dela\mone \, \cdot \, (\hat{\cC} \ot \hat{\cC} \ot \hat{\cC} \ot \hat{\cC})\big) \,. 
        \label{GeneratingFunction_Triangulation}
    \ee
\end{proposition}
\begin{proposition}[Amplitude in the dual complex]
    Consider $N$ tetrahedra sharing a single edge and let $\cL_N$ be the closed loop made of $2N$ half bulk links. The closed loop spans a face dual to the edge.
    Let $\Gamma^*$ be the complex dual to a three dimensional triangulation $\Gamma$, built as the combination of such loops $\cL_N$. Each loop is   made of  some number $N$ links (in general different for each loop).
    The HAFT amplitude for the graph $\Gamma^*$ in the $H$ polarization is
    \be
        \cA_{\Gamma^*} = 
        \prod_{\{\cL_N\}} \, \frac{1}{\mu^N \, V_H^{N}} \, \int_{H^{2 N}} \,
        \Delta^{2N}_{1 \, -2 \, 3 \, -4 \, \dots \, (2N-1) \, -2N} \, \delh\mone \,.
        \label{GeneratingFunction_DualComplex}
    \ee
\end{proposition}
\begin{proposition}[Amplitude in terms of the plane-wave]
The amplitude of the Feynman diagram $\Gamma^*$  can be represented in terms of the plane wave $\sigma$, as
    \be
        \cA_{\Gamma^*} = 
        \prod_{\{\cL_N\}} \, \frac{1}{\sqrt{\mu^{2N+1}} \, V_H^{N}} \, \left(\int_{H^{2 N}} \ot \int_A\right) \,
        \big(\Delta^{2N}_{1 \, -2 \, 3 \, -4 \, \dots \, (2N-1) \, -2N} \ot \id\big) \,\sigma \,.
        \label{GeneratingFunction_Double}
    \ee
\end{proposition}
\noindent
In the appendix \ref{App_PartitionFunction}, we show how to derive these three amplitudes.

\subsection{Topological invariance}
Given two simplicial decompositions of a manifold $\cM$, one can find a finite set of transformations, called Pachner moves, that maps one simplicial decomposition into the other \cite{Pachner1991Pachner}. If the associated amplitudes are invariant under the action of these Pachner moves (up to some possible constant re-scaling), then we say that the amplitude is a \textit{topological invariant}.
\begin{proposition}
\label{Prop_Pachner}
    The amplitude \eqref{GeneratingFunction_Double} is a topological invariant.
\end{proposition}
\noindent
In three dimensions there exist two Pachner moves, denoted $P_{(1,4)}$ and $P_{(2,3)}$. We provide below the precise relation between the different relevant amplitudes, showing the topological invariance of the model.
In App. \ref{App_TopologicalInvariance} we provide the details of the proofs.  
\paragraph{Pachner move $P_{(1,4)}$.}    
The Pachner move $P_{(1,4)}$ takes the amplitude of one tetrahedron (denoted $\cA_{\cV}$) into the amplitude of four tetrahedra (denoted $\cA_{\cV^4}$). The (reduced) amplitudes of one and four tetrahedra are
\be
    \begin{aligned}
        \cA_{\cV} & =
        \frac{1}{\mu^3} \, \int_{A^{6}} \,
        \big(\hat{\cC}_{-6 \, -4 \, -2} \cdot 
        \hat{\cC}_{-5 \, 6 \, -1} \cdot 
        \hat{\cC}_{-3 \, 4 \, 5} \cdot \hat{\cC}_{1 \, 2 \, 3}\big)
        \,,\\
        \cA_{\cV^4} & =
        \frac{1}{\mu^{12}} \, \int_{A^{6}} \,
        \big(
        \hat{\cC}_{-1 \, -3 \, -2} \cdot
        \hat{\cC}_{-5 \, 1 \, -4} \cdot 
        \hat{\cC}_{-6 \, 4 \, 2} \cdot 
        \hat{\cC}_{3 \, 5 \, 6}\big) \,.
    \end{aligned}
\ee
The action of the $P_{(1,4)}$ move gives
\be
    (P_{(1,4)} \, \cA_{\cV}) 
    = 
    \cA_{\cV^4} =
    \frac{1}{\mu^9} \, \cA_{\cV} \,.
    \label{Pachner14}
\ee

\paragraph{Pachner move $P_{(2,3)}$.}    
The Pachner move $P_{(2,3)}$ takes the amplitude of two tetrahedra (denoted $\cA_{\cV^2}$) into the one of three (denoted $\cA_{\cV^3}$). The (reduced) amplitudes of two and three tetrahedra are
\be
    \begin{aligned}
        \cA_{\cV^2} & =
        \frac{1}{\mu^{6}} \, \int_{A^{9}} \,
        \big(\hat{\cC}_{-6 \, -4 \, -2} \cdot 
        \hat{\cC}_{-5 \, 6 \, -1} \cdot 
        \hat{\cC}_{-3 \, 4 \, 5} \cdot
        \hat{\cC}_{-9 \, 2 \, -7} \cdot 
        \hat{\cC}_{-8 \, 1 \, 9} \cdot 
        \hat{\cC}_{3 \, 8 \, 7}\big)
        \,,\\
        \cA_{\cV^3} & =
        \frac{1}{\mu^{9}} \, \int_{A^{9}} \,
        \big(\hat{\cC}_{-5 \, -3 \, -2} \cdot 
        \hat{\cC}_{-4 \, 5 \, -1} \cdot 
        \hat{\cC}_{-8 \, 2 \, -6} \cdot 
        \hat{\cC}_{-7 \, 1 \, 8} \cdot 
        \hat{\cC}_{-9 \, 6 \, 3} \cdot 
        \hat{\cC}_{4 \, 7 \, 9}\big) \,.
    \end{aligned}
\ee
The relation between the two amplitudes, encoded by the Pachner move $P_{(2,3)}$ is 
\be
    (P_{(2,3)} \, \cA_{\cV^2}) 
    = 
    \cA_{\cV^3} =
    \frac{1}{\mu^3} \, \cA_{\cV^2} \,.
    \label{Pachner23}
\ee
\medskip \\
In Fig. \ref{Fig_Pachner(1,4)} and \ref{Fig_Pachner(2,3)} we give the representation of the Pachner moves $P_{(1,4)}$ and $P_{(2,3)}$ in the triangulation picture.
The amplitudes associated to the building blocks of any Feynman diagram of the HAFT are invariant (up to some constant factors) under the action of the Pachner moves. Hence, the HAFT we proposed is a topological  model.
\begin{figure}
    \centering
    \input{Pics/Pachner14}
    \caption{The Pachner move $P_{(1,4)}$ takes one tetrahedron to the combination of four. In the triangulation it is realized by connecting the center of the tetrahedron with its four vertices. The surfaces ranging between the center of the tetrahedron and any pair of vertices are internal triangles. The four external faces of the initial tetrahedron thus become the external faces of the four different tetrahedra that share the six internal triangles. 
    In the dual complex, the move is realized by taking four nodes (the colored ones on the right) at the place of the central node (in blue on the left). Each node is connected to one of the four external links and to each of the other nodes.}
    \label{Fig_Pachner(1,4)}
\end{figure}
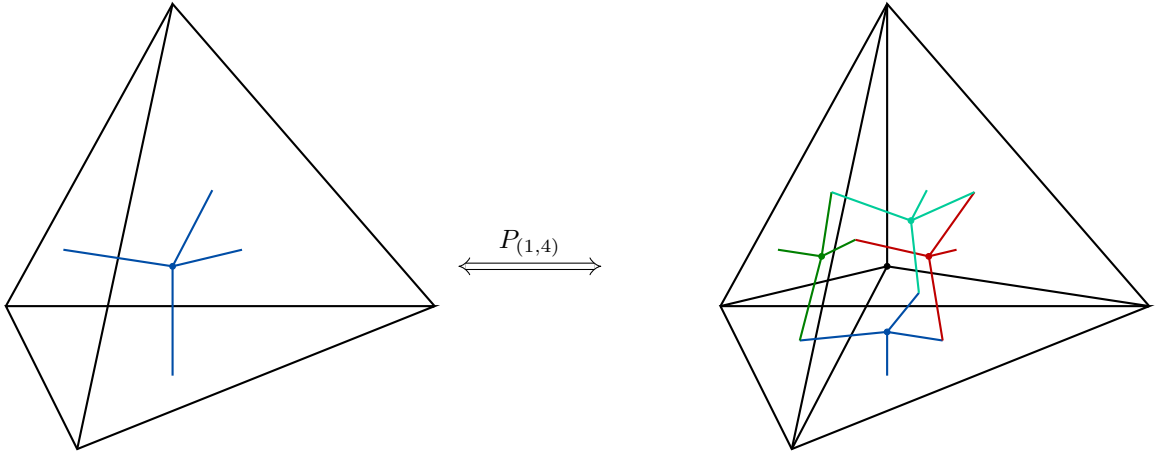
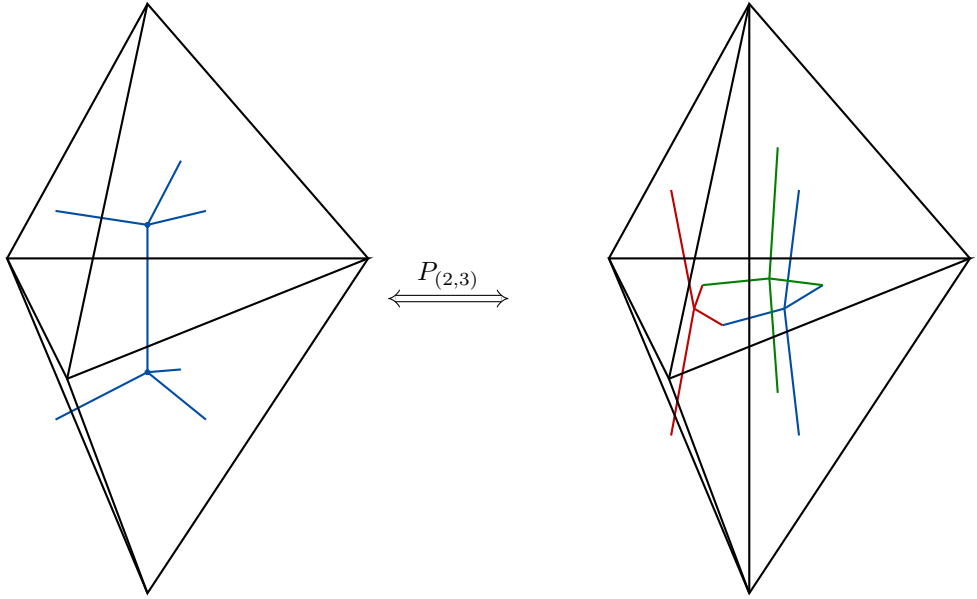
\begin{figure}
    \centering
    \input{Pics/Pachner23}
    \caption{The Pachner move $P_{(2,3)}$ takes two tetrahedra to the combination of three. In the triangulation it is realized by connecting the vertices above and below through an internal edge. The surfaces ranging between the internal edge and one of the three vertices shared by the two initial tetrahedra become internal triangles. The three faces above, that initially belonged to a single tetrahedron, now belong to three different tetrahedra. The same for the faces below.
    In the dual complex, the move is realized by taking three nodes (the colored ones on the right) at the place of the two initial ones (in blue on the left). Each node is thus connected to each other by two of its links, plus to one external link above and one below.}
    \label{Fig_Pachner(2,3)}
\end{figure}
\medskip \\
Note that, in the standard analysis of the three dimensional Pachner moves, the proportionality constant for the move $P_{(1,4)}$ amounts to the cube of the volume of the group (or Hopf algebra in our case) that decorates the dual complex, whereas, the proportionality constant for the Pachner move $P_{(2,3)}$ is a single volume term. 
Geometrically, these terms amount to the difference between the number of (independent) internal edges from the initial and to the final amplitudes of a given move. In \eqref{Pachner14} and \eqref{Pachner23} we do not encounter the (usually) expected proportionality constants. One can retrieve them by setting the normalization factor \eqref{Normalization_Factor} to the identity, $\mu = 1$ (as in the standard case), and by removing the volume term $V_H$ in the definition of gauge projector \eqref{GaugeProj_Field}. This would prevent the operator $\cP_L$ in \eqref{GaugeProj_Field} to be a projector, as it would satisfy the identity $\cP_L^2 = V_H \, \cP_L$ instead of $\cP_L^2 = \cP_L$, but this choice would  lead to the standard proportionality constants in the moves \eqref{Pachner14} and \eqref{Pachner23},
\be
    (P_{(1,4)} \, \cA_{\cV}) 
    : = 
    \cA_{\cV^4} =
    V_H^3 \, \cA_{\cV}
    \,\,,\quad
    (P_{(2,3)} \, \cA_{\cV^2}) 
    : = 
    \cA_{\cV^3} =
    V_H \, \cA_{\cV^2} \,.
\ee

\section{Example: $q$-deformed group field theory}
\label{Sec_Examples}

We use our formalism to construct the HAFT when dealing with the quantum group which would generate the Turaev-Viro amplitude \cite{Turaev:1992hq}, which is  Boulatov's original group field theory \cite{Boulatov:1992vp}. We will  show how the plane-wave $\sigma$ allows to recover the discretization of (Euclidian) BF theory in the presence of a (negative) cosmological constant. 
\smallskip \\
A quantum group is typically seen as a deformation of the usual notion of group. Therefore the HAFT we will construct can be seen as a deformation of the usual notion of group field theory. In appendix \ref{App_Examples}, we discuss the cases of a finite group and the undeformed $\SU(2)$ case.

\subsection{$q$-deformed group field theory}

We would like to illustrate our construction with the Hopf algebras $H$ and $A$ taken to be the $q$-deformations of $F(\SU(2))$ and $\cU(\su(2))$ respectively. 
Such standard deformations \cite{Majid:1996kd} are the non commutative and non co-commutative Hopf algebras denoted $\mathcal{U}_q(\an(2)) \cong F(\SU_q(2))$ and $\mathcal{U}_q(\mathfrak{su}(2)) \cong F(\AN_q(2))$.

\subsubsection{Building blocks}

Let the \textit{real} deformation parameter be $q = e^{\ell \lambda}$, where $\ell$ has the dimension of a length and is the characteristic scale of $\SU(2)$, while $\lambda$ has the dimension of an inverse length and is the characteristic scale of $\AN(2)$. 
Let $H,X_{\pm}$ be the generators of $\mathcal{U}_q(\mathfrak{su}(2))$ with the dimension of a length ($\ell$) and obeying Hopf algebra structure
\be
    \begin{aligned}
        \text{Product:} 
        & \quad
        [H , X_{\pm}] = {\pm} 2 \ell X_{\pm} 
        \qquad \qquad \to \quad 
        e^{\lambda H} X_{\pm} =  q^{\pm 2} \, X_{\pm} e^{\lambda H} \,,\\
        & \quad
        [X_+ , X_-] = \ell^2 q\mone \frac{\sinh(\lambda H)}{\sinh(\ell \lambda)} \,,\\
        \text{Co-product:} 
        & \quad
        \Delta H = H \ot 1 + 1 \ot H \\
        & \quad
        \Delta X_+ = X_+ \ot 1 + e^{-\lambda H} \ot X_+ \\
        & \quad
        \Delta X_- = X_- \ot e^{\lambda H} + 1 \ot X_- \,,\\
        \text{Co-unit:}
        & \quad 
        \varepsilon(H) = \varepsilon(X_{\pm}) = 0 \,,\\
        \text{Antipode:}
        & \quad 
        S(H) = - H \,,\\
        & \quad
        S(X_+) = - e^{\lambda H} X_+ \,,\\
        & \quad
        S(X_-) = -  X_- e^{-\lambda H} \,.
    \end{aligned}
    \label{Bi-Algebra_Uq(su)}
\ee
Let $\phi,\varphi_{\pm}$ be the generators of $F(SU_q(2))$ with the dimension of an inverse length ($\lambda$) and with Hopf algebra structure specified by
\be
    \begin{aligned}
        \text{Product:} 
        & \quad
        [\phi , \varphi_{\pm}] = -i\lambda \varphi_{\pm}
        \qquad \qquad \to \quad 
        e^{i\ell \phi} \varphi_{\pm} =  q \, \varphi_{\pm} e^{i\ell \phi} \,,\\
        & \quad
        [\varphi_+ , \varphi_-] = 0 \,,\\
        \text{Co-product:} 
        & \quad
        \Delta \phi = \frac{i}{\ell} 
        \log\bigg(
        \frac{1}{\Delta \varphi_0} 
        \big( \varphi_0 e^{-i\ell \phi} \ot \varphi_0 e^{-i\ell \phi} - \ell^2 \varphi_- \ot \varphi_+ \big)\bigg) \,,\\
        & \quad
        \Delta \varphi_+ = \varphi_+ \ot \varphi_0 e^{-i\ell \phi} + e^{i\ell \phi} \varphi_0 \ot \varphi_+ \,,\\
        & \quad
        \Delta \varphi_- = \varphi_- \ot e^{i\ell \phi} \varphi_0 + \varphi_0 e^{-i\ell \phi} \ot \varphi_- \,,\\
        \text{Co-unit:}
        & \quad 
        \varepsilon(\phi) = \varepsilon(\varphi_{\pm}) = 0 \,,\\
        \text{Antipode:}
        & \quad
        S(\phi) = -\phi \,,\\
        & \quad
        S(\varphi_{\pm}) = - q^{\mp} \varphi_{\pm} \,.
    \end{aligned}
    \label{Bi-Algebra_F(SUq)}
\ee
with $\varphi_0 = \sqrt{1 - q\mone \ell^2 \varphi_- \varphi_+}$.
We express the two Hopf algebras through their PBW basis \cite{Chari:1994pz}:
\be
    \{X_+^b H^a X_-^c\}_{abc=0}^{\infty} \in \mathcal{U}_q(\mathfrak{su}(2))
    \,\,,\quad
    \{\varphi_-^b \phi^a \varphi_+^c\}_{abc=0}^{\infty} \in \mathcal{U}_q(\an(2)) \,.
\ee
The above $\cU_q(\su(2))$ basis  is obtained from the standard symmetric deformation (with generators $\{J_{\pm},J_3\}$) presented in \cite{Majid:1996kd}, with the re-scaling
\be
    J_{\pm} \,\,\,\to\,\,\, e^{\mp\frac{1}{2}\lambda H} \, X_{\pm}
    \,\,,\quad
    J_3 \,\,\,\to\,\,\, H \,.
\ee 
The $F(\SU_q(2))$ basis given above is obtained by the following change of coordinates on the standard $\SU_q(2)$ matrix element \cite{Majid:1996kd}
\be
    g =
    \begin{pmatrix}
    a & b \\
    -qb^* & a^*
    \end{pmatrix}
    \,\,\,\to\,\,\,
    g =
    \begin{pmatrix}
    \varphi_0 e^{-i \ell \phi} & i\ell \varphi_- \\
    i\ell \varphi_+ & e^{i \ell \phi} \varphi_0
    \end{pmatrix}
    \quad \Rightarrow \quad
    \left\{\,\,
    \begin{aligned}
        & 
        a = \varphi_0 e^{-i \ell \phi} \,, \\
        & 
        b = i\ell \varphi_- \,.
    \end{aligned}
    \right.
    \label{SUq_GroupElement}
\ee
The Haar measures on $F(\SU_q(2))$ and $F(\AN_q(2))$ are resp. given by the standard $q$-deformation of the Haar measure on $\SU(2)$ \cite{Majid:1996kd} and by the $q$-deformation of the Haar measure on $AN(2)$, in the coordinate basis used above.

\subsubsection{Plane-wave}

The $q$-deformed group field theory can be seen as an example of a HAFT presented in Sec. \ref{Sec_HopfAlgebraFieldTheory}, specified by  $H = F_q(\SU(2)) \cong \cU_q(\an(2))$ and $A = \mathcal{U}_q(\mathfrak{su}(2)) \cong F_q(AN(2))$. 
\begin{proposition}[Generalized quantum double of $q$-deformed group field theory]
\label{Prop_qDeformedGenQuantumDouble}
    The generalized quantum double (Def. \ref{Def_GenQuantumDouble}) of the $q$-deformation of group field theory is the Drinfeld double of $\mathcal{U}_q(\mathfrak{su}(2))$:
    \be
        \cD(\cU_q(\su(2)), \cU_q(\an(2)), \sigma) =
        \cD(\cU_q(\su(2)) \cong
        \cU_q(\su(2)) \bowtie_{\sigma} \cU_q(\an(2)) \,,
    \ee
    with (diagonal or canonical) skew pairing
    \be
        \sigma\big(X_+^j H^i X_-^k \,,\, \varphi_-^b \phi^a \varphi_+^c\big)
        = i^{a+b+c} \,
        \delta_{ai} \delta_{bj} \delta_{ck} \,
        a![b]_{q^{2}}![c]_{q^{-2}}! \,,
        \label{qDeformedSkewMap}
    \ee
    and mutual actions
    \be
        \begin{aligned}
            &
            H \rhd \phi = 0 \,,\\
            &
            X_{\pm} \rhd \phi = \mp \ell q^{2B} \varphi_{\pm} \,,\\
            &
            H \rhd \varphi_{\pm} = \pm \ell \varphi_{\pm} \,,\\
            &
            X_{\pm} \rhd \varphi_{\pm} = 0 \,,\\
            &
            X_{\pm} \rhd \varphi_{\mp} = \pm i(\varphi_0 e^{-i\ell \phi} - e^{i\ell \phi} \varphi_0) \,,
        \end{aligned}
        \qquad
        \begin{aligned}
            &
            H \lhd \phi = 0 \,,\\
            &
            X_{\pm} \lhd \phi = i\lambda X_{\pm} \,,\\
            &
            H \lhd \varphi_{\pm} = 0 \,,\\
            &
            X_{\pm} \rhd \varphi_{\pm} = 0 \,,\\
            &
            X_{\pm} \rhd \varphi_{\mp} = \mp i(1-q e^{\pm \lambda H}) \,.
        \end{aligned}
    \ee
    Here
        \be
            [n]_f = \sum_{i=0}^n \, f^i = \frac{1-f^{n+1}}{1-f} 
            \,\,,\quad
            [n]_f! = \prod_{i=1}^n \, [i]_f \,,
        \ee
        are respectively the $q$-number and the $q$-factorial.
\end{proposition}
\begin{proof}
        Let us prove that the skew pairing
        \be
            \sigma\big(X_+^j H^i X_-^k \,,\, \varphi_-^b \phi^a \varphi_+^c\big)
            = i^{a+b+c} \,
            \delta_{ai} \delta_{bj} \delta_{ck} \,
            a![b]_{q^{2}}![c]_{q^{-2}}!
        \ee
        satisfies the relations \eqref{SkewMap_Multiplication-CoProduct}. We carry out the proof just using the bi-algebra structures \eqref{Bi-Algebra_Uq(su)} of $\cU_q(\su(2))$ and \eqref{Bi-Algebra_F(SUq)} of $F(\SU_q(2))$. To compute the $n^{th}$ power of the co-product of $X_{\pm}$ we will make use of the $q$-deformed binomial theorem (see \cite{Majid:1996kd}). Given two elements $A,B$ with obeying the relation $AB = q^{\alpha} BA$, the deformed binomial theorem states that
        \be
            (A + B)^n = \sum_{i=0}^n \, \frac{A^i B^{n-i}}{[n]_{q^{\alpha}}} \,.
        \ee
        Consider first
        \begin{align}
            \sigma\big(\Delta (X_+^j H^i X_-^k) & \,,\, 
            \varphi_-^b \phi^a \varphi_+^c \ot \varphi_-^r \phi^p \varphi_+^s\big) :=
            \sum_{u,v,w=0}^{i,j,k} \, 
            \binom{i}{u}
            \begin{bmatrix}
            j \\
            v
            \end{bmatrix}_{q^{2}}
            \begin{bmatrix}
            k \\
            w
            \end{bmatrix}_{q^{-2}} \, 
            \nonumber \\
            & \qquad
            \sigma\big(
            X_+^{j-v} e^{-v\lambda H} H^{i-u} X_-^{w} \ot 
            X_+^{v} H^{u} e^{w\lambda H} X_-^{k-w} \,,\, 
            \varphi_-^b \phi^a \varphi_+^c \ot \varphi_-^r \phi^p \varphi_+^s\big)
            \nonumber \\
            & =
            \sum_{u,v,w=0}^{i,j,k} \sum_{\alpha,\beta=0}^{\infty} \, 
            \binom{i}{u}
            \begin{bmatrix}
            j \\
            v
            \end{bmatrix}_{q^{2}}
            \begin{bmatrix}
            k \\
            w
            \end{bmatrix}_{q^{-2}} \, \frac{\lambda^{\alpha+\beta} (-v)^{\alpha} w^{\beta}}{\alpha!\beta!}
            \nonumber \\
            & \qquad
            \sigma\big(
            X_+^{j-v} H^{i-u+\alpha} X_-^{w} \ot 
            X_+^{v} H^{u+\beta} X_-^{k-w} \,,\,
            \varphi_-^b \phi^a \varphi_+^c \ot \varphi_-^r \phi^p \varphi_+^s\big)
            \nonumber \\
            & = i^{a+p+j+k} \,
            \delta_{j,b+r} \delta_{k,c+s} [b+r]_{q^2} [c+s]_{q^{-2}} \, \lambda^{a+p-i} \, \sum_{u=0}^i \,
            \frac{i!a!p! \, (-r)^{a-i+u} c^{p-u}}{u!(i-u)!(p-u)!(a-i+u)!}
            \,,\\
            \sigma\big(X_+^j H^i X_-^k & \,,\, 
            \varphi_-^b \phi^a \varphi_+^c \cdot \varphi_-^r \phi^p \varphi_+^s\big) =
            \sigma\big(X_+^j H^i X_-^k \,,\, 
            \varphi_-^b \varphi_-^r (\phi - i\lambda r)^a (\phi + i\lambda c)^p \varphi_+^c \varphi_+^s\big) 
            \nonumber \\
            & =
            \sum_{u,v=0}^{a,p} \, \binom{a}{u} \binom{p}{v} \, (i\lambda)^{a+p-u-v} (-r)^{a-u} c^{p-v} \,
            \sigma\big(X_-^j H^i X_+^k \,,\, 
            \varphi_+^{b+r} \phi^{u+v} \varphi_-^{c+s}\big)
            \nonumber \\
            & 
            = i^{a+p+j+k} \,
            \delta_{j,b+r} \delta_{k,c+s} \, [b+r]_{q^2} [c+s]_{q^{-2}} \, \lambda^{a+p-i} \, \sum_{u=0}^p \,
            \frac{i!a!p! \, (-r)^{a-i+u} c^{p-u}}{u!(p-u)!(i-u)!(p-i+u)!}
            \,.
        \end{align}
        Here
        \be
            \begin{bmatrix}
            n \\
            m
            \end{bmatrix}_f
            = \frac{[n]_f!}{[m]_f! \, [n-m]_f!} \,,
        \ee
        is the $q$-binomial.
        The second identity is more involved, as in the flat case, it requires the star product on $\cU_q(\su(2))$. We refer to App. \ref{App_Bi-Alg_Derivation} which provides the general setting for such type of computation.
        Note that also the unit relations are trivially satisfied.
\end{proof}
\noindent
\begin{proposition}[Dual of the generalized quantum double of $q$-deformed group field theory]
\label{Prop_qDeformedDualGenQuantumDouble}
    The dual of generalized quantum double (Def. \ref{Def_DualGenQuantumDouble}) of the $q$-deformation of group field theory is
    \be
        \mathcal{D}(\mathcal{U}_q(\an(2)), \mathcal{U}_q(\mathfrak{su}(2)), \sigma) =
        \mathcal{U}_q(\an(2)) \blackbowtie_{\sigma} \mathcal{U}_q(\mathfrak{su}(2)) \,,
    \ee
    with skew co-pairing element given by the $q$-star exponential
    \be
        \sigma
        =
        e_{\star \, q^{2}}^{i\varphi_+ \ot X_-} \,
        e_{\star}^{i\phi \ot H} \,
        e_{\star \, q^{-2}}^{i\varphi_- \ot X_+} \,,
        \label{qDeformedSkewEL}
    \ee
    and co-actions \eqref{QuantumDouble_CoActions}
    \begin{align}
        \alpha(\phi) 
        & = 
        (\phi \ot 1) - 
        i\lambda \sum_{n=1}^{\infty} \,
        \frac{(1-q^2)^n}{1-q^{2n}} \, \Big(
        q^{2n} (e^{2ni\ell \phi} \ot e^{-n\lambda H}) \,
        e_{\star \, q^{2}}^{-i\varphi_- \ot X_+} \,
        (i \varphi_+ \ot X_-)^n
        e_{\star \, q^{-2}}^{i\varphi_- \ot X_+} -
        (-i \varphi_- \ot X_+)^n
        \Big)
        \nonumber \\
        & \approx
        \phi \ot 1 + \lambda (\varphi_+ \ot X_- + \varphi_- \ot X_+) + \cO(\lambda^2) + \cO(\ell) + \cO(\ell\lambda)
        \label{CoActions_qDeformed_0}
        \,,\\
        \alpha(\varphi_{\pm})
        & =
        e_{\star \, q^{2}}^{-i\varphi_- \ot X_+} \,
        e_{\star \, q^{-2}}^{i\varphi_- \ot X_+ \, q^{-2}} \,
        (\varphi_{\pm} \ot e^{-\lambda H})
        \approx 
        \varphi_{\pm} \ot 1 - \lambda \varphi_{\pm} \ot H
         + \cO(\lambda^2) + \cO(\ell) + \cO(\ell\lambda) \,,
         \label{CoActions_qDeformed_pm}
        \nonumber \\
        \beta(H) 
        & =
        (1 \ot H) -2\ell \sum_{n=1}^{\infty} \,
        \frac{(1-q^2)^n}{1-q^{2n}} \, \Big(
        q^{2n} (e^{2ni\ell \phi} \ot e^{-n\lambda H}) \,
        e_{\star \, q^{2}}^{-i\varphi_- \ot X_+} \,
        (i \varphi_+ \ot X_-)^n
        e_{\star \, q^{-2}}^{i\varphi_- \ot X_+} -
        (-i \varphi_- \ot X_+)^n
        \Big)
        \nonumber \\
        & \approx
        1 \ot H - 2i\ell (\varphi_+ \ot X_- + \varphi_- \ot X_+)
        + \cO(\lambda) + \cO(\ell^2) + \cO(\ell\lambda)
        \,,\\
        \beta(X_{\pm})
        & \approx
        1 \ot X_{\pm} \pm i\ell (\varphi_{\pm} \ot H - 2 \phi \ot X_{\pm})
        + \cO(\lambda) + \cO(\ell^2) + \cO(\ell\lambda) \,.
    \end{align}
    Since the explicit expression of the mutual co-actions is very involved as it require the use of the star products, we provided the first terms of the expansions wither in the parameter $\ell$ or $\lambda$.
    The $q$-exponential is defined as
    \be
        e_q^x = 
        \sum_{n=0}^{\infty} \, \frac{x^n}{[n]_q!} = 
        \exp\bigg(\sum_{n=1}^{\infty} \, \frac{x^n}{n} \frac{(1-q)^n}{1-q^n}\bigg) \,.
    \ee
\end{proposition}
\begin{proof}
    Let us prove that the skew co-pairing (quantum plane wave)
    \be\label{sigmaqgp}
        \sigma
        =
        e_{\star \, q^{2}}^{i\varphi_+ \ot X_-} \,
        e_{\star}^{i\phi \ot H} \,
        e_{\star \, q^{-2}}^{i\varphi_- \ot X_+} \,,
    \ee
    satisfies the relations \eqref{SkewEl_id-CoProduct},\eqref{SkewEl_CoProduct-id}. We carry out the proof just using the bi-algebra structures \eqref{Bi-Algebra_Uq(su)} of $\cU_q(\su(2))$ and \eqref{Bi-Algebra_F(SUq)} of $F(\SU_q(2))$. 
    In order to split the $q$-$\star$ exponentials we used their property related to the $q$-deformed binomial theorem (see \cite{Majid:1996kd}). Given two elements $A,B$ with obeying the relation $AB = q^{\alpha} BA$, the $q$-exponential satisfies the relation
        \be
            e_{q^{\alpha}}^A \, e_{q^{\alpha}}^B = e_{q^{\alpha}}^{A+B} \,.
        \ee
    Consider first
    \begin{align}
        (\id \ot \Delta) \, \sigma := &
        e_{\star \, q^{2}}^{i\varphi_+ \ot \Delta X_-} \,
        e_{\star}^{i\phi \ot \Delta H} \,
        e_{\star \, q^{-2}}^{i\varphi_- \ot \Delta X_+} =
        e_{\star \, q^{2}}^{i\varphi_+ \ot (1 \ot X_- + X_- \ot e^{\lambda H})} \,
        e_{\star}^{i\phi \ot (1 \ot H + H \ot 1)} \,
        e_{\star \, q^{-2}}^{i\varphi_- \ot (e^{-\lambda H} \ot X_+ + X_+ \ot 1)}
        \nonumber \\
        = &
        e_{\star \, q^{2}}^{i\varphi_+ \ot 1 \ot X_-}
        e_{\star \, q^{2}}^{i\varphi_+ \ot X_- \ot e^{\lambda H}}\,
        e_{\star}^{i\phi \ot 1 \ot H}
        e_{\star}^{i\phi \ot H \ot 1} \,
        e_{\star \, q^{-2}}^{i\varphi_- \ot e^{-\lambda H} \ot X_+}
        e_{\star \, q^{-2}}^{i\varphi_- \ot X_+ \ot 1}
        \nonumber \\
        = &
        e_{\star \, q^{2}}^{i\varphi_+ \ot 1 \ot X_-}
        e_{\star}^{i\phi \ot 1 \ot H} \,
        e_{\star \, q^{2}}^{i\varphi_+ \ot X_- \ot 1}
        e_{\star \, q^{-2}}^{i\varphi_- \ot 1 \ot X_+} \,
        e_{\star}^{i\phi \ot H \ot 1}
        e_{\star \, q^{-2}}^{i\varphi_- \ot X_+ \ot 1}
        \nonumber \\
        = &
        e_{\star \, q^{2}}^{i\varphi_+ \ot 1 \ot X_-}
        e_{\star}^{i\phi \ot 1 \ot H}
        e_{\star \, q^{-2}}^{i\varphi_- \ot 1 \ot X_+} \,
        e_{\star \, q^{2}}^{i\varphi_+ \ot X_- \ot 1}
        e_{\star}^{i\phi \ot H \ot 1}
        e_{\star \, q^{-2}}^{i\varphi_- \ot X_+ \ot 1} 
        := \sigma_{13} \, \sigma_{12} \,.
    \end{align}
    The second identity is more involved, as it requires the star product on $\cU_q(\su(2))$. We refer to App. \ref{App_Bi-Alg_Derivation} which provides the general setting for such type of computation. 
    Note that also the co-unit relations are trivially satisfied.
\end{proof}
\noindent
We note that the plane wave \eqref{qDeformedSkewEL} is similar to the exponential map obtained in \cite{Bonechi:1993sn}.

\subsubsection{Closure constraint and Feynman diagram amplitude}

Let us now derive the fundamental ingredients of the $q$-deformation of three dimensional group field theory, using the HAFT based on the dual of the generalized quantum double of $\cU_q(\su(2))$ above. We review the expressions of the closure constraint and the amplitude of a given graph. \\
The closure constraint \eqref{ClosureConstraint} in the triangulation picture is given by the delta function on $\cU_q(\su(2))$, which enforces the following co-products to vanish 
\be
    \begin{aligned}
        &
        H \ot 1 \ot 1 + 1 \ot H \ot 1 + 1 \ot 1 \ot H = 0 \,,\\
        &
        X_+ \ot 1 \ot 1 + e^{-\lambda H} \ot X_+ \ot 1 + e^{-\lambda H} \ot e^{-\lambda H} \ot X_+ = 0 \,,\\
        &
        X_- \ot e^{\lambda H} \ot e^{\lambda H} + 1 \ot X_- \ot e^{\lambda H} + 1 \ot 1 \ot X_- = 0 \,.
    \end{aligned}
\ee
Such closure condition can be re-packaged as a product of three $AN_q(2)$ group elements, see \cite{Dupuis:2020ndx}. This is the quantum Gauss constraint associated to a triangle \cite{Dupuis:2013lka}.
\medskip \\
The amplitude \eqref{GeneratingFunction_Double} for a graph $\Gamma^*$ is  
\be
    \cA_{\Gamma^*} = 
    %\sum_{\Gamma^*_{\{\cL\}}} 
    \prod_{\{\cL_N\}} \,
    \frac{1}{V_{\SU_q(2)}^{N}} \,
    \int [\dd \phi \dd \varphi_{\pm}]^{N} [\dd H \dd X_{\pm}] \,
    e_{\star \, q^{2}}^{i \Delta^{2N} \varphi_+ \ot X_-} \,
    e_{\star}^{i \Delta^N \phi \ot H} \,
    e_{\star \, q^{-2}}^{i \Delta^N \varphi_- \ot X_+} \,,
    \label{PartitionFunction_qDeformed}
\ee
where $V_{\SU_q(2)}$ is the volume of $F(\SU_q(2))$, and for simplicity, we assumed\footnote{We recall that $\mu$ comes from the normalization of the integration of the delta function (or of the double integration of  the plane-wave), see \eqref{Normalization_Factor}.} $\mu = 1$.

\subsubsection{Relation with Turaev-Viro model (with $q$ real)}

We are going to show in the next section that the Feynman amplitude \eqref{PartitionFunction_qDeformed} is naturally associated to the discretized action of a $BF$ model with non-vanishing cosmological constant.  However before that, we are going to discuss how it is related to the Turaev-Viro model.
To this aim, it is enough to explain how to recover the original Boulatov model \cite{Boulatov:1992vp}. In his work, Boulatov defines a field theory based on the representations of a given group $G$, where the fundamental field is expanded in the Fourier decomposition
\be
    \begin{aligned}
        \Phi(x,y,z) = \sum_{i_1,j_2,j_3} \sum_{\{m,n,k\}} \,
        &
        \Phi^{m_1 \, m_2 \, m_3 ; k_1 \, k_2 \, k_3}_{j_1 \, j_2 \, j_3} \,
        D^{j_1}_{m_1,n_1}(x) \, D^{j_2}_{m_2,n_2}(y) \, D^{j_3}_{m_3,n_3}(z) \\
        \int \dd \omega \, &
        D^{j_1}_{n_1,k_1}(\omega) \, D^{j_2}_{n_2,k_2}(\omega) D^{j_3}_{n_3,k_3}(\omega) \,,
    \end{aligned}
\ee
where $x,y,z,\omega \in G$ and $D^j_{m,n}(x)$ are matrix elements obeying the orthogonality condition
\be
    \int \dx \, D^{j_1}_{m_1,n_1}(x) \, D^{j_2}_{m_2,n_2}(x) = 
    \frac{1}{d_j} \, \delta_{j_1,j_2} \delta_{m_1,m_2} \delta_{n_1,n_2} \,,
\ee
with $d_j$ being the dimension of the irreducible representation associated to $j$. Taking $G = \SU(2)$, the matrix elements $D^j_{m,n}(x)$ are the standard Wigner $D$-matrices, whereas -- as Boulatov explains -- for $G = SU_q(2)$ the matrix elements $D^j_{m,n}(x)$ become the $q$-deformed Wigner matrices \cite{Majid:1996kd}. 
In Hopf algebra field theory, the fundamental field $\phi$ is an element of the tensor product Hopf algebra $H^3$. We claimed that the $q$-deformed group field theory is obtained by choosing $H = F(\SU_q(2))$, for which we used the parametrization in terms of the coordinates $\{\varphi_+,\phi,\varphi_-\}$ of the $\SU_q(2)$ group element \eqref{SUq_GroupElement}. Each field $\Phi$ in Hopf algebra field theory is thus given by the tensor product of three copies of the linear combination of monomials $\varphi_+^j \, \phi^i \, \varphi_-^k$. Therefore, to make contact between Hopf algebra field theory and the original Boulatov model, we define the Fourier expansion for the fields $\Phi \in F(\SU_q(2)^{\times 3})$
\be
    \begin{aligned}
        \Phi = \sum_{i_1,j_2,j_3} \sum_{\{m,n,k\}} \,
        &
        \Phi^{m_1 \, m_2 \, m_3 ; k_1 \, k_2 \, k_3}_{j_1 \, j_2 \, j_3} \,
        D^{j_1}_{m_1,n_1}(\{\varphi_{\pm},\phi\}_1) \, 
        D^{j_2}_{m_2,n_2}(\{\varphi_{\pm},\phi\}_2) \, 
        D^{j_3}_{m_3,n_3}(\{\varphi_{\pm},\phi\}_3) 
        \\
        \int \dd \varphi_+' \dd \phi' \dd \varphi_-' \, &
        D^{j_1}_{n_1,k_1}(\{\varphi_{\pm}',\phi'\}) \, D^{j_2}_{n_2,k_2}(\{\varphi_{\pm}',\phi'\}) \, 
        D^{j_3}_{n_3,k_3}(\{\varphi_{\pm}',\phi'\}) \,,
    \end{aligned}
\ee
where we used the symbol $\{\varphi_{\pm},\phi\}_i$ for the monomial in the coordinates $\varphi_+,\phi,\varphi_-$ in the $i^{th}$ tensor space, and once again the $D^{j}_{m,n}$ are matrix elements obeying the same orthogonality condition. 
The HAFT defined for $H=\SU_q(2)$ we proposed is thus equivalent to the Boulatov model. In particular, the Feynman amplitude of the Boulatov model, and thus the one of HAFT \eqref{GeneratingFunction_Triangulation}, can be given in terms of the $q$-deformed 6-$j$ symbols and thus it allows to recover the Turaev-Viro invariant \cite{Turaev:1992hq} with $q$ real.

\subsection{Plane-wave and discretization of the $BF$ action}
\label{sec:pw-disc}

The amplitudes we calculated are expressed in terms of the plane-wave. Since the model is topological we expect that it should be related to the $BF$-action. More concretely, we expect that the amplitude \eqref{PartitionFunction_qDeformed} provides a regularization/discretization of $e^{i\cS^{\lambda}_{BF}}$, where $\cS^{\lambda}_{BF}$ is the $BF$ action with a cosmological constant as given in \cite{Dupuis:2020ndx}.

\subsubsection{$BF$ action and its limits}
\label{Sec_BF}

Let us focus on the Euclidian case with a negative or null cosmological constant $\Lambda \leq 0$, as we know this corresponds to the the deformation with $q$ real \cite{Dupuis:2020ndx}.
The fundamental fields are the frame field $e$ and the connection $A$, which are respectively 1-forms with values in the boosts $K$ and the Lie algebra $\su(2)$. Noting $\la\,,\,\ra$ the Killing form of $\sl(2,C)$, the $BF$ model is governed by the action
\be
\label{BFAction_Boosts}
    \cS^\Lambda_{BF}[A,e] = \int_{\cM} \, 
    \la e \,,\, \Big(F + \frac{\Lambda}{6} \, (e \times e)\Big)\ra \,.
\ee
The frame field variables are difficult to quantize since they are valued in the boosts. It was shown in \cite{Dupuis:2020ndx} that it is convenient to do a canonical transformation
\be
    \omega_I = A_I + \epsilon_{IJK} n^J e_K
    \,\,,\quad 
    n^I = (0,0,\lambda)
    \,\,,\quad 
    n^2 = \lambda^2 = -\Lambda \,,
\ee
which makes the frame field discretizable with values in the Lie algebra $\an(2)$. With this new connection the action becomes, up to a boundary term, 
\be
\label{BFAction_AN}
    \cS^\lambda_{BF} [\omega,e] = \int_{\cM} \, 
    \big(e^I \wedge F_I + \lambda \, (\omega_- \wedge e_3\wedge e_+-\omega_+ \wedge e_-\wedge e_3)\big) \,.
\ee
Note that this action is also valid if we have $\Lambda = \lambda^2 = 0$ but $\lambda \neq 0$. On the other hand if $\lambda = 0$, we just recover the standard $BF$ action $\cS_{BF}^{\Lambda=0}= \cS_{BF}^{\lambda=0}$.
\smallskip \\
Let us emphasize that this action depends on two constants, Newton's constant $G$ (hence the Planck length $\ell$) and the cosmological constant (hence $\lambda$). The constant $\ell$ typically encodes the curvature (radius) of $\SU(2)$, which is why  the fluxes (Poisson) non-commutativity is encoded by $\ell$, the Planck length. The cosmological constant or $\lambda$ encodes the curvature of the group $\AN(2)$. It encodes then the (Poisson) non-commutativity of the holonomy components.  

When looking at the action \eqref{BFAction_AN}, we can therefore take two limits, one where $\lambda\rightarrow0$ and another one where $\ell\rightarrow0$. This latter one is less common from the (quantum) gravity side since in a sense it amounts to switch off gravity.    

In terms of the underlying symmetry structure, the actions \eqref{BFAction_Boosts} or \eqref{BFAction_AN} are built from the classical Drinfeld double $\sl(2,C)\cong \su(2)\bowtie \an(2)$. The two limits we just discussed become then 
\begin{align}
    \su(2)\bowtie \an(2) 
    \,\,\,\rightarrow\,\,\,
    \left\{\begin{array}{c}
         \su(2)\ltimes \bbR^3, \quad \lambda=0  \\
         \bbR^3\rtimes \an(2), \quad \ell=0
    \end{array} \right. 
\end{align}
In terms of the action, we can write \eqref{BFAction_AN} highlighting the different brackets which would be turned to zero, 
\begin{align}
    \cS^\lambda_{BF} [\omega,e] = 
    \int_\cM \, \Big(
     \la e \,,\, d \omega\ra  +\frac{1}{2}  \la e\,,\,  [\omega\,,\,\omega]_{\su(2)} \ra + \frac{1}{2}
     \la [e \,,\, e]_{\an(2)} \,,\, \omega\ra\Big) \,.
\end{align}
We have therefore the two types of limiting actions
\be
\label{limiting actions}
    \cS^{\lambda}_{BF} [\omega,e]  
    \,\,\,\to\,\,\,
    \left\{\,\,
    \begin{aligned}
        &
        \cS^{\su(2)} =
        \int_{\cM} \, \Big(
        \la e \,,\, d \omega + \frac{1}{2} [\omega\,,\,\omega]_{\su(2)} \ra\Big)
        \,\,\,,\quad \lambda=0 \\
        &
        \cS^{\an(2)} =
        \int_{\cM} \, \Big(
        \la d e + \frac{1}{2}
        [e \,,\, e]_{\an(2)} \,,\, \omega\ra\Big) \, 
        \,\,,\quad \ell=0
    \end{aligned} 
    \right. 
\ee
where in the last case, we omitted the boundary term and we recall that in this case, $\omega$ is with value in the abelian Lie algebra $\bbR^3$.  In the first case, it is $e$ which is with value in $\bbR^3$. The first action is $\su(2)$ $BF$ theory, while the second one is a $\an(2)$ $BF$ theory.

\subsubsection{Discretized $BF$ actions}

Let us recall the standard discretization procedure of a $\g$ $BF$ theory, with no cosmological constant.
\begin{align}
\cS= \int_\cM \la B, F(A)\ra\,.   
\end{align}
The field $B$ is a $\g^*$ valued one-form and $A$ is a $\g$-valued connection. The 3$d$ manifold $\cM$ is discretized on a cellular decomposition $\Gamma$, noting $\Gamma^*$ its dual complex. According to the Poincar\'e duality, we discretize $B$ and $F$ on  dual structures. 
The curvature $F(A)$ is chosen to be discretized along an holonomy $g_e$ in $G$ forming a closed loop, spanning a face $e^*$ dual to an edge $e$. In this construction, we can  attribute to each link $l$ an holonomy $g_l \in G$, such that $g_e = \prod_l \, g_l$. Upon quantization, such product is actually characterized by the co-product
\be
    g_e = \prod_l g_l \,\,\,\rightarrow\,\,\, \Delta^n g \,.
\ee
The closed loop defines the boundary of a surface dual to an edge. We discretize $B$ on such edge as a 3$d$ vector\footnote{We mean that that the coordinates are $\star$ non-commutative coordinates, so that $\mathbf{x}_e \in F_{\star}(\bbR^d)$. } $\mathbf{x}_e \in \bbR^d_{\star}$, where $d$ is the dimension of $\g$. 
The natural discretization of the three dimensional $BF$ action can be represented with the help of the co-product or the $\star$ product representation. 
\be
    e^{i \cS_{BF}[A,B]} \approx 
    e^{\la g_e \, \mathbf{x}\ra} 
    \,\,\,\to\,\,\,\left\{ \begin{array}{l}
        (\Delta^N\otimes \id) e^{ g \cdot  \mathbf{x}_e } = e^{ \Delta^{N} g \cdot  \mathbf{x}_e } \\
         e_\star^{ \la g \,,\,  \mathbf{x}_e\ra  } =  e_\star^{ \la\prod_l g_l \,,\,  \mathbf{x}_e\ra } = \star_l  e_\star^{\la  g_l \,,\, \mathbf{x}_e\ra } 
    \end{array}\right.
     \,,
    \label{FlatDiscretizationPlaneWave}
\ee
where $\Delta^{N} g \cdot \mathbf{x}_e $ is the quantization of $\la g_e \,,\, \mathbf{x}\ra$. We emphasize that there is a freedom in the choice of coordinates for the element $g_e$. As operators we will have clearly different operator coordinates, and dually different $\mathbf{x}_e$. We will have then different co-product for the coordinate operators which is equivalent to different $\star$ products. 

For instance, taking the specific example of $\SU(2)$, we can use the Euler angles $(\varphi_\pm, \phi)$ as a parametrization of an element of $\SU(2)$, or the "standard" parameterization  $(p_\pm,p_3,p_0)$
\cite{Baratin:2010wi})  
\be
    g =
    \begin{pmatrix}
        e^{i\phi} & ie^{i\phi}\varphi_- \\
        i e^{i\phi}\varphi_+ & e^{-i\phi} -\varphi_-\varphi_+e^{i\phi} \\
    \end{pmatrix} = 
    \begin{pmatrix}
        p_0 + i p_3 & i  p_- \\
        i  p_+ & p_0 - i p_3
    \end{pmatrix}\,.
    \label{SU(2)_R^3_GroupElements}
\ee
where $p_0$ depends on the other coordinates. The two sets of coordinates are obviously related as explicitly seen in  \eqref{SU(2)_R^3_GroupElements}. 
\begin{align}\label{change}
    p_\pm = e^{i\phi}\varphi_\pm 
    \,\,,\quad 
    p_3 = \sin \phi + \frac{i}{2}\varphi_+\varphi_- e^{i\phi}
    \,\,,\quad  
    p_0 =\cos \phi - \frac{1}{2}\varphi_+\varphi_- e^{i\phi}.
 \end{align}
We also always have to bear in mind that one might need more than one coordinate patch so one needs to be cautious \cite{Joung:2008mr}. 

The dual variables to the coordinates $(\varphi_\pm, \phi)$ and  $(p_\pm,p_3,p_0)$ are different. They can be seen as the vectors fields generating the translations on the sphere $S^3\cong \SU(2)$. Depending on the choice of coordinates, they take different shapes. The co-product in the $p$-variables is different than the one in the $\varphi$-variables, or equivalently the $\star$ product for the dual variables of the $p$-variables is different than the one of the dual of the $\varphi$-variables. Nevertheless they are related by the change of coordinates \eqref{change}. 

Finally, we emphasize that while we chose the group $\SU(2)$ to illustrate our discussion, this is true for any group, in particular $\AN(2)$.

\medskip

We have seen that we can have two actions $\cS^{\su(2)}$, $\cS^{\an(2)}$ obtained as different limits of a single action $\cS^{\lambda}$. The limiting actions in \eqref{limiting actions} can be discretized along the same procedure. Since we use  the same fields $(e,\omega)$ in both cases, the fields $(e,\omega)$ should be discretized on the same structure in each case. Following the loop quantum gravity picture, we choose to discretize the frame field $e$ on the edges of $\Gamma$, while $\omega$ is discretized on the links of $\Gamma^*$. 

The fundamental element in both discretization  takes the same general shape but with different locations according to the  different cases. 
\be
\left\{ \begin{array}{c} e^{\la g_l \textbf{x}_e\ra}, \textrm{ with }
     g_l\in \SU(2) \textrm{ on } \Gamma^*, \quad  \textbf{x}_e \in \bbR^3  \textrm{ on } \Gamma\\
     e^{\la \textbf{x}_l g_e \ra}, \textrm{ with } g_e\in \AN(2) \textrm{ on } \Gamma, \quad  \textbf{x}_l \in \bbR^3  \textrm{ on } \Gamma^* 
\end{array}\right.
\ee
Using the co-product, we can construct the discretized actions. We note the choice of coordinates $k$, $q$ on respectively $\SU(2)$ and  $\AN(2)$ and $X$, $\textbf{x}$ on respectively $\bbR^3_{\star_\su}$ and $\bbR^3_{\star_\an}$.  We have therefore the discretization of the actions \eqref{limiting actions} given by    
\begin{align}
    \,\cS^{\su(2)}_d \approx \, \cS_d^{\lambda=0}
    & = \Delta^n k ^a \otimes X_a
    \label{DiscreteAction_Flat} \\
    \,\cS^{\an(2)}_d \approx \, \cS_d^{\ell=0}
    & = \Delta^n \textbf{x} ^a \otimes q_a 
    \label{no-grav}
\end{align}
To obtain the discretization of the general case $e^{i\cS^\lambda_{BF} [\omega,e]}$, we need to consider the HAFT and its associated plane-wave.

\subsubsection{Recovering the discretized $BF$ actions from the plane-wave}

As we alluded earlier, we expect that the plane-wave \eqref{qDeformedSkewEL} would provide a discretization of \eqref{BFAction_AN} since this formulation is the proper one leading to the quantum group variables \cite{Dupuis:2020ndx}. The fact that the triangulation is decorated by quantum group elements indicates that we are dealing with the quantum version of a homogeneously curved discrete geometry \cite{Dupuis:2013haa}.

Following the previous section, if we consider an edge dual to a loop formed of $n$ links, we would have the discretization scheme,
\be
    (\Delta^n \ot \id)\, \sigma 
    \,\,\,\leftrightarrow\,\,\,
    e^{i\cS^\lambda_{BF} [\omega,e]} \,.
    \label{PlaneWave_BFAction_curved}
\ee
We note that the main complication relies on the fact that the amplitude expressed in terms of the  fields \eqref{PlaneWave_BFAction_curved} is expressed in terms of the fields all included in a single exponential, contrary to the plane-wave formula \eqref{qDeformedSkewEL}, where each component has its own exponential (similarly to the Euler angles parametrization). Putting them together under a common exponential generally requires to use  the Baker-Campbell-Hausdorff formula
\be
    e^{A} e^{B} \approx e^{A+ B+ \demi [A,B]} \,,
    \label{BCHFormula}
\ee 
which we will truncate at first order to identify the relationship between the continuum actions and their discrete versions.

\medskip

As recalled in the table \ref{Tab_ClassicalLimit}, one can obtain different bi-algebras as limiting cases of the deformed case. We recover the bi-algebras associated to $\AN(2)$ in the limit $\ell \to 0$ and those associated to $\SU(2)$ when $\lambda \to 0$. This is consistent with the different limits we discussed in the previous section. 
\begin{table}
    \centering
    \begin{tabular}{ | l | l | l | }
        \cline{2-3}
        \multicolumn{1}{c|}{$q = e^{\ell \lambda}$} & 
        \multicolumn{1}{|c|}{\cellcolor{lightgray} $\lambda \rightarrow 0$ : $\cU(\su(2))$} &
        \multicolumn{1}{|c|}{\cellcolor{lightgray} $\ell \rightarrow 0$ : $F(AN(2))$} \\
        \hline \hline
        $[H , X_{\pm}] = {\pm} 2 \ell J_{\pm}$ &
        $[H , X_{\pm}] = {\pm} 2 \ell J_{\pm}$ &
        $[H , X_{\pm}] = 0$ \\
        $[X_+ , X_-] = \ell^2 \frac{\sinh(\lambda H)}{\sinh(\ell\lambda)}$ &
        $[X_+ , X_-] = \ell H$ &
        $[X_+ , X_-] = 0$ \\
        $\Delta H = H \ot 1 + 1 \ot H$ &
        $\Delta H = H \ot 1 + 1 \ot H$ &
        $\Delta H = H \ot 1 + 1 \ot H$ \\
        $\Delta X_+ = X_+ \ot 1 + e^{-\lambda H} \ot X_+$ &
        $\Delta X_+ = X_+ \ot 1 + 1 \ot X_+$ &
        $\Delta X_+ = X_+ \ot 1 + e^{-\lambda H} \ot X_+$ \\
        $\Delta X_- = X_- \ot e^{\lambda H} + 1 \ot X_-$ &
        $\Delta X_- = X_- \ot 1 + 1 \ot X_-$ &
        $\Delta X_- = X_- \ot e^{\lambda H} + 1 \ot X_-$ \\
        \hline 
        \cline{2-3}
        \multicolumn{1}{c|}{} & 
        \multicolumn{1}{|c|}{\cellcolor{lightgray} $\lambda \rightarrow 0$ : $F(SU(2))$} &
        \multicolumn{1}{|c|}{\cellcolor{lightgray} $\ell \rightarrow 0$ : $\cU(\an(2))$} \\
        \hline
        $\varphi_0^2 = 1 - \ell^2 \varphi_+ \varphi_-$ &
        $\varphi_0^2 = 1 - \ell^2 \varphi_+ \varphi_-$ &
        $\varphi_0^2 = 1$ \\
        $[\phi , \varphi_{\pm}] = -i\lambda \varphi_{\pm}$ &
        $[\phi , \varphi_{\pm}] = 0$ &
        $[\phi , \varphi_{\pm}] = -i\lambda \varphi_{\pm}$ \\
        $[\varphi_+ , \varphi_-] = 0$ &
        $[\varphi_+ , \varphi_-] = 0$ &
        $[\varphi_+ , \varphi_-] = 0$ \\
        \scriptsize{$\Delta (\varphi_0 e^{-i\ell\phi}) = 
        \big( \varphi_0 e^{-i\ell \phi} \ot \varphi_0 e^{-i\ell \phi} - q \ell^2 \varphi_- \ot \varphi_+ \big)$} &
        \scriptsize{$\Delta (\varphi_0 e^{-i\ell\phi}) = 
        \big( \varphi_0 e^{-i\ell \phi} \ot \varphi_0 e^{-i\ell \phi} - q \ell^2 \varphi_- \ot \varphi_+ \big)$} &
        $\Delta \phi = \phi \ot 1 + 1 \ot \phi$ \\
        $\Delta \varphi_+ = \varphi_+ \ot \varphi_0 e^{-i\ell \phi} + e^{i\ell \phi} \varphi_0 \ot \varphi_+$ &
        $\Delta \varphi_+ = \varphi_+ \ot \varphi_0 e^{-i\ell \phi} + e^{i\ell \phi} \varphi_0 \ot \varphi_+$ &
        $\Delta \varphi_+ = \varphi_+ \ot 1 + 1 \ot \varphi_+$ \\
        $\Delta \varphi_- = \varphi_- \ot e^{i\ell \phi} \varphi_0 + \varphi_0 e^{-i\ell \phi} \ot \varphi_-$ &
        $\Delta \varphi_- = \varphi_- \ot e^{i\ell \phi} \varphi_0 + \varphi_0 e^{-i\ell \phi} \ot \varphi_-$ &
        $\Delta \varphi_- = \varphi_- \ot 1 + 1 \ot \varphi_-$
        \\
        \hline
    \end{tabular} 
    \caption{The different limiting cases.}
    \label{Tab_ClassicalLimit}
\end{table} 

One of our main statements is that the plane wave \eqref{qDeformedSkewEL} provides the discretization of the $BF$ amplitude with non-vanishing cosmological constant
\be
    e^{i\cS_{BF}^{\lambda}[\omega,e]} \approx 
    e^{i\cS_d^{\Lambda}} = (\Delta^n \ot \id) \, \big(
    e_{\star \, q^{2}}^{i \varphi_+ \ot X_-} \,
    e_{\star}^{i \phi \ot H} \,
    e_{\star \, q^{-2}}^{i \varphi_- \ot X_+} \big) \,.
\ee
Using the full BCH formula would provide us the exact discretization of the $BF$ action with or without cosmological constant. However since it is very complicated, we will truncate the formula at the first order in $\ell$ and $\lambda$.

Since we can recover the actions \eqref{limiting actions} from the main action $\cS^\lambda$, we should be recovering the discretized actions  \eqref{DiscreteAction_Flat} and \eqref{no-grav} from the full plane-wave. 
To this aim, we use the BCH formula on  the $q$-deformed plane wave \eqref{qDeformedSkewEL} to express it as a single exponential and take the  different limits.
\begin{align}
    \sigma & =
    e_{\star \, q^{2}}^{i \varphi_+ \ot X_-} \,
    e_{\star}^{i \phi \ot H} \,
    e_{\star \, q^{-2}}^{i \varphi_- \ot X_+} 
    \approx
    e_{\star}^{i \varphi_+ \ot X_- + \frac{1}{2} \ell\lambda \varphi_+^2 \ot X_-^2} \,
    e_{\star}^{i \phi \ot H} \,
    e_{\star}^{i \varphi_- \ot X_+ - \frac{1}{2} \ell\lambda \varphi_-^2 \ot X_+^2} 
    \nonumber \\
    & \approx
    e_{\star}^{
    i \varphi_+ \ot X_- + i \phi \ot H + i \varphi_- \ot X_+ +
    \frac{1}{2} \ell\lambda \varphi_+^2 \ot X_-^2 - 
    \frac{1}{2} \ell\lambda \varphi_-^2 \ot X_+^2 +
    \frac{1}{2} [\phi \ot H \,,\, \varphi_+ \ot X_-] -
    \frac{1}{2} [\phi \ot H \,,\, \varphi_- \ot X_+] +
    \frac{1}{2} [\varphi_- \ot X_+ \,,\, \varphi_+ \ot X_-]} 
    \nonumber \\
    & \approx  
    e_{\star}^{
    i\varphi_+ (1 + i\ell \phi) \ot X_- + 
    i(1 + i\ell \phi) \varphi_- \ot X_+ +
    i(\phi - \frac{i}{2}\ell \varphi_+\varphi_-) \ot H +
    \frac{i\lambda}{2} \, (\varphi_- \ot H X_+ - \varphi_+ \ot X_- H)} \,.
    \label{CurvedPlaneWave_BCH}
\end{align}
In the case where $\lambda=0$, we are looking at the standard GFT with group $\SU(2)$, and in this case   $H = F(\SU(2))$ and $A = \mathcal{U}(\mathfrak{su}(2)) \cong F_\star(\bbR^3)$. 
At first order in $\ell$, we have the contribution in the exponent  is given by 
\begin{align}
    \varphi_+ (1+\ell \phi) \ot X_- + 
    \varphi_- (1+\ell \phi) \ot X_+ H + 
    (\phi - i\ell/2 \, \varphi_- \varphi_+) \ot H\,,
\end{align}
which allows us to identify the coordinates $k$ in the discretization \eqref{DiscreteAction_Flat}, at first order in $\ell$, 
\begin{align}
    k_+= \varphi_+ (1+\ell \phi), \quad k_-= \varphi_- (1+\ell \phi), \quad k_3= (\phi - i\ell/2 \, \varphi_- \varphi_+)\,. 
\end{align}
This means that the discretized $BF$ action with $\lambda=0$ can be written as 
\begin{align}
    \,\cS_{\su(2)} \approx \, \cS_d^{\lambda=0}
    & = \Delta^n k^a \otimes X_a\nonumber \\  
    &\approx  (\Delta^n \ot \id) 
    \, \big(
    \varphi_+ (1+\ell \phi) \ot X_- + 
    \varphi_- (1+\ell \phi) \ot X_+  + 
    (\phi - i\ell/2 \, \varphi_- \varphi_+) \ot H\big) \,.
\end{align}

We are now interested at looking at the case where $H = F_\star(\bbR^3) \cong \cU(\an(2))$ and $A = F(AN(2))$. 
The contribution in the exponent  is 
\begin{align}
    \varphi_+  \ot X_-(1-\frac{\lambda}{2} H) + 
    \varphi_-  \ot (1+\frac{\lambda}{2} H)X_+ + 
 \phi  \ot H\,. 
\end{align}
The coordinates $(\varphi_\pm, \phi)=(x_\pm,x)=\textbf{x}$ are non-commutative variables of the $\an(2)$ Lie algebra type, while the coordinates $H,X_\pm$ are commutative according to table \ref{Tab_ClassicalLimit}. Noting the coordinate $q$ on the group $\AN(2)$, we have then at first order in $\lambda$
\begin{align}
    q_-=X_-(1-\frac{\lambda}{2} H), \quad q_+= (1+\frac{\lambda}{2} H)X_+, \quad q_3=H. 
\end{align}

The discretized action \eqref{no-grav} in the $\ell=0$ case becomes then 
\begin{align}
    \,\cS_{\an(2)} \approx \, \cS_d^{\ell=0}
    & = (\Delta^n \textbf{x}) ^a \otimes q_a\nonumber \\  
    &\approx  (\Delta^n \ot \id) 
    \, \big(
    \varphi_+ \ot X_-(1-\frac{\lambda}{2} H) + 
    \varphi_- \ot (1+\frac{\lambda}{2} H)X_+  + 
    \phi\ot H\big) \,.
\end{align}

\medskip 

Finally, the case where both $\ell$ and $\lambda$ are not zero corresponds to the case  specified by  $H = F_q(\SU(2)) \cong \cU_q(\an(2))$ and $A = \mathcal{U}_q(\mathfrak{su}(2)) \cong F_q(AN(2))$. The approximation \eqref{CurvedPlaneWave_BCH} provides therefore the basic building block to construct the discretization of the $BF$ action $\cS^{\lambda}_{BF}$. 
The extra term is naturally identified with the cosmological constant contribution in \eqref{BFAction_AN}. Therefore,  working  at the first order in $\ell$ and $\lambda$, we obtain the discretization of the the $BF$ action with cosmological constant.
\begin{align}
    \cS_{BF}^{\lambda}[\omega,e] &=\cS_{BF}^{\lambda=0}[\omega,e] 
    +\lambda 
    \, \big( \omega_- \wedge e_3\wedge e_+-\omega_+ \wedge e_-\wedge e_3\big) 
    \nonumber \\
    \approx \cS_d^{\lambda} & =
    \cS_d^{\lambda=0} + \frac{\lambda}{2}  \,  \big(\Delta ^n\varphi_- \ot H X_+ - \Delta ^n\varphi_+ \ot X_- H\big) \,.
\end{align}
This corresponds to the generalization of the non-commutative simplicial gravity picture introduced by Baratin and Oriti \cite{Baratin:2010wi}. We have non-commutative coordinates both on the triangulation side (as they did) and in the dual complex. On the triangulation side the non-commutativity is further modified with respect to \cite{Baratin:2010wi} due to the presence of the cosmological constant.

\section{Conclusion}

We proposed a  model based on Hopf algebras that provides a (3d) topological state sum. The formulation of the model is a generalization of the notion group field theories, where the fundamental degrees of freedom are encoded in a gauge invariant field and the Feynman diagram amplitudes are expressed as a sum over the amplitudes associated to three dimensional geometries. 

The field theory is discussed both in the configuration and momentum spaces, and a map (generalized Fourier transform) between the two is provided. The object used as the kernel of such transformation is a generalized notion of plane-wave given by the canonical element of $D^*(A,H,\sigma)$.
It is the first time, to the best of our knowledge, that such canonical element takes a central place in the definition of of a topological model and is used to provide a discretization of the $BF$ action. 

The main scope of this work was to construct a field theory generating curved discrete geometries as Feynman diagrams, in such a way that the amplitudes would provide a discretization of the $BF$ action with non-vanishing cosmological constant. In this sense, HAFT provides a non-commutative formulation of simplicial 3d gravity. The non-commutativity appears due to the homogeneous curvature induced by the cosmological constant (and also due to the Planck length \cite{Baratin:2010wi}).  We emphasize that the discretization is not  done at the level of the action but instead at the level of the amplitude $e^{i\cS^\lambda_{BF}}$. 

Let us mention several directions that we would find interesting to explore. 

\paragraph{Relation with integrable systems.}
The canonical element/plane-wave associated to the dual Hopf algebras also appeared in the context of integrable systems under the shape of the transfer/transport matrix, or T-matrix \cite{Fronsdal:1991gf, Fronsdal:1993ip, Bonechi_1994}. In this context, we consider a lattice which sites are associated with a T-matrix $\sigma$. It can be written as 
\be 
    \sigma=e^{x\cdot k} \,, 
    \label{Lax}
\ee 
where $L=x\cdot k$ originates from a Lax pair and the $x^i$ are the dynamical variables, which can be seen as generating a Lie algebra $\g_x$, while the $k_i$ are the generators of the Lie algebra $\g_k$. We can multiply the T-matrices sitting at different sites, with dynamical variables $x$ and $x'$ to obtain a new one $x''=x \oplus x'$. This is precisely the structure we would expect from a plane-wave. The dynamical variables are equipped with a Poisson bracket which induces a Poisson structure for the T-matrix. We expect then that the Poisson bracket (associated with the dynamical variables) should be compatible with the product of the T-matrices. Since there is an obvious symmetry in \eqref{Lax}, we can make the same type of construction for the $k$ sector. 
The main features of the T-matrix are captured by the expected plane-wave properties 
\be
    \begin{aligned}
        &
        \sigma_{x,k_1}\sigma_{x,k_2} = e^{ix\cdot k_1} e^{ix\cdot k_2} = e^{ix\cdot (k_1\oplus k_2)} \equiv  \sigma_{x, (k_1\oplus k_2)}  \,,\\
        &
        \sigma_{x_1,k}\sigma_{x_2,k} = e^{ix_1\cdot k}  e^{ix_2\cdot k} = e^{i(x_1\oplus x_2)\cdot k} \equiv \sigma_{(x_1\oplus x_2)\cdot k} \,,
    \end{aligned}
\ee
such that the product is a Poisson map.
Note that, at the quantum level, by performing some projection in terms of representations, we can recover from the T-matrix the notion of $R$-matrix \cite{Majid:1996kd} which also play a fundamental role in the study of integrable systems.   
It would be interesting to explore whether the fact that this common structure appears both for integrable systems and 3d gravity could clarify some interesting questions, such as holography for example.

\paragraph{Generalization to quasi Hopf algebras.} 
Our construction relied on (co-associative) Hopf algebras. The original TV model is defined for quasi Hopf algebras, such as for the deformation at $q$ root of unity of $\SU(2)$. It would be interesting to explore how our proposed framework extends to this quasi case. As a starting point we could explore how twisting the Drinfeld double and weakening the co-associativity would modify our construction.

\paragraph{Introducing matter.} 
Krasnov has argued that to introduce matter, one should not only consider a field on some copies of $H=F(\SU(2))$ as we did, but instead on the full Drinfeld double $H =  F(\SU(2)) \rtimes \cU(\su(2))$ \cite{Krasnov:2005GFT3dMatterDrinfeldDouble}. 
It would be interesting to see how this is implemented in our general framework so that we could deal with particles coupled to gravity in the case where the cosmological constant is not zero. Along a similar line, it is well understood now that the Kitaev model is equivalent to a $BF$ theory with particle excitations, in particular the fundamental symmetry is given by the Drinfeld double. In \cite{Buerschaper2010}, some abstract Fourier transform was discussed, for the specific case of a quadrangulation which is therefore self-dual. It would be interesting to check how our proposal of Fourier transform could be applied in this set up.

\paragraph{Four dimensional model: quantum 2-groups.}
Despite the HAFT can be easily generalized to the $d$ dimensional case, we discussed the three dimensional version of the model for two reasons. 
First,  in 3$d$ it is easier  to visualize the underlying geometric objects. Second, if we were to consider the 4d case, instead of using Hopf algebras we could/should use categorified Hopf algebras according to the categorical ladder \cite{Crane:1994fw}. 
In particular, in 4d we expect to have triangulation faces, as dual to the links of the dual complex to be decorated. Without any edge decoration, it is actually not possible to decorate faces by non-abelian group elements (eg at the classical level) due to the Eckmann-Hilton argument \cite{EckmannHilton:19622Group}. This therefore  prevents  the description of a discrete curved geometry unless we also decorate the edges. A 2-group or crossed module or their deformation appears therefore as the natural structure to describe discrete curved geometries in 4d. Even though there already exists a proposal for the notion quantum 2-group (see \cite{Majid:2008iz}) the relevant structures and duality properties (such as a Fourier transform for 2-groups) to extend our construction   are not clear yet. Note that some recent work has extended the notion of GFT when using 2-groups instead of groups \cite{Girelli:2022bdf}. It is therefore an interesting question to determine what is the notion of plane-wave when dealing with 2-groups.

\section*{Acknowledgment}
FG would like to thank A. Baratin for discussions at the early stage of this work, a long long time ago.

\appendix

\section{Further mathematical details}
\label{app:hopfproof}

\input{AppProofs}

\section{Further details on the Hopf algebra field theory}

\input{AppHAFT}

\section{Derivation of skew symmetric bi-algebras}
\label{App_Bi-Alg_Derivation}

\input{AppDerivationQuantumDouble}

\section{Further examples}
\label{App_Examples}

\input{AppExamples}

\bibliographystyle{Biblio}
\bibliography{biblio}

\end{document}

%% file: Pics/Field.tex
\begin{tikzpicture}[scale=1]
    \draw[- , thick] (1.5,0) -- (-1.5,0) -- (0,2.5) node[right] {$\hat{\phi}$} -- (1.5,0);
    
    \coordinate (c) at ($ 1/3*(1.5,0) + 1/3*(-1.5,0) + 1/3*(0,2.5) $);
    
    \draw[fill , darkblue] (c) circle [radius=0.04] node[above] {$\phi$}; 
    
    \draw[- , thick , darkblue] (c) -- (0,-0.6);
    \draw[- , thick , darkblue] (c) -- (-1.5,1.5);
    \draw[- , thick , darkblue] (c) -- (1.5,1.5);
\end{tikzpicture}

%% file: Pics/GuageSymmetry.tex
\begin{tikzpicture}[scale=1.2 , rotate around y=-20]
    \coordinate (v1) at (1.5,0,0);
    \coordinate (v2) at (-1.5,0,0);
    \coordinate (v3) at (0,2.5,0);
    
    \coordinate (w1) at (1.5,0,-4);
    \coordinate (w2) at (-1.5,0,-4);
    \coordinate (w3) at (0,2.5,-4);
    
    \coordinate (c) at ($ 1/3*(v1) + 1/3*(v2) + 1/3*(v3) $);
    \coordinate (c') at ($ 1/3*(w1) + 1/3*(w2) + 1/3*(w3) $);
    
    \draw[- , thick, dotted] (v1) -- (v2) -- (v3) -- cycle;
    \draw[- , thick, dotted] (w1) -- (w2) -- (w3) -- cycle;
    
    \draw[fill , darkblue] (c) circle [radius=0.04];
    \draw[fill , darkblue] (c') circle [radius=0.04] node[above , scale=0.9] {$\phi$}; 
    
    \draw[- , thick , darkblue] (c') -- (0,-0.6,-4);
    \draw[- , thick , darkblue] (c') -- (-1.5,1.5,-4);
    \draw[- , thick , darkblue] (c') -- (1.5,1.5,-4);
    
    \draw[- , thick , dashed , darkblue] (c) -- (c');
    
    \draw[- , thick, dotted , darkblue] (c) -- (0,-0.6,0);
    \draw[- , thick, dotted , darkblue] (c) -- (-1.5,1.5,0);
    \draw[- , thick, dotted , darkblue] (c) -- (1.5,1.5,0);
\end{tikzpicture} 

%% file: Pics/Propagator.tex
\begin{tikzpicture}[scale=1.6 , rotate around y=-45]
    \coordinate (v1) at (1.5,0,0);
    \coordinate (v2) at (-1.5,0,0);
    \coordinate (v3) at (0,2.5,0);
    
    \coordinate (w1) at (1.5,0,-1.5);
    \coordinate (w2) at (-1.5,0,-1.5);
    \coordinate (w3) at (0,2.5,-1.5);
    
    \coordinate (c) at ($ 1/3*(v1) + 1/3*(v2) + 1/3*(v3) $);
    \coordinate (c') at ($ 1/3*(w1) + 1/3*(w2) + 1/3*(w3) $);
    
    \draw[-> , thick] ($ 0.975*(v1) + 0.025*(v2) $) -- ($ 0.025*(v1) + 0.975*(v2) $);
    \draw[-> , thick] ($ 0.975*(v2) + 0.025*(v3) $) -- ($ 0.025*(v2) + 0.975*(v3) $);
    \draw[-> , thick] ($ 0.975*(v3) + 0.025*(v1) $) -- ($ 0.025*(v3) + 0.975*(v1) $);
    \draw[<- , thick] ($ 0.975*(w1) + 0.025*(w2) $) -- ($ 0.025*(w1) + 0.975*(w2) $);
    \draw[<- , thick] ($ 0.975*(w2) + 0.025*(w3) $) -- ($ 0.025*(w2) + 0.975*(w3) $);
    \draw[<- , thick] ($ 0.975*(w3) + 0.025*(w1) $) -- ($ 0.025*(w3) + 0.975*(w1) $);
    
    \draw[fill , darkblue] (c) circle [radius=0.03];
    \draw[fill , darkblue] (c') circle [radius=0.03]; 
    
    \draw[-> , thick , darkblue] (c') -- (0,-0.6,-1.5);
    \draw[-> , thick , darkblue] (c') -- (-1.5,1.5,-1.5);
    \draw[-> , thick , darkblue] (c') -- (1.5,1.5,-1.5);
    
    \draw[-> , thick , darkblue] (c) -- (0,-0.6,0);
    \draw[-> , thick , darkblue] (c) -- (-1.5,1.5,0);
    \draw[-> , thick , darkblue] (c) -- (1.5,1.5,0);
\end{tikzpicture} 

%% file: Pics/Tetrahedron.tex
\begin{tikzpicture}[scale=0.7 , rotate around y=5]
    \coordinate (v1) at (0,0,4);
    \coordinate (v2) at (7,0,0);
    \coordinate (v3) at (0,0,-4);
    \coordinate (v4) at ($ 1/3*(v1) + 1/3*(v2) + 1/3*(v3) + (0,6.5,0) $);
    
    \coordinate (v12) at ($ 1/2*(v1) + 1/2*(v2) $);
    \coordinate (v13) at ($ 1/2*(v1) + 1/2*(v3) $);
    \coordinate (v14) at ($ 1/2*(v1) + 1/2*(v4) $);
    \coordinate (v23) at ($ 1/2*(v2) + 1/2*(v3) $);
    \coordinate (v24) at ($ 1/2*(v2) + 1/2*(v4) $);
    \coordinate (v34) at ($ 1/2*(v3) + 1/2*(v4) $);
    
    \coordinate (c) at ($ 1/4*(v1) + 1/4*(v2) + 1/4*(v3) + 1/4*(v4)$);
    
    \coordinate (v123) at ($ 1/3*(v1) + 1/3*(v2) + 1/3*(v3) $);
    \coordinate (v124) at ($ 1/3*(v1) + 1/3*(v2) + 1/3*(v4) $);
    \coordinate (v134) at ($ 1/3*(v1) + 1/3*(v3) + 1/3*(v4) $);
    \coordinate (v234) at ($ 1/3*(v2) + 1/3*(v3) + 1/3*(v4) $);
    
    \draw[fill] (c) circle [radius=0.05];
    
    \draw[fill , darkblue] (v123) circle [radius=0.03];
    \draw[fill , darkblue] (v124) circle [radius=0.03];
    \draw[fill , darkblue] (v134) circle [radius=0.03];
    \draw[fill , darkblue] (v234) circle [radius=0.03];
    
%Internal Links
    \draw[- , thick , darkred] (c) -- (v124);
    \draw[- , dotted , thick , darkred] (c) -- (v123);
    \draw[- , dotted , thick , darkred] (c) -- (v134);
    \draw[- , thick , darkred] (c) -- (v234);
    
%Links
    \draw[- , dotted , thick , darkblue] (v123) -- ($ 0.01*(v123) + 0.99*(v12) $);
    \draw[- , dotted , thick , darkblue] (v123) -- ($ 0.01*(v123) + 0.99*(v13) $);
    \draw[- , dotted , thick , darkblue] (v123) -- ($ 0.01*(v123) + 0.99*(v23) $);
    
    \draw[- , dotted , thick , darkblue] (v124) -- ($ 0.01*(v124) + 0.99*(v12) $);
    \draw[- , dotted , thick , darkblue] (v124) -- ($ 0.01*(v124) + 0.99*(v14) $);
    \draw[- , thick , darkblue] (v124) -- ($ 0.01*(v124) + 0.99*(v24) $);
    
    \draw[- , dotted , thick , darkblue] (v134) -- ($ 0.01*(v134) + 0.99*(v13) $);
    \draw[- , dotted , thick , darkblue] (v134) -- ($ 0.01*(v134) + 0.99*(v14) $);
    \draw[- , dotted , thick , darkblue] (v134) -- ($ 0.01*(v134) + 0.99*(v34) $);
    
    \draw[- , dotted , thick , darkblue] (v234) -- ($ 0.01*(v234) + 0.99*(v23) $);
    \draw[- , thick , darkblue] (v234) -- ($ 0.01*(v234) + 0.99*(v24) $);
    \draw[- , dotted , thick , darkblue] (v234) -- ($ 0.01*(v234) + 0.99*(v34) $);
    
%Edges
    \draw[- , thick] (v1) -- (v2) -- (v4) -- cycle;
    \draw[- , thick] (v1) -- (v3);
    \draw[- , thick] (v3) -- (v2);
    \draw[- , thick] (v3) -- (v4);
\end{tikzpicture}

%% file: Pics/Pachner14.tex
\begin{tikzpicture}[scale=0.95]
    \coordinate (v1) at (0,0,5.2);
    \coordinate (v2) at (3,0,0);
    \coordinate (v3) at (-3,0,0);
    \coordinate (v4) at ($ 1/3*(v1) + 1/3*(v2) + 1/3*(v3) + (0,4.9,0) $);
    
    \coordinate (w1) at (10,0,5.2);
    \coordinate (w2) at (13,0,0);
    \coordinate (w3) at (7,0,0);
    \coordinate (w4) at ($ 1/3*(w1) + 1/3*(w2) + 1/3*(w3) + (0,4.9,0) $);

    \coordinate (v123) at ($ 1/3*(v1) + 1/3*(v2) + 1/3*(v3) $);
    \coordinate (v124) at ($ 1/3*(v1) + 1/3*(v2) + 1/3*(v4) $);
    \coordinate (v134) at ($ 1/3*(v1) + 1/3*(v3) + 1/3*(v4) $);
    \coordinate (v234) at ($ 1/3*(v2) + 1/3*(v3) + 1/3*(v4) $);
    
    \coordinate (w123) at ($ 1/3*(w1) + 1/3*(w2) + 1/3*(w3) $);
    \coordinate (w124) at ($ 1/3*(w1) + 1/3*(w2) + 1/3*(w4) $);
    \coordinate (w134) at ($ 1/3*(w1) + 1/3*(w3) + 1/3*(w4) $);
    \coordinate (w234) at ($ 1/3*(w2) + 1/3*(w3) + 1/3*(w4) $);
    
    \coordinate (w12c) at ($ 5/12*(w1) + 5/12*(w2) + 1/12*(w3) + 1/12*(w4) $);
    \coordinate (w13c) at ($ 5/12*(w1) + 1/12*(w2) + 5/12*(w3) + 1/12*(w4) $);
    \coordinate (w14c) at ($ 5/12*(w1) + 1/12*(w2) + 1/12*(w3) + 5/12*(w4) $);
    \coordinate (w23c) at ($ 1/12*(w1) + 5/12*(w2) + 5/12*(w3) + 1/12*(w4) $);
    \coordinate (w24c) at ($ 1/12*(w1) + 5/12*(w2) + 1/12*(w3) + 5/12*(w4) $);
    \coordinate (w34c) at ($ 1/12*(w1) + 1/12*(w2) + 5/12*(w3) + 5/12*(w4) $);
    
    \coordinate (c) at ($ 0.25*(v1) + 0.25*(v2) + 0.25*(v3) + 0.25*(v4) $);
    \coordinate (c') at ($ 0.25*(w1) + 0.25*(w2) + 0.25*(w3) + 0.25*(w4) $);
    \coordinate (c1) at ($ 0.25*(w1) + 0.25*(w2) + 0.25*(w3) + 0.25*(c') $);
    \coordinate (c2) at ($ 0.25*(w1) + 0.25*(w2) + 0.25*(w4) + 0.25*(c') $);
    \coordinate (c3) at ($ 0.25*(w1) + 0.25*(w3) + 0.25*(w4) + 0.25*(c') $);
    \coordinate (c4) at ($ 0.25*(w2) + 0.25*(w3) + 0.25*(w4) + 0.25*(c') $);
    
%1 Tetrahedron
    \draw[fill , darkblue] (c) circle [radius=0.04];
    
    \draw[- , thick] (v1) -- (v2) -- (v3) -- cycle;
    \draw[- , thick] (v1) -- (v4) -- (v2);
    \draw[- , thick] (v3) -- (v4);
    
    \draw[- , thick , darkblue] (c) -- ($ 1.25*(v123) - 0.25*(c) $);
    \draw[- , thick , darkblue] (c) -- ($ 1.25*(v124) - 0.25*(c) $);
    \draw[- , thick , darkblue] (c) -- ($ 1.25*(v134) - 0.25*(c) $);
    \draw[- , thick , darkblue] (c) -- ($ 1.25*(v234) - 0.25*(c) $);

    \draw[{Implies}-{Implies} , double distance=2] ($ 0.6*(c) + 0.4*(c') $) -- node[above] {$P_{(1,4)}$} ($ 0.4*(c) + 0.6*(c') $);

%4 Tetrahedra
    \draw[fill] (c') circle [radius=0.04];
    \draw[fill , darkblue] (c1) circle [radius=0.04];
    \draw[fill , darkred] (c2) circle [radius=0.04];
    \draw[fill , darkgreen] (c3) circle [radius=0.04];
    \draw[fill , cargreen] (c4) circle [radius=0.04];
    
    \draw[- , thick] (w1) -- (w2) -- (w3) -- cycle;
    \draw[- , thick] (w1) -- (w4) -- (w2);
    \draw[- , thick] (w3) -- (w4);
    \draw[- , thick] (w1) -- (c') -- (w2);
    \draw[- , thick] (w3) -- (c') -- (w4);
    
    \draw[- , thick , darkblue] (c1) -- ($ 1.25*(w123) - 0.25*(c') $);
    \draw[- , thick , darkblue] (c1) -- (w12c);
    \draw[- , thick , darkblue] (c1) -- (w13c);
    \draw[- , thick , darkblue] (c1) -- (w23c);
    
    \draw[- , thick , darkred] (c2) -- ($ 1.25*(w124) - 0.25*(c') $);
    \draw[- , thick , darkred] (c2) -- (w12c);
    \draw[- , thick , darkred] (c2) -- (w14c);
    \draw[- , thick , darkred] (c2) -- (w24c);
    
    \draw[- , thick , darkgreen] (c3) -- ($ 1.25*(w134) - 0.25*(c') $);
    \draw[- , thick , darkgreen] (c3) -- (w13c);
    \draw[- , thick , darkgreen] (c3) -- (w14c);
    \draw[- , thick , darkgreen] (c3) -- (w34c);
    
    \draw[- , thick , cargreen] (c4) -- ($ 1.25*(w234) - 0.25*(c') $);
    \draw[- , thick , cargreen] (c4) -- (w23c);
    \draw[- , thick , cargreen] (c4) -- (w24c);
    \draw[- , thick , cargreen] (c4) -- (w34c);
\end{tikzpicture}

%% file: Pics/Pachner23.tex
\begin{tikzpicture}[scale=0.8]
    \coordinate (v1) at (0,0,5.2);
    \coordinate (v2) at (3,0,0);
    \coordinate (v3) at (-3,0,0);
    \coordinate (v4) at ($ 1/3*(v1) + 1/3*(v2) + 1/3*(v3) + (0,4.9,0) $);
    \coordinate (v5) at ($ 1/3*(v1) + 1/3*(v2) + 1/3*(v3) + (0,-4.9,0) $);
    
    \coordinate (w1) at (10,0,5.2);
    \coordinate (w2) at (13,0,0);
    \coordinate (w3) at (7,0,0);
    \coordinate (w4) at ($ 1/3*(w1) + 1/3*(w2) + 1/3*(w3) + (0,4.9,0) $);
    \coordinate (w5) at ($ 1/3*(w1) + 1/3*(w2) + 1/3*(w3) + (0,-4.9,0) $);
    
    \coordinate (v123) at ($ 1/3*(v1) + 1/3*(v2) + 1/3*(v3) $);
    \coordinate (v124) at ($ 1/3*(v1) + 1/3*(v2) + 1/3*(v4) $);
    \coordinate (v134) at ($ 1/3*(v1) + 1/3*(v3) + 1/3*(v4) $);
    \coordinate (v234) at ($ 1/3*(v2) + 1/3*(v3) + 1/3*(v4) $);
    \coordinate (v125) at ($ 1/3*(v1) + 1/3*(v2) + 1/3*(v5) $);
    \coordinate (v135) at ($ 1/3*(v1) + 1/3*(v3) + 1/3*(v5) $);
    \coordinate (v235) at ($ 1/3*(v2) + 1/3*(v3) + 1/3*(v5) $);
    
    \coordinate (w123) at ($ 1/3*(w1) + 1/3*(w2) + 1/3*(w3) $);
    \coordinate (w124) at ($ 1/3*(w1) + 1/3*(w2) + 1/3*(w4) $);
    \coordinate (w134) at ($ 1/3*(w1) + 1/3*(w3) + 1/3*(w4) $);
    \coordinate (w234) at ($ 1/3*(w2) + 1/3*(w3) + 1/3*(w4) $);
    \coordinate (w125) at ($ 1/3*(w1) + 1/3*(w2) + 1/3*(w5) $);
    \coordinate (w135) at ($ 1/3*(w1) + 1/3*(w3) + 1/3*(w5) $);
    \coordinate (w235) at ($ 1/3*(w2) + 1/3*(w3) + 1/3*(w5) $);
    
    \coordinate (w145) at ($ 1/3*(w1) + 1/3*(w4) + 1/3*(w5) $);
    \coordinate (w245) at ($ 1/3*(w2) + 1/3*(w4) + 1/3*(w5) $);
    \coordinate (w345) at ($ 1/3*(w3) + 1/3*(w4) + 1/3*(w5) $);
    
    \coordinate (c1) at ($ 0.25*(v1) + 0.25*(v2) + 0.25*(v3) + 0.25*(v4) $);
    \coordinate (c2) at ($ 0.25*(v1) + 0.25*(v2) + 0.25*(v3) + 0.25*(v5) $);
    \coordinate (c1') at ($ 0.25*(w1) + 0.25*(w2) + 0.25*(w3) + 0.25*(w4) $);
    \coordinate (c2') at ($ 0.25*(w1) + 0.25*(w2) + 0.25*(w3) + 0.25*(w5) $);
    \coordinate (cc1') at ($ 0.25*(w1) + 0.25*(w2) + 0.25*(w4) + 0.25*(w5) $);
    \coordinate (cc2') at ($ 0.25*(w1) + 0.25*(w3) + 0.25*(w4) + 0.25*(w5) $);
    \coordinate (cc3') at ($ 0.25*(w2) + 0.25*(w3) + 0.25*(w4) + 0.25*(w5) $);

    \draw[{Implies}-{Implies} , double distance=2] ($ 0.6*(v123) + 0.4*(w123) $) -- node[above] {$P_{(2,3)}$} ($ 0.4*(v123) + 0.6*(w123) $);

%2 Tetrahedra
    \draw[fill , darkblue] (c1) circle [radius=0.04];
    \draw[fill , darkblue] (c2) circle [radius=0.04];
    \draw[- , thick , darkblue] (c1) -- (c2);
    \draw[- , thick , darkblue] (c1) -- ($ 1.25*(v124) - 0.25*(c1) $);
    \draw[- , thick , darkblue] (c1) -- ($ 1.25*(v134) - 0.25*(c1) $);
    \draw[- , thick , darkblue] (c1) -- ($ 1.25*(v234) - 0.25*(c1) $);
    \draw[- , thick , darkblue] (c2) -- ($ 1.25*(v125) - 0.25*(c2) $);
    \draw[- , thick , darkblue] (c2) -- ($ 1.25*(v135) - 0.25*(c2) $);
    \draw[- , thick , darkblue] (c2) -- ($ 1.25*(v235) - 0.25*(c2) $);
    
    \draw[- , thick] (v1) -- (v2) -- (v3) -- cycle;
    \draw[- , thick] (v1) -- (v4) -- (v2) -- (v5) -- cycle;
    \draw[- , thick] (v4) -- (v3) -- (v5);
    
%3 Tetrahedra
    \draw[- , thick , darkblue] (cc1') -- ($ 1.25*(w124) - 0.25*(cc1') $);
    \draw[- , thick , darkblue] (cc1') -- ($ 1.25*(w125) - 0.25*(cc1') $);
    \draw[- , thick , darkblue] (cc1') -- (w145);
    \draw[- , thick , darkblue] (cc1') -- (w245);
    
    \draw[- , thick , darkred] (cc2') -- ($ 1.25*(w134) - 0.25*(cc2') $);
    \draw[- , thick , darkred] (cc2') -- ($ 1.25*(w135) - 0.25*(cc2') $);
    \draw[- , thick , darkred] (cc2') -- (w145);
    \draw[- , thick , darkred] (cc2') -- (w345);
    
    \draw[- , thick , darkgreen] (cc3') -- ($ 1.25*(w234) - 0.25*(cc3') $);
    \draw[- , thick , darkgreen] (cc3') -- ($ 1.25*(w235) - 0.25*(cc3') $);
    \draw[- , thick , darkgreen] (cc3') -- (w245);
    \draw[- , thick , darkgreen] (cc3') -- (w345);
    
    \draw[- , thick] (w1) -- (w2) -- (w3) -- cycle;
    \draw[- , thick] (w1) -- (w4) -- (w2) -- (w5) -- cycle;
    \draw[- , thick] (w4) -- (w3) -- (w5);
    \draw[- , thick] (w4) -- (w5);
\end{tikzpicture}

%% file: AppProofs.tex
\subsection{Hopf algebra proofs}
\label{App_HopfAlgebraProofs}

Here we provide some relevant proofs regarding the bi-algebra notions introduced in Sec. \ref{Sec_HopfAlgebra_Intro}. \\
\medskip \\
\textbf{Inverse of the Fourier transform: Prop. \ref{Prop_Fourier}.} \\
We prove Prop. \ref{Prop_Fourier}, which states that the map \eqref{Inverse_FourierTransform} is the inverse of the Fourier transform \eqref{FourierTransform}. 
Here we first use both the definitions of Fourier transform and inverse
\be
    (\mathcal{F} \circ \mathcal{F}^{-1}) [\hat{\Phi}]
    =
    \frac{1}{\sqrt{\mu}} \left( \int_{H}^L \ot \, \id \right) \,
    \big(\sigma \cdot (\mathcal{F}^{-1} [\hat{\Phi}] \ot 1)\big)
    = 
    \frac{1}{\mu} \left( \int_{H}^L \ot \int_{A}^L \ot \id \right) \,
    \big(\sigma_{13} \cdot \sigma_{12}\mone \cdot
    (1 \ot \hat{\Phi} \ot 1)\big) \,.
\ee
Then, using the definition of inverse skew co-pairing \eqref{InvSkewEl}, the expression becomes
\be
    (\mathcal{F} \circ \mathcal{F}^{-1}) [\hat{\Phi}]
    =
    \frac{1}{\mu} \left( \int_{H}^L \ot \int_{A}^L \ot \id \right) \,
    (\id \ot S \ot \id) \, \big((1 \ot S\mone\hat{\Phi} \ot 1) \cdot (\sigma_{13} \cdot \sigma_{12})\big) \,.
\ee
Since the co-integral of $H$ of the product of co-pairing $\sigma_{13} \cdot \sigma_{12}$ can be written as the co-product of the delta function of $A$,
\be
    \frac{1}{\sqrt{\mu}} \left( \int_{H}^L \ot \id \ot \id \right) \,
    (\sigma_{13} \cdot \sigma_{12}) =
    \frac{1}{\sqrt{\mu}} \left(\int_H^L \ot \Delta\right) \,
    \sigma =
    \Delta \dela \,,
\ee
one can use the first of the identities \eqref{Delta_CoProduct_Property} to reduce the expression as
\be
    (\mathcal{F} \circ \mathcal{F}^{-1}) [\hat{\Phi}]
    =
    \frac{1}{\sqrt{\mu}} \left( \int_{A}^L \ot \id \right) \,
    (S \ot \id) \, \big((S\mone \hat{\Phi} \ot 1) \cdot \Delta \dela\big)
    =
    \frac{1}{\sqrt{\mu}} \left( \int_{A}^L \ot \id \right) \,
    (S \ot \id) \, \big((1 \ot \hat{\Phi}) \cdot \Delta \dela\big) \,.
\ee
Last, since the antipode is an anti-homomorphism map, is passes through the co-product, which becomes $\Delta_{21}$. This allows us to use the left invariance property of the co-integral and compute it using the normalization property of the delta function \eqref{Delta_Normalization}. This closes the proof of the proposition.
\be
    (\mathcal{F} \circ \mathcal{F}^{-1}) [\hat{\Phi}]
    =
    \frac{1}{\sqrt{\mu}} \left( \int_{A}^L \ot \id \right) \,
    \big(\Delta_{-21} (S\dela) \cdot (1 \ot \hat{\Phi})\big)
    =
    \frac{1}{\sqrt{\mu}} \left( \int_{A}^L \ot \id \right) \,
    \big((S\dela \ot 1) \cdot (1 \ot \hat{\Phi})\big)
    =
    \hat{\Phi} \,.
\ee
Similarly, we can also prove the identity $(\mathcal{F}^{-1} \circ \mathcal{F}) = \id$. \\
\medskip \\
\textbf{Delta function of an Hopf algebra: Prop. \ref{Prop_Delta}.} \\
% \label{delta}
Let us prove that the delta function of the algebra $A$, defined in \eqref{Delta},
satisfies the identity \eqref{Delta_CoProduct_Property} of Prop. \ref{Prop_Delta}. We use first the definition of delta function and that of Fourier transform, plus the property \eqref{SkewEl_id-CoProduct} of the skew co-pairing element.
\begin{align}
    (\hat{\Phi} \ot 1) \cdot \Delta \dela
    & =
    \frac{1}{\sqrt{\mu}} \left( \int_{H}^L \ot \id \ot \id \right) \,
    \big((1 \ot \hat{\Phi} \ot 1) \cdot (\id \ot \Delta) \, \sigma\big)
    =
    \frac{1}{\sqrt{\mu}} \left( \int_{H}^L \ot \id \ot \id \right) \,
    \big((1 \ot \hat{\Phi} \ot 1) \cdot \sigma_{13} \cdot \sigma_{12}\big)
    \nonumber \\
    & =
    \frac{1}{\mu} \left( \int_{H}^L \ot \int_{H}^L \ot \id \ot \id \right) \,
    \big(\sigma_{13} \cdot \sigma_{24} \cdot \sigma_{23} \cdot (\Phi \ot 1 \ot 1 \ot 1)\big) \,.
\end{align}
Now multiply the skew co-pairings from the left by $\sigma_{14}\mone \cdot \sigma_{14}$ and use the property \eqref{SkewEl_CoProduct-id}, plus the left invariance of the co-integral on $H$.
\begin{align}
    (\hat{\Phi} \ot 1) \cdot \Delta \dela
    & =
    \frac{1}{\mu} \left( \int_{H}^L \ot \int_{H}^L \ot \id \ot \id \right) \,
    \big(\sigma_{14}\mone \cdot \sigma_{14} \cdot \sigma_{24} \cdot \sigma_{13} \cdot \sigma_{23} \cdot (\Phi \ot 1 \ot 1 \ot 1)\big)
    \nonumber \\
    & =
    \frac{1}{\mu} \left( \int_{H}^L \ot \int_{H}^L \ot \id \ot \id \right) \,
    \big(\sigma_{14}\mone \cdot (\Delta \ot \id \ot \id) \, (\sigma_{12} \cdot \sigma_{13}) \cdot (\Phi \ot 1 \ot 1 \ot 1)\big)
    \nonumber \\
    & =
    \frac{1}{\mu} \left( \int_{H}^L \ot \int_{H}^L \ot \id \ot \id \right) \,
    \big(\sigma_{14}\mone \cdot \sigma_{23} \cdot \sigma_{24} \cdot (\Phi \ot 1 \ot 1 \ot 1)\big) \,.
\end{align}
Last, using the definitions of inverse skew co-pairing \eqref{InvSkewEl}, that of Fourier transform and delta function, we prove the first part of the proposition.
\begin{align}
    (\hat{\Phi} \ot 1) \cdot \Delta \dela
    & =
    \frac{1}{\mu} \left( \int_{H}^L \ot \int_{H}^L \ot \id \ot \id \right) \,
    \big((\id \ot \id \ot \id \ot S) \, \sigma_{14} \cdot (\Phi \ot 1 \ot 1 \ot 1) \cdot (\id \ot \id \ot \Delta) \, \sigma_{23}\big)
    \nonumber \\
    & =
    \frac{1}{\sqrt{\mu}} \left( \int_{H}^L \ot \id \ot \id \right) \,
    \big((\id \ot \id \ot S) \, \sigma_{13} \cdot (\Phi \ot 1 \ot 1)\big)
    \, \cdot \, 
    (\id \ot \Delta) \, \dela
    \nonumber \\
    & =
    (1 \ot S \hat{\Phi}) \cdot \Delta \dela \,.
\end{align}
The proofs of the other three identities of the first part of Prop. \ref{Prop_Delta} are similar.
\smallskip \\
The proof of normalization of the delta function \eqref{Delta_Normalization}, in the second part of the proposition, is straightforward. \\
\medskip \\
\textbf{Delta function of an Hopf algebra as integral in it: Prop. \ref{Prop_DeltaIntegral}.} \\
We show that the delta function in the Hopf algebra $A$ is a left integral in it, as stated in Prop. \ref{Prop_DeltaIntegral}. First note that the co-unit of an element of $A$ coincides with the co-integral of its dual:
\be
    \varepsilon \hat{\Phi}
    =
    \frac{1}{\sqrt{\mu}} \left(\int_H^L \ot \varepsilon\right) \,  \big(\sigma \cdot (\Phi \ot 1)\big)
    =
    \frac{1}{\sqrt{\mu}} \, \int_H^L \, \Phi \,.
\ee
Here we simply used the definition of Fourier transform \eqref{FourierTransform} and the unit property \eqref{SkewEl_id-CoProduct} of the skew co-pairing element.   
Now take the product $\hat{\Phi} \cdot \dela$, use the definition of delta function \eqref{Delta} and that of Fourier transform.
\be
    \hat{\Phi} \cdot \dela
    =
    \frac{1}{\sqrt{\mu}} \left( \int_{H}^L \ot \id \right) \,
    \big((1 \ot \hat{\Phi}) \cdot \sigma\big)
    =
    \frac{1}{\mu} \left( \int_{H}^L \ot \int_{H}^L \ot \id \right) \,
    \big(\sigma_{13} \cdot \sigma_{23} \cdot (\Phi \ot 1 \ot 1)\big) \,.
\ee
Then, using the property \eqref{SkewEl_CoProduct-id} and the left invariance of the co-integral on $H$, we prove the identity.
\begin{align}
    \hat{\Phi} \cdot \dela
    & =
    \frac{1}{\mu} \left( \int_{H}^L \ot \int_{H}^L \ot \id \right) \,
    \big((\Delta \ot \id) \sigma \cdot (\Phi \ot 1 \ot 1)\big)
    =
    \frac{1}{\mu} \left( \int_{H}^L \ot \int_{H}^L \ot \id \right) \,
    \big(\sigma_{23} \cdot (\Phi \ot 1 \ot 1)\big)
    \nonumber \\
    & =
    \frac{1}{\sqrt{\mu}} \, \dela \int_H \Phi
    =
    \dela \, \varepsilon(\hat{\Phi}) .
\end{align}
Similarly, one can prove that $\delh$ is the left integral in $H$ and that $\dela^{-1}$ and $\delh^{-1}$ defined in \eqref{Delta_Opposite} are the right integrals respectively in $A$ and $H$.
\medskip \\
Let us provide a further property of the skew co-pairing element.
\begin{proposition}[Integral property of the skew co-pairing element]
\label{Prop_SkewEl}
    Let $H$ and $A$ be skew co-paired bi-algebras with invertible skew co-pairing. It satisfies the identities below.
    \be
        \begin{aligned}
            &
            \left(\int_H \ot \id\right) \, \sigma = 
            \left(\int_H \ot \id \ot \id\right) \, (\sigma^{\pm 1}_{23} \cdot \sigma_{13})
            \,,\\
            &
            \left(\int_H \ot \id\right) \, \sigma\mone = 
            \left(\int_H \ot \id \ot \id\right) \, (\sigma\mone_{13} \cdot \sigma^{\pm 1}_{23}) \,,
        \end{aligned}
        \qquad
        \begin{aligned}
            &
            \left(\id \ot \int_A\right) \, \sigma = 
            \left(\id \ot \id \ot \int_A\right) \, (\sigma_{13} \cdot \sigma^{\pm 1}_{12})
            \,,\\
            &
            \left(\id \ot \int_A\right) \, \sigma\mone = 
            \left(\id \ot \id \ot \int_A\right) \, (\sigma^{\pm 1}_{12} \cdot \sigma\mone_{13}) \,.
        \end{aligned}
        \label{SkewEl_ConvInv_CheatProperty}
    \ee
\end{proposition}
\begin{proof}
    Consider the first of the identities \eqref{SkewEl_ConvInv_CheatProperty}. We prove it simply using the property \eqref{SkewEl_CoProduct-id} and the left invariance of the co-integral on $H$.
    \begin{align}
        \left(\int_H \ot \id\right) \, \sigma 
        & =
        \left(\int_H \ot \id \ot \id\right) \, (\sigma_{23}\mone \cdot \sigma_{23} \cdot \sigma_{13})
        \nonumber \\
        & =
        \left(\int_H \ot \id \ot \id\right) \, (\sigma_{23}\mone \cdot (\Delta_{2 1} \ot \id) \, \sigma)
        =
        \left(\int_H \ot \id \ot \id\right) \, (\sigma_{23}\mone \cdot \sigma_{13}) \,.
    \end{align}
    Alternatively, using the same properties, we have
    \begin{align}
        \left(\int_H \ot \id\right) \, \sigma 
        & =
        \left(\int_H \ot \id \ot \id\right) \, (\sigma_{23} \cdot \sigma\mone_{23} \cdot \sigma_{13})
        =
        \left(\int_H \ot \id \ot \id\right) \, (\sigma_{23} \cdot (\id \ot S\mone \ot \id) \, (\sigma_{23} \cdot \sigma_{13}))
        \nonumber \\
        & =
        \left(\int_H \ot \id \ot \id\right) \, (\sigma_{23} \cdot (\id \ot S\mone \ot \id) \circ (\Delta_{21} \ot \id) \, \sigma)
        =
        \left(\int_H \ot \id \ot \id\right) \, (\sigma_{23} \cdot \sigma_{13}) \,.
    \end{align}
    The other three identities in \eqref{SkewEl_ConvInv_CheatProperty} can be proved in a similar way.
\end{proof}
\noindent
Note that, by direct application of the above proposition, it is straightforward to derive the further identities
\be
    \left(\int_H \ot \id\right) \, \sigma =
    \left(\int_H \ot \id\right) \, \sigma\mone
    \,\,,\quad
    \left(\id \ot \int_A\right) \, \sigma =
    \left(\id \ot \int_A\right) \, \sigma\mone \,.
\ee

\subsection{Geometric computations}
\label{App_GeometricComputations}

Let us provide here some useful results which will be used later in this section. The computations below are strictly related to the geometric interpretation of Hopf algebra field theory, as a path integral formulation of a model of discrete curved geometries. 
\smallskip \\
We remind that the closure constraint \eqref{ClosureConstraint} represents the closure of the boundary of a triangle. 
Let us use the tensor notation $\hat{\cC} = \hat{\cC}_{1 \, 2 \, 3 } \in A^3$.
We list below some propositions involving the geometric operations on the closure constraint that we will need in Hopf algebra field theory, such as the cyclic permutation of the edges, inversion or gluing of two triangles.
\begin{proposition}[Cyclic permutation and inversion properties of a the closure constraint]
\label{Prop_ClosureConstraint_Properties}
    The closure constraint \eqref{ClosureConstraint} is invariant under cyclic permutation of its elements and under inversion.
    \be
        \begin{aligned}
            &
            \hat{\cC}_{1 \, 2 \, 3} =
            \hat{\cC}_{-3 \, -2 \, -1}
            \,,\\
            &
            \hat{\cC}_{1 \, 2 \, 3} =
            \hat{\cC}_{2 \, 3 \, 1} \,.    
        \end{aligned}
    \label{ClosureConstraint_CyclicPermutation+Inversion}
    \ee
    Moreover, the closure constraint is  a projector,
    \be
        \hat{\cC} \cdot \hat{\cC} = \hat{\cC} \,. 
    \label{ClosureConstraint_Reduction}
    \ee
\end{proposition}
\begin{proof}
    Let us prove first the invariance under the inversion. 
    We carry out the computation just using the definition of closure constraint \eqref{ClosureConstraint} and that of opposite delta function \eqref{Delta_Opposite}, plus the left invariance of the co-integral on $H$.
    \begin{align}
        \hat{\cC} 
        & =
        \frac{1}{V_H} \left(\int_H \ot \idd \right) \,
        \big((\id \ot \Delta^3) \, \sigma\mone\big)
        = \frac{1}{V_H} \left(\int_H \ot \idd\right) \,
        \big(\sigma\mone_{12} \sigma\mone_{13} \sigma\mone_{14}\big)
        \nonumber \\
        & =
        \frac{1}{V_H^2} \left(\int_{H^2} \ot \idd\right) \,
        \big(
        (\sigma\mone_{13} \sigma\mone_{14} \sigma\mone_{15}) \cdot
        (\sigma\mone_{23} \sigma\mone_{24} \sigma\mone_{2 5}) \cdot
        (\sigma_{2 5} \sigma_{24} \sigma_{23})\big)
        \qquad \text{\tiny{use \eqref{SkewEl_CoProduct-id} and the left invariance}}
        \nonumber \\
        % & =
        % \frac{1}{V_H^2} \left(\int_{H^2} \ot \idd\right) \,
        % \big(
        % (\Delta_{21} \ot \idd) \, ((\sigma\mone \ot 1 \ot 1) \cdot (\sigma\mone \ot 1) \cdot \sigma\mone) \cdot 
        % (\sigma_{2 5} \sigma\mone_{24} \sigma_{23})\big)
        % \nonumber \\
        & =
        \frac{1}{V_H^2} \left(\int_{H^2} \ot \idd\right) \,
        \big(
        (\sigma\mone_{13} \sigma\mone_{14} \sigma\mone_{1 5}) \cdot 
        (\sigma_{2 5} \sigma_{24} \sigma_{23})\big)
        \qquad \text{\tiny{use \eqref{SkewEl_CoProduct-id} with the antipode and the left invariance}}
        \nonumber \\
        % & =
        % \frac{1}{V_H^2} \left(\int_{H^2} \ot \idd\right) \,
        % \big(
        % (S \ot \id^{\ot 4}) \circ (\Delta \ot \idd) \, 
        % (\sigma \cdot (\sigma \ot 1) \cdot (\sigma \ot 1 \ot 1)\big)
        % \nonumber \\
        & =
        \frac{1}{V_H^2} \left(\int_{H^2} \ot \idd\right) \,
        \big(\sigma_{2 5} \sigma_{24} \sigma_{23}\big)
        =
        \frac{1}{V_H} \left(\int_{H} \ot \idd\right) \,
        \big(\sigma_{1 4} \sigma_{13} \sigma_{12}\big) 
        =
        S^3 \circ \tau^3 \, \hat{\cC} = \hat{\cC}_{-3 \, -2 \, -1} \,.
    \end{align}
    We can prove in a similar way the invariance under cyclic permutations by multiplying the closure constraint by the properly chosen identity. 

    The fact that the product of two closure constraints reduces to a single one, eq. \eqref{ClosureConstraint_Reduction}, is expected since the closure constraint is the Fourier transform of a projector. We carry out the proof writing the closure constraints in terms of skew co-pairing elements and then using the left invariance of the co-integral on $H$.
    \begin{align}
        \hat{\cC} \cdot \hat{\cC} 
        & =
        \frac{\mu}{V_H^2} \big(\Delta^3 \dela\mone \cdot \Delta^3 \dela\mone\big)
        \nonumber \\
        & =
        \frac{1}{V_H^2} \left(\int_{H^2} \ot \id^{\ot 3}\right) \,
        \big((\sigma\mone_{13} \sigma\mone_{14} \sigma\mone_{15}) \cdot
        (\sigma\mone_{23} \sigma\mone_{24} \sigma\mone_{25})\big)
        \qquad \text{\tiny{use \eqref{SkewEl_CoProduct-id_inv} and the left invariance}}
        \nonumber \\
        % & =
        % \frac{1}{V_H^2} \left(\int_{H^2} \ot \idd\right) \,
        % \big(
        % (\Delta \ot \idd) \, ((\sigma\mone \ot 1 \ot 1) \cdot (\sigma\mone \ot 1) \cdot \sigma\mone)\big)
        % \nonumber \\
        & =
        \frac{1}{V_H^2} \left(\int_{H^2} \ot \idd\right) \,
        \big(\sigma\mone_{13} \sigma\mone_{14} \sigma\mone_{15}\big)
        =
        \frac{1}{V_H} \left(\int_{H} \ot \idd\right) \,
        \big(\sigma\mone_{12} \sigma\mone_{14} \sigma\mone_{15}\big)
        =
        \hat{\cC} \,.
    \end{align}
\end{proof}
\medskip
\noindent
We now consider the kernel of the vertex amplitude \eqref{VertexAmplitude} associated to a tetrahedron and illustrate how to reduce its expression. The reduction consists in \textit{gluing} the closure constraints in it, by evaluating the co-product $\uDe$. Geometrically, it amounts to identifying the edges of the four triangles pairwise.
\begin{proposition}[Gluing triangles]
\label{Prop_ClosureConstraint_Gluing}
    Consider the kernel \eqref{VertexAmplitude} and two closure constraints \eqref{ClosureConstraint} in it. 
    Consider also the co-product $\uDe$, given in \eqref{Underline_CoProduct_3d} (which is a combination of six co-products), and take the co-product $\Delta_{ij}$ shared by the two closure constraints.
    The gluing of the two triangles is performed by taking the co-integral on the $j^{th}$ tensor space:
    \be
        \int_{A_{j_1}} \,
        \big(\Delta_{i_1 \, j_1} \dela\mone \cdot 
        \cC_{i_1 \, i_2 \, i_3} \cdot 
        \cC_{j_1 \, j_2 \, j_3}\big) 
        =
        \sqrt{\mu} \,
        \cC_{-i_1 \, j_2 \, j_3} \cdot 
        \cC_{i_1 \, i_2 \, i_3} \,.
        \label{ClosureConstraint_Gluing}
    \ee
    Similarly, the following identities hold.
    \be
        \begin{aligned}
            &
            \int_{A_{j_1}} \,
            \big(\Delta_{-i_1 \, j_1} \dela\mone \cdot 
            \cC_{i_1 \, i_2 \, i_3} \cdot 
            \cC_{j_1 \, j_2 \, j_3}\big) 
            =
            \sqrt{\mu} \,
            \cC_{i_1 \, i_2 \, i_3} \cdot 
            \cC_{i_1 \, j_2 \, j_3} 
            \,,\\
            &
            \int_{A_{j_1}} \,
            \big(\cC_{i_1 \, i_2 \, i_3} \cdot 
            \cC_{j_1 \, j_2 \, j_3} \cdot 
            \Delta_{i_1 \, -j_1} \dela\mone\big) 
            =
            \sqrt{\mu} \,
            \cC_{i_1 \, i_2 \, i_3} \cdot 
            \cC_{i_1 \, j_2 \, j_3} 
            \,,\\
            &
            \int_{A_{j_1}} \,
            \big(\cC_{i_1 \, i_2 \, i_3} \cdot 
            \cC_{j_1 \, j_2 \, j_3} \cdot 
            \Delta_{-i_1 \, -j_1} \dela\mone\big) 
            =
            \sqrt{\mu} \,
            \cC_{-i_1 \, j_2 \, j_3} \cdot
            \cC_{i_1 \, i_2 \, i_3}
            \,.
        \end{aligned}
        \label{ClosureConstraint_Gluing_Antipode}
    \ee
\end{proposition}
\begin{proof}
    In order to prove the proposition we first use the property \eqref{Delta_CoProduct_Property} for the co-product of the opposite delta function, moving the first component of the second closure constraint from the $j_1^{th}$ to the $i_1^{th}$ tensor space.
    \be
        \int_{A_{j_1}} \,
        \big(\Delta_{i_1 \, j_1} \dela\mone \cdot 
        \cC_{i_1 \, i_2 \, i_3} \cdot 
        \cC_{j_1 \, j_2 \, j_3}\big) 
        =
        \int_{A_{j_1}} \,
        \big(\Delta_{i_1 \, j_1} \dela\mone \cdot 
        \cC_{-i_1 \, j_2 \, j_3} \cdot 
        \cC_{i_1 \, i_2 \, i_3}\big) \,.
    \ee
    Then use the left invariance of the co-integral and the normalization property \eqref{Delta_Normalization} to remove the opposite delta function. 
    \be
        \int_{A_{j_1}} \,
        \big(\Delta_{i_1 \, j_1} \dela\mone \cdot 
        \cC_{i_1 \, i_2 \, i_3} \cdot 
        \cC_{j_1 \, j_2 \, j_3}\big) 
        =
        \left(\int_{A} \, \dela\mone\right) \, 
        \big(\cC_{-i_1 \, j_2 \, j_3} \cdot 
        \cC_{i_1 \, i_2 \, i_3}\big) 
        =
        \sqrt{\mu} \, 
        \cC_{-i_1 \, j_2 \, j_3} \cdot 
        \cC_{i_1 \, i_2 \, i_3} \,.
    \ee
    The proof of the other three identities \eqref{ClosureConstraint_Gluing_Antipode} is completely analogous.
\end{proof}
\medskip
\noindent
Let us now illustrate how to glue two tetrahedra (amplitude \eqref{VertexAmplitude}) by identifying their faces (one closure constraint for each amplitude). 
\begin{proposition}[Gluing two tetrahedra]
\label{Prop_Tetrahedra_Gluing}
    Consider the tensor product of a pair of independent tetrahedron amplitudes \eqref{GeneratingFunction_1-Tetra} obtained by the use of Prop. \ref{Prop_ClosureConstraint_Gluing}; take their product with a propagator amplitude \eqref{PropagatorAmplitude}. 
    The gluing of two tetrahedra is performed by taking the co-integral of such product:
    \begin{align}
        \int_{A^{3}_{j_2 \, j_4 \, j_6}} \,
        &
        (\hat{\cC}_{-i_6 -i_4 -i_2} \cdot 
        \hat{\cC}_{-i_5 i_6 -i_1} \cdot 
        \hat{\cC}_{-i_3 i_4 i_5} \cdot
        \hat{\cC}_{i_1 i_2 i_3}) \cdot 
        \Delta^{(3)}_{i_1 i_2 i_3 \,\, j_4 j_6 j_2} \dela\mone \cdot
        (\hat{\cC}_{-j_6 -j_4 -j_2} \cdot 
        \hat{\cC}_{-j_5 j_6 -j_1} \cdot 
        \hat{\cC}_{-j_3 j_4 j_5} \cdot
        \hat{\cC}_{j_1 j_2 j_3})
        \nonumber \\
        & =
        \sqrt{\mu^3} \,
        (\hat{\cC}_{-i_6 -i_4 -i_2} \cdot 
        \hat{\cC}_{-i_5 i_6 -i_1} \cdot 
        \hat{\cC}_{-i_3 i_4 i_5} \cdot
        \hat{\cC}_{i_1 i_2 i_3}) \cdot
        (\hat{\cC}_{-j_5 -i_2 -j_1} \cdot 
        \hat{\cC}_{-j_3 -i_1 j_5} \cdot
        \hat{\cC}_{-i_3 j_3 j_1}) \,.
        \label{VertexAmplitude_Gluing}
    \end{align}
    As in Prop. \ref{Prop_ClosureConstraint_Gluing}, there are other identities similar to \eqref{VertexAmplitude_Gluing}, involving antipodes or different combinatorics, whose derivation is completely analogous to that of the above identity.
\end{proposition}
\begin{proof}
    The proof of the proposition follows the same pattern of Prop. \ref{Prop_ClosureConstraint_Gluing}: we use the property \eqref{Delta_CoProduct_Property} to move the components of the second amplitude from the $j_1 j_2 j_3$ to the $i_1 i_2 i_3$ tensor spaces, so that we can use the left invariance of the co-integral and the normalization property \eqref{Delta_Normalization} to remove the propagator amplitude. Note that, for non-commutativity reasons, in order to correctly use the property \eqref{Delta_CoProduct_Property}, some of the closure constraints of the second tetrahedron amplitude have to be permuted or inverted using Prop. \ref{Prop_ClosureConstraint_Properties}.
    \begin{align}
        \int_{A^{3}_{j_2 \, j_4 \, j_6}} \,
        &
        (\hat{\cC}_{-i_6 -i_4 -i_2} \cdot 
        \hat{\cC}_{-i_5 i_6 -i_1} \cdot 
        \hat{\cC}_{-i_3 i_4 i_5} \cdot
        \hat{\cC}_{i_1 i_2 i_3}) \cdot 
        \Delta^{(3)}_{i_1 i_2 i_3 \,\, j_4 j_6 j_2} \dela\mone \cdot
        (\hat{\cC}_{-j_6 -j_4 -j_2} \cdot 
        \hat{\cC}_{-j_5 j_6 -j_1} \cdot 
        \hat{\cC}_{-j_3 j_4 j_5} \cdot
        \hat{\cC}_{j_1 j_2 j_3})
        \nonumber \\
        & =
        \int_{A^{3}_{j_2 \, j_4 \, j_6}} \,
        (\hat{\cC}_{-i_6 -i_4 -i_2} \cdot 
        \hat{\cC}_{-i_5 i_6 -i_1} \cdot 
        \hat{\cC}_{-i_3 i_4 i_5} \cdot
        \hat{\cC}_{i_1 i_2 i_3}) \cdot 
        \Delta^{(3)}_{i_1 i_2 i_3 \,\, j_4 j_6 j_2} \dela\mone \cdot
        (\hat{\cC}_{j_4 j_6 j_2} \cdot 
        \hat{\cC}_{-j_5 j_6 -j_1} \cdot 
        \hat{\cC}_{-j_3 j_4 j_5} \cdot
        \hat{\cC}_{j_2 j_3 j_1})
        \nonumber \\
        & =
        \int_{A^{3}_{j_2 \, j_4 \, j_6}} \,
        (\hat{\cC}_{-i_6 -i_4 -i_2} \cdot 
        \hat{\cC}_{-i_5 i_6 -i_1} \cdot 
        \hat{\cC}_{-i_3 i_4 i_5} \cdot
        \hat{\cC}_{i_1 i_2 i_3}) \cdot 
        \Delta^{(3)}_{i_1 i_2 i_3 \,\, j_4 j_6 j_2} \dela\mone \cdot
        (\hat{\cC}_{-i_1 -i_2 -i_3} \cdot 
        \hat{\cC}_{-j_5 -i_2 -j_1} \cdot 
        \hat{\cC}_{-j_3 -i_1 j_5} \cdot
        \hat{\cC}_{-i_3 j_3 j_1})
        \nonumber \\
        & =
        \Big(\int_{A^{3}_{j_4 \, j_4 \, j_2}} \, \dela\mone\Big) \,
        (\hat{\cC}_{-i_6 -i_4 -i_2} \cdot 
        \hat{\cC}_{-i_5 i_6 -i_1} \cdot 
        \hat{\cC}_{-i_3 i_4 i_5} \cdot
        \hat{\cC}_{i_1 i_2 i_3}) \cdot
        (\hat{\cC}_{-i_1 -i_2 -i_3} \cdot 
        \hat{\cC}_{-j_5 -i_2 -j_1} \cdot 
        \hat{\cC}_{-j_3 -i_1 j_5} \cdot
        \hat{\cC}_{-i_3 j_3 j_1})
        \nonumber \\
        & =
        \sqrt{\mu^3} \,
        (\hat{\cC}_{-i_6 -i_4 -i_2} \cdot 
        \hat{\cC}_{-i_5 i_6 -i_1} \cdot 
        \hat{\cC}_{-i_3 i_4 i_5} \cdot
        \hat{\cC}_{i_1 i_2 i_3}) \cdot
        (\hat{\cC}_{-i_1 -i_2 -i_3} \cdot 
        \hat{\cC}_{-j_5 -i_2 -j_1} \cdot 
        \hat{\cC}_{-j_3 -i_1 j_5} \cdot
        \hat{\cC}_{-i_3 j_3 j_1})
        \nonumber \\
        & =
        \sqrt{\mu^3} \,
        (\hat{\cC}_{-i_6 -i_4 -i_2} \cdot 
        \hat{\cC}_{-i_5 i_6 -i_1} \cdot 
        \hat{\cC}_{-i_3 i_4 i_5} \cdot
        \hat{\cC}_{i_1 i_2 i_3}) \cdot
        (\hat{\cC}_{-j_5 -i_2 -j_1} \cdot 
        \hat{\cC}_{-j_3 -i_1 j_5} \cdot
        \hat{\cC}_{-i_3 j_3 j_1}) \,.
    \end{align}
    In the last step we used a the identity \eqref{ClosureConstraint_Reduction} to remove the closure constraint $\hat{\cC}_{-i_1 -i_2 -i_3}$.
\end{proof}
\medskip
\noindent
Finally, we give the details of the gluing of three tetrahedra amplitudes \eqref{VertexAmplitude} sharing a single \textit{internal} edge.
\begin{proposition}[Gluing three tetrahedra: internal edge]
\label{Prop_InternalEdge}
    Consider the tensor product of three independent tetrahedron amplitudes \eqref{GeneratingFunction_1-Tetra}. Take their product with three propagators as in Prop. \ref{Prop_Tetrahedra_Gluing}, in such a way that the three pairs of closure constraints identified (multiplied by the propagators) share a single tensor space (a common edge). Focus on the six closure constraints involved in the product.
    The pairwise identification of the six triangles (closure constraints) sharing one edge is performed by taking the co-integral of such product:
    \be
        \begin{aligned}
            \int_{A^3_{i,j,k}} \,
            &
            (\hat{\cC}_{-i \, i_3 \, i_4} \cdot
            \hat{\cC}_{i \, i_1 \, i_2}) \cdot
            \Delta_{-i \, j} \dela\mone \cdot
            \Delta_{-i \, k} \dela\mone \cdot
            (\hat{\cC}_{-j \, j_3 \, j_4} \cdot
            \hat{\cC}_{j \, j_1 \, j_2}) \cdot
            \Delta_{-j \, k} \dela\mone \cdot
            (\hat{\cC}_{-k \, k_3 \, k_4} \cdot
            \hat{\cC}_{k \, k_1 \, k_2})
            \\
            & =
            \sqrt{\mu^3} \,
            \big(\hat{\cC}_{i_4 \, i_3 \, i_1 \, i_2} \cdot
            \hat{\cC}_{-i_3 \, -i_4 \, j_3 \, j_4} \cdot
            \hat{\cC}_{i_4 \, i_3 \, j_1 \, j_2} \cdot
            \hat{\cC}_{-i_3 \, -i_4 \, k_3 \, k_4} \cdot
            \hat{\cC}_{i_4 \, i_3 \, k_1 \, k_2}\big) \,.
        \end{aligned}
    \ee
    Here the co-integral is on the Hopf algebras related to the internal edge: tensor spaces $i,j,k$.
\end{proposition}
\begin{proof}
    The proof follows a pattern similar to that of Prop. \ref{Prop_Tetrahedra_Gluing}. We first use the identity \eqref{Delta_CoProduct_Property} on the co-product $\Delta_{-i \, k}$ and then the left invariance of the co-integral on the $k^{th}$ tensor space and the normalization property of the delta function \eqref{Delta_Normalization} to remove this co-product.
    \begin{align}
        \int_{A^3_{i,j,k}} \,
        &
        (\hat{\cC}_{-i \, i_3 \, i_4} \cdot
        \hat{\cC}_{i \, i_1 \, i_2}) \cdot
        \Delta_{-i \, j} \dela\mone \cdot
        \Delta_{-i \, k} \dela\mone \cdot
        (\hat{\cC}_{-j \, j_3 \, j_4} \cdot
        \hat{\cC}_{j \, j_1 \, j_2}) \cdot
        \Delta_{-j \, k} \dela\mone \cdot
        (\hat{\cC}_{-k \, k_3 \, k_4} \cdot
        \hat{\cC}_{k \, k_1 \, k_2})
        \nonumber \\
        & =
        \int_{A^3_{i,j,k}} \,
        (\hat{\cC}_{-i \, i_3 \, i_4} \cdot
        \hat{\cC}_{i \, i_1 \, i_2}) \cdot
        \Delta_{-i \, j} \dela\mone \cdot
        (\hat{\cC}_{-j \, j_3 \, j_4} \cdot
        \hat{\cC}_{j \, j_1 \, j_2}) \cdot
        \Delta_{-j \, i} \dela\mone \cdot
        (\hat{\cC}_{-i \, k_3 \, k_4} \cdot
        \hat{\cC}_{i \, k_1 \, k_2}) \cdot
        \Delta_{-i \, k} \dela\mone
        \nonumber \\
        & =
        \int_{A^2_{i,j}} \,
        (\hat{\cC}_{-i \, i_3 \, i_4} \cdot
        \hat{\cC}_{i \, i_1 \, i_2}) \cdot
        \Delta_{-i \, j} \dela\mone \cdot
        (\hat{\cC}_{-j \, j_3 \, j_4} \cdot
        \hat{\cC}_{j \, j_1 \, j_2}) \cdot
        \Delta_{-j \, i} \dela\mone \cdot
        (\hat{\cC}_{-i \, k_3 \, k_4} \cdot
        \hat{\cC}_{i \, k_1 \, k_2}) \,
        \Big(\int_{A_k} \, \dela\mone \Big)
        \nonumber \\
        & =
        \sqrt{\mu} \, \int_{A^2_{i,j}} \,
        (\hat{\cC}_{-i \, i_3 \, i_4} \cdot
        \hat{\cC}_{i \, i_1 \, i_2}) \cdot
        \Delta_{-i \, j} \dela\mone \cdot
        (\hat{\cC}_{-j \, j_3 \, j_4} \cdot
        \hat{\cC}_{j \, j_1 \, j_2}) \cdot
        \Delta_{-j \, i} \dela\mone \cdot
        (\hat{\cC}_{-i \, k_3 \, k_4} \cdot
        \hat{\cC}_{i \, k_1 \, k_2}) \,.
    \end{align}
    Then repeating the same procedure with the co-product $\Delta_{-i \, j}$, we get
    \begin{align}
        \int_{A^3_{i,j,k}} \,
        &
        (\hat{\cC}_{-i \, i_3 \, i_4} \cdot
        \hat{\cC}_{i \, i_1 \, i_2}) \cdot
        \Delta_{-i \, j} \dela\mone \cdot
        \Delta_{-i \, k} \dela\mone \cdot
        (\hat{\cC}_{-j \, j_3 \, j_4} \cdot
        \hat{\cC}_{j \, j_1 \, j_2}) \cdot
        \Delta_{-j \, k} \dela\mone \cdot
        (\hat{\cC}_{-k \, k_3 \, k_4} \cdot
        \hat{\cC}_{k \, k_1 \, k_2})
        \nonumber \\
        & =
        \mu \, \int_{A_i} \,
        (\hat{\cC}_{-i \, i_3 \, i_4} \cdot
        \hat{\cC}_{i \, i_1 \, i_2}) \cdot
        (\hat{\cC}_{-i \, j_3 \, j_4} \cdot
        \hat{\cC}_{i \, j_1 \, j_2}) \cdot
        \Delta_{-i \, i} \dela\mone \cdot
        (\hat{\cC}_{-i \, k_3 \, k_4} \cdot
        \hat{\cC}_{i \, k_1 \, k_2})
        \nonumber \\
        & =
        V_H \sqrt{\mu} \, \int_{A_i} \,
        (\hat{\cC}_{-i \, i_3 \, i_4} \cdot
        \hat{\cC}_{i \, i_1 \, i_2}) \cdot
        (\hat{\cC}_{-i \, j_3 \, j_4} \cdot
        \hat{\cC}_{i \, j_1 \, j_2}) \cdot
        (\hat{\cC}_{-i \, k_3 \, k_4} \cdot
        \hat{\cC}_{i \, k_1 \, k_2}) \,.
    \end{align}
    Here in the last step we used that $\sqrt{\mu} \, \Delta_{-i \, i} \dela\mone = V_H$, which is straightforward to verify. The last step to prove the proposition consists in using the definition of closure constraint \eqref{ClosureConstraint} for the first constraint $\hat{\cC}_{-i \, i_1 \, i_2}$, together with the inverse property \eqref{ClosureConstraint_CyclicPermutation+Inversion}. Use the property \eqref{Delta_CoProduct_Property} on this co-product, then the left invariance of the co-integral on the $i^{th}$ tensor space and the normalization property of the delta to remove it.
    \begin{align}
        \int_{A^3_{i,j,k}} \,
        &
        (\hat{\cC}_{-i \, i_3 \, i_4} \cdot
        \hat{\cC}_{i \, i_1 \, i_2}) \cdot
        \Delta_{-i \, j} \dela\mone \cdot
        \Delta_{-i \, k} \dela\mone \cdot
        (\hat{\cC}_{-j \, j_3 \, j_4} \cdot
        \hat{\cC}_{j \, j_1 \, j_2}) \cdot
        \Delta_{-j \, k} \dela\mone \cdot
        (\hat{\cC}_{-k \, k_3 \, k_4} \cdot
        \hat{\cC}_{k \, k_1 \, k_2})
        \nonumber \\
        & =
        V_H \sqrt{\mu} \, \int_{A_i} \,
        (\hat{\cC}_{-i_4 \, -i_3 \, i} \cdot
        \hat{\cC}_{i \, i_1 \, i_2}) \cdot
        (\hat{\cC}_{-i \, j_3 \, j_4} \cdot
        \hat{\cC}_{i \, j_1 \, j_2}) \cdot
        (\hat{\cC}_{-i \, k_3 \, k_4} \cdot
        \hat{\cC}_{i \, k_1 \, k_2})
        \nonumber \\
        & =
        \mu \, \int_{A_i} \,
        \Delta^3_{-i_4 \, -i_3 \, i} \, \dela\mone \cdot
        \big(\hat{\cC}_{i \, i_1 \, i_2} \cdot
        \hat{\cC}_{-i \, j_3 \, j_4} \cdot
        \hat{\cC}_{i \, j_1 \, j_2} \cdot
        \hat{\cC}_{-i \, k_3 \, k_4} \cdot
        \hat{\cC}_{i \, k_1 \, k_2}\big)
        \nonumber \\
        & =
        \mu \, \int_{A_i} \,
        \big(\hat{\cC}_{i_4 \, i_3 \, i_1 \, i_2} \cdot
        \hat{\cC}_{-i_3 \, -i_4 \, j_3 \, j_4} \cdot
        \hat{\cC}_{i_4 \, i_3 \, j_1 \, j_2} \cdot
        \hat{\cC}_{-i_3 \, -i_4 \, k_3 \, k_4} \cdot
        \hat{\cC}_{i_4 \, i_3 \, k_1 \, k_2}\big) \cdot
        \Delta^3_{-i_4 \, -i_3 \, i} \, \dela\mone
        \nonumber \\
        & =
        \mu \,
        \big(\hat{\cC}_{i_4 \, i_3 \, i_1 \, i_2} \cdot
        \hat{\cC}_{-i_3 \, -i_4 \, j_3 \, j_4} \cdot
        \hat{\cC}_{i_4 \, i_3 \, j_1 \, j_2} \cdot
        \hat{\cC}_{-i_3 \, -i_4 \, k_3 \, k_4} \cdot
        \hat{\cC}_{i_4 \, i_3 \, k_1 \, k_2}\big) \,
        \Big(\int_{A_i} \, \dela\mone \Big)
        \nonumber \\
        & =
        \sqrt{\mu^3} \,
        \big(\hat{\cC}_{i_4 \, i_3 \, i_1 \, i_2} \cdot
        \hat{\cC}_{-i_3 \, -i_4 \, j_3 \, j_4} \cdot
        \hat{\cC}_{i_4 \, i_3 \, j_1 \, j_2} \cdot
        \hat{\cC}_{-i_3 \, -i_4 \, k_3 \, k_4} \cdot
        \hat{\cC}_{i_4 \, i_3 \, k_1 \, k_2}\big) \,.
    \end{align}
\end{proof}

%% file: AppHAFT.tex
In this part we provide further mathematical details of the Hopf algebra field theory presented in section \ref{Sec_HopfAlgebraFieldTheory}, from the definition of gauge projector \eqref{GaugeProj_Field}, to the derivation of tetrahedron amplitude \eqref{VertexAmplitude}, and then we provide the details of the derivation of the Feynman diagram amplitudes of Hopf algebra field theory \eqref{GeneratingFunction_Triangulation}, \eqref{GeneratingFunction_DualComplex}, \eqref{GeneratingFunction_Double} and show the topological invariance of the model.

\subsection{Gauge projector}
\label{App_GaugeProjector}

First notice that, by analogy with ordinary GFT, we referred to the map $\cP_L$ in \eqref{GaugeProj_Field} as a \textit{projector}. A map $\cP$ is a projector $\cP$ if it satisfies the condition
\be
    \cP \circ \cP = \cP \,.
\ee
To verify that the map $\cP_L$ is actually a projector, we need the following identity to hold
\be\label{prop1}
    (m^n \ot m^n) \circ \Delta^{(n)} = \Delta \circ m^n \,.
\ee
Given an Hopf algebra $H \ni h$, let the elements of its tensor product $H^n$ be $(h_1 \ot \cdots \ot h_n)$. We verify the above identity simply using the definitions of product \eqref{1-d_Multiplication} and co-product \eqref{d-d_CoProduct}. Using the Sweedler notation, the lhs is
\begin{align}
    (m^n \ot m^n) \circ \Delta^{(n)} (h_1 \ot \cdots \ot h_n)
    & =
    (m^d \ot m^d) \big((h_{1 \, (1)} \ot \cdots \ot h_{n \, (1)}) \ot (h_{1 \, (2)} \ot \cdots \ot h_{n \, (2)}) \big)
    \nonumber \\
    & =
    (h_{1 \, (1)} \cdot \cdots h_{n \, (1)}) \ot 
    (h_{1 \, (2)} \cdot \cdots h_{n \, (2)}) \,,
\end{align}
while the rhs is
\be
    \Delta \circ m^n (h_1 \ot \cdots \ot h_n)
    =
    \Delta (h_1 \cdot \cdots h_n)
    =
    (h_{1 \, (1)} \cdot \cdots h_{n \, (1)}) \ot 
    (h_{1 \, (2)} \cdot \cdots h_{n \, (2)}) \,,
\ee
which proves the identity.
Then, using it together with the definition of $\cP_L$ \eqref{GaugeProj_Field} and the left invariance of the co-integral, we get
\begin{align}
    (\cP_L \, \cP_L \, \Phi)
    & =
    \frac{1}{V_H^2} \left(\int_H \ot \int_H \ot \idd \right) \, 
    \big((m^3 \ot m^3 \ot \idd) \circ (\idd \ot \Delta^{(3)}) \circ \Delta^{(3)} \Phi\big) \,  \textrm{\tiny use co-associativity}
    \nonumber \\
    & =
    \frac{1}{V_H^2} \left(\int_H \ot \int_H \ot \idd \right) \, 
    \big((m^3 \ot m^3 \ot \idd) \circ (\Delta^{(3)} \ot \idd) \circ \Delta^{(3)} \Phi\big)\, \textrm{\tiny use property \eqref{prop1}} 
    \nonumber \\
    & =
    \frac{1}{V_H^2} \left(\int_H \ot \int_H \ot \idd \right) \, 
    \big((\Delta \ot \idd) \circ (m^3 \ot \idd) \circ \Delta^{(3)} \Phi\big)\,  \textrm{\tiny use the left invariance of the co-integral}  % what happened to go to the next line??
    \nonumber \\
    & =
    \frac{1}{V_H^2} \left(\int_H \ot \int_H \ot \idd \right) \, 
    \big((\id \ot m^3 \ot \idd) \circ (1 \ot \Delta^{(3)} \Phi)\big)\,  \textrm{\tiny use the normalization of the co-integral}
    \nonumber \\
    & =
    \frac{1}{V_H} \left(\int_H \ot \idd \right) \, 
    \big((m^3 \ot \idd) \circ \Delta^{(3)} \Phi\big)
    =
    (\cP_L \, \Phi) \,.
\end{align}
In the second last step we also used the normalization of the co-integral $\int_H \, 1 = V_H$.

\subsection{Tetrahedron amplitude}
\label{App_TetrahedronAmplitude}

We show here that the interaction term \eqref{InteractionTerm} can be expressed as in \eqref{KineticTerm_Amplitude} with kernels of the tetrahedron amplitude \eqref{VertexAmplitude}.
To simplify the computation we use the tensor product notation for the field $\Phi \equiv \Phi_{1 \, 2 \, 3} \in H^3$.
Let us start by the kernel $\hat{\cV}$ in the representation of $A$. We use the inverse Fourier transform together with the closure constraint \eqref{ClosureConstraint} for each of the projected fields $(\cP_L \, \Phi)$. We also use the definition of the co-product \eqref{Underline_CoProduct_3d} and that of the opposite delta function.
\begin{align}
    \cS_{\cV} 
    & =
    \int_{H^6} \, \um (\cP\Phi \ot \cP\Phi \ot \cP\Phi \ot \cP\Phi)
    =
    \int_{H^6} \, (\cP\Phi_{1 \, 2 \, 3} \ot \cP\Phi_{3 \, 4 \, 5} \ot \cP\Phi_{5 \, 6 \, 1} \ot \cP\Phi_{6 \, 4 \, 2})
    \nonumber \\
    & =
    \frac{1}{\mu^{12}} \left(\int_{H^6} \ot \int_{A^{12}}\right) \, 
    \big(
    (\sigma\mone_{1 7} \cdot
    \sigma\mone_{2 8} \cdot
    \sigma\mone_{3 9}) \cdot
    (\sigma\mone_{3 \, 10} \cdot
    \sigma\mone_{4 \, 11} \cdot
    \sigma\mone_{5 \, 12}) \cdot
    (\sigma\mone_{5 \, 13} \cdot
    \sigma\mone_{6 \, 14} \cdot
    \sigma\mone_{1 \, 15}) \cdot
    (\sigma\mone_{6 \, 16} \cdot
    \sigma\mone_{4 \, 17} \cdot
    \sigma\mone_{2 \, 18})
    \nonumber \\
    & \qquad \qquad \qquad \qquad \qquad \cdot
    (1^{\ot 6} \ot \hat{\cC} \cdot \hat{\Phi} \ot \hat{\cC} \cdot \hat{\Phi} \ot \hat{\cC} \cdot \hat{\Phi} \ot \hat{\cC} \cdot \hat{\Phi})\big)
    \nonumber \\
    & =
    \frac{1}{\mu^{12}} \left(\int_{H^6} \ot \int_{A^{12}}\right) \, 
    \big(
    (\id^{\ot 6} \ot \uDe)\,  (\Sigma^6)\mone \cdot
    (1^{\ot 6} \ot \hat{\cC} \cdot \hat{\Phi} \ot \hat{\cC} \cdot \hat{\Phi} \ot \hat{\cC} \cdot \hat{\Phi} \ot \hat{\cC} \cdot \hat{\Phi})\big)
    \nonumber \\
    & =
    \frac{1}{\mu^{6}} \, \int_{A^{12}} \, 
    \big(
    \uDe \dela\mone \cdot
    (\hat{\cC} \cdot \hat{\Phi} \ot \hat{\cC} \cdot \hat{\Phi} \ot \hat{\cC} \cdot \hat{\Phi} \ot \hat{\cC} \cdot \hat{\Phi})\big) \,.
\end{align}
Note that the opposite delta function is an element of the tensor product Hopf algebra $A^6$. 
In order to derive the kernel $\cV$ in the dual complex picture, we first use the definition of Fourier transform on each of the four projected fields, eq. \eqref{ClosureConstraint_Transform}, and write the projected field as in \eqref{GaugeProj_Field}.
\begin{align}
    \cS_{\cV} 
    & =
    \frac{1}{\mu^{12}} \left(\int_{H^{12}} \ot \int_{A^{12}}\right) \, 
    \big(
    (1^{\ot 12} \ot \uDe \dela\mone) \cdot
    \prod_{i=1}^{12} \, \sigma_{i \, 12+i} \cdot
    (\cP\Phi \ot \cP\Phi \ot \cP\Phi \ot \cP\Phi \ot 1^{\ot 12})\big)
    \nonumber \\
    & =
    \frac{1}{V_H^4 \, \mu^{12}} \left(\int_{H^4} \ot \int_{H^{12}} \ot \int_{A^{12}}\right) \, 
    \big(
    (1^{\ot 16} \ot \uDe \dela\mone) \cdot
    \prod_{i=1}^{12} \, \sigma_{4+i \, 16+i} \cdot
    ((m^3 \ot \id^{\ot 3}) \circ \Delta^{3} \, \Phi \ot 1^{\ot 9})
    \nonumber \\
    & \qquad \qquad \qquad \qquad \qquad \qquad \qquad \, \cdot
    (1 \ot (m^3 \ot \id^{\ot 3}) \circ \Delta^{3} \,\Phi \ot 1^{\ot 6}) \cdot
    (1^{\ot 2} \ot (m^3 \ot \id^{\ot 3}) \circ \Delta^{3} \,\Phi \ot 1^{\ot 3})
    \nonumber \\
    & \qquad \qquad \qquad \qquad \qquad \qquad \qquad \, \cdot
    (1^{\ot 3} \ot (m^3 \ot \id^{\ot 3}) \circ \Delta^{3} \,\Phi) \ot 1^{\ot 12})\big) \,.
\end{align}
Now multiply the skew co-pairings from the left, such that they can be written as products of three dimensional co-products. Then use the left invariance of the co-integral on $H$ to remove such co-products.
\begin{align}
    \cS_{\cV} 
    & =
    \frac{1}{V_H^4 \, \mu^{12}} \left(\int_{H^4} \ot \int_{H^{12}} \ot \int_{A^{12}}\right) \, 
    \big(
    (1^{\ot 16} \ot \uDe \dela\mone) \cdot
    (\sigma_{1 \, 19}\mone \cdot \sigma_{1 \, 18}\mone \cdot \sigma_{1 \, 17}\mone) \cdot
    (\sigma_{2 \, 22}\mone \cdot \sigma_{2 \, 21}\mone \cdot \sigma_{2 \, 20}\mone) \cdot
    (\sigma_{3 \, 25}\mone \cdot \sigma_{3 \, 24}\mone \cdot \sigma_{3 \, 23}\mone) \nonumber \\
    & \qquad \qquad \qquad \qquad \qquad \qquad \qquad \, \cdot
    (\sigma_{4 \, 28}\mone \cdot \sigma_{4 \, 27}\mone \cdot \sigma_{4 \, 26}\mone) \cdot 
    ((m^3 \ot \id^{\ot 15} \ot \id^{\ot 12}) \circ (\Delta^{3} \ot \id^{\ot 3}) \, (\Sigma^3 \cdot (\Phi \ot 1^{\ot 9} \ot 1^{\ot 12}))
    \nonumber \\
    & \qquad \qquad \qquad \qquad \qquad \qquad \qquad \, \cdot
    (1 \ot (m^3 \ot \id^{\ot 14} \ot \id^{\ot 12}) \circ (\Delta^{3} \ot \id^{\ot 3}) \, (\Sigma^3 \cdot (\Phi \ot 1^{\ot 6} \ot 1^{\ot 12}))
    \nonumber \\
    & \qquad \qquad \qquad \qquad \qquad \qquad \qquad \, \cdot
    (1^{\ot 2} \ot (m^3 \ot \id^{\ot 13} \ot \id^{\ot 12}) \circ (\Delta^{3} \ot \id^{\ot 3}) \, (\Phi \ot 1^{\ot 3} \ot 1^{\ot 12}))
    \nonumber \\
    & \qquad \qquad \qquad \qquad \qquad \qquad \qquad \, \cdot
    (1^{\ot 3} \ot (m^3 \ot \id^{\ot 12} \ot \id^{\ot 12}) \circ (\Delta^{3} \ot \id^{\ot 3}) \, (\Sigma \cdot (\Phi \ot 1^{\ot 12}))\big) 
    \nonumber \\
    & =
    \frac{1}{V_H^4 \, \mu^{12}} \left(\int_{H^4} \ot \int_{H^{12}} \ot \int_{A^{12}}\right) \, 
    \big(
    (1^{\ot 16} \ot \uDe \dela\mone) \cdot
    (\sigma_{1 \, 19}\mone \cdot \sigma_{1 \, 18}\mone \cdot \sigma_{1 \, 17}\mone) \cdot
    (\sigma_{2 \, 22}\mone \cdot \sigma_{2 \, 21}\mone \cdot \sigma_{2 \, 20}\mone) \cdot
    (\sigma_{3 \, 25}\mone \cdot \sigma_{3 \, 24}\mone \cdot \sigma_{3 \, 23}\mone) \nonumber \\
    & \qquad \qquad \qquad \qquad \qquad \qquad \qquad \, \cdot
    (\sigma_{4 \, 28}\mone \cdot \sigma_{4 \, 27}\mone \cdot \sigma_{4 \, 26}\mone) \cdot 
    (\sigma_{5 \, 17} \cdot \sigma_{6 \, 18} \cdot \sigma_{7 \, 19}) \cdot 
    (\sigma_{8 \, 20} \cdot \sigma_{9 \, 21} \cdot \sigma_{10 \, 22}) \cdot 
    (\sigma_{11 \, 23} \cdot \sigma_{12 \, 24} \cdot \sigma_{13 \, 25}) 
    \nonumber \\
    & \qquad \qquad \qquad \qquad \qquad \qquad \qquad \, \cdot
    (\sigma_{14 \, 26} \cdot \sigma_{15 \, 27} \cdot \sigma_{16 \, 28}) \cdot 
    (1^{\ot 4} \ot \Phi \ot \Phi \ot \Phi \ot \Phi \ot 1^{\ot 12})\big) \,.
\end{align}
Last, we merge the closure constraints writing the co-product $\uDe$ as in \eqref{Underline_CoProduct_3d} and using the property \eqref{Delta_CoProduct_Property}; then we group the skew co-pairings four by four with the property \eqref{SkewEl_CoProduct-id}.
To group them properly we need to use the inverse and cyclic permutation identities \eqref{ClosureConstraint_CyclicPermutation+Inversion} on the skew co-pairings.
Then use the left invariance of the co-integral on $A$ and the normalization property \eqref{Delta_Normalization} to remove of the opposite delta function. Finally using the definition of opposite delta function and the co-product \eqref{Underline_CoProduct_Simplex_3d} to group them, we get the expression \eqref{VertexAmplitude} of the kernel $\cV$ in the $H$ representation.
\begin{align}
    \cS_{\cV} 
    & =
    \frac{1}{V_H^4 \, \mu^{12}} \left(\int_{H^4} \ot \int_{H^{12}} \ot \int_{A^{12}}\right) \, 
    \big(
    \Delta_{17 \, 25} \dela\mone \cdot
    \Delta_{18 \, 28} \dela\mone \cdot
    \Delta_{19 \, 20} \dela\mone \cdot
    \Delta_{21 \, 27} \dela\mone \cdot
    \Delta_{22 \, 23} \dela\mone \cdot
    \Delta_{24 \, 26} \dela\mone
    \nonumber \\
    & \qquad \qquad \qquad \qquad \qquad \qquad \qquad \, \cdot
    (\sigma_{1 \, 19}\mone \cdot \sigma_{1 \, 18}\mone \cdot \sigma_{1 \, 17}\mone) \cdot
    (\sigma_{2 \, 22}\mone \cdot \sigma_{2 \, 21}\mone \cdot \sigma_{2 \, 20}\mone) \cdot
    (\sigma_{3 \, 25}\mone \cdot \sigma_{3 \, 24}\mone \cdot \sigma_{3 \, 23}\mone) \cdot 
    (\sigma_{4 \, 28}\mone \cdot \sigma_{4 \, 27}\mone \cdot \sigma_{4 \, 26}\mone) \nonumber \\
    & \qquad \qquad \qquad \qquad \qquad \qquad \qquad \, \cdot
    (\sigma_{5 \, 17} \cdot \sigma_{6 \, 18} \cdot \sigma_{7 \, 19}) \cdot 
    (\sigma_{8 \, 20} \cdot \sigma_{9 \, 21} \cdot \sigma_{10 \, 22}) \cdot 
    (\sigma_{11 \, 23} \cdot \sigma_{12 \, 24} \cdot \sigma_{13 \, 25}) 
    \nonumber \\
    & \qquad \qquad \qquad \qquad \qquad \qquad \qquad \, \cdot
    (\sigma_{14 \, 26} \cdot \sigma_{15 \, 27} \cdot \sigma_{16 \, 28}) \cdot 
    (1^{\ot 4} \ot \Phi \ot \Phi \ot \Phi \ot \Phi \ot 1^{\ot 12})\big)
    \nonumber \\
    & =
    \frac{1}{V_H^4 \, \mu^{12}} \left(\int_{H^4} \ot \int_{H^{12}} \ot \int_{A^{12}}\right) \, 
    \big(
    \Delta_{17 \, 25} \dela\mone \cdot
    \Delta_{18 \, 28} \dela\mone \cdot
    \Delta_{19 \, 20} \dela\mone \cdot
    \Delta_{21 \, 27} \dela\mone \cdot
    \Delta_{22 \, 23} \dela\mone \cdot
    \Delta_{24 \, 26} \dela\mone
    \nonumber \\
    & \qquad \qquad \qquad \qquad \qquad \qquad \qquad \, \cdot
    (\sigma_{4 \, 28}\mone \cdot \sigma_{4 \, 27}\mone \cdot \sigma_{4 \, 26}\mone)  \cdot
    (\sigma_{3 \, 24}\mone \cdot \sigma_{3 \, 23}\mone \cdot \sigma_{3 \, 25}\mone) \cdot 
    (\sigma_{2 \, 21} \cdot \sigma_{2 \, 22} \cdot \sigma_{2 \, 20}) \cdot
    (\sigma_{1 \, 18}\mone \cdot \sigma_{1 \, 17}\mone \cdot \sigma_{1 \, 19}\mone)
    \nonumber \\
    & \qquad \qquad \qquad \qquad \qquad \qquad \qquad \, \cdot
    (\sigma_{5 \, 17} \cdot \sigma_{6 \, 18} \cdot \sigma_{7 \, 19}) \cdot 
    (\sigma_{8 \, 20} \cdot \sigma_{9 \, 21} \cdot \sigma_{10 \, 22}) \cdot 
    (\sigma_{11 \, 23} \cdot \sigma_{12 \, 24} \cdot \sigma_{13 \, 25}) 
    \nonumber \\
    & \qquad \qquad \qquad \qquad \qquad \qquad \qquad \, \cdot
    (\sigma_{14 \, 26} \cdot \sigma_{15 \, 27} \cdot \sigma_{16 \, 28}) \cdot 
    (1^{\ot 4} \ot \Phi \ot \Phi \ot \Phi \ot \Phi \ot 1^{\ot 12})\big)
    \nonumber \\
    & =
    \frac{1}{V_H^4 \, \mu^{12}} \left(\int_{H^4} \ot \int_{H^{12}} \ot \int_{A^{12}}\right) \, 
    \big(
    (1^{\ot 16} \ot 1^{\ot 3} \ot \dela\mone \ot 1^{\ot 2} \ot \dela\mone \ot 1 \ot \dela\mone \ot \dela\mone \ot \dela\mone \ot \dela\mone)
    \nonumber \\
    & \qquad \qquad \qquad \qquad \qquad \qquad \qquad \, \cdot
    (\sigma_{16 \, 18}\mone \cdot \sigma_{4 \, 18} \cdot \sigma_{1 \, 18}\mone \cdot \sigma_{6 \, 18}) \cdot
    (\sigma_{15 \, 21}\mone \cdot \sigma_{4 \, 21} \cdot \sigma_{2 \, 21} \cdot \sigma_{9 \, 21}) \cdot
    (\sigma_{14 \, 24}\mone \cdot \sigma_{4 \, 24} \cdot \sigma_{3 \, 24}\mone  \cdot \sigma_{12 \, 24}) 
    \nonumber \\
    & \qquad \qquad \qquad \qquad \qquad \qquad \qquad \, \cdot
    (\sigma_{11 \, 22}\mone \cdot \sigma_{3 \, 22} \cdot \sigma_{2 \, 22} \cdot \sigma_{10 \, 22}) \cdot
    (\sigma_{13 \, 17}\mone \cdot  \sigma_{3 \, 17} \cdot \sigma_{1 \, 17}\mone \cdot \sigma_{5 \, 17}) \cdot
    (\sigma_{8 \, 19}\mone \cdot \sigma_{2 \, 19}\mone \cdot \sigma_{1 \, 19}\mone \cdot \sigma_{7 \, 19})
    \nonumber \\
    & \qquad \qquad \qquad \qquad \qquad \qquad \qquad \, \cdot
    (1^{\ot 4} \ot \Phi \ot \Phi \ot \Phi \ot \Phi \ot 1^{\ot 12})\big)
    \nonumber \\
    & =
    \frac{1}{V_H^4 \, \mu^{9}} \left(\int_{H^4} \ot \int_{H^{12}} \ot \int_{A^{6}}\right) \, 
    \big(
    \Delta^4_{-16 \, 4 \, -1 \, 6} \, \sigma \cdot
    \Delta^4_{-15 \, 4 \, 2 \, 9} \, \sigma \cdot
    \Delta^4_{-14 \, 4 \, -3 \, 12} \, \sigma \cdot
    \Delta^4_{-11 \, 3 \, 2 \, 10} \, \sigma \cdot
    \Delta^4_{-13 \, 3 \, -1 \, 5} \, \sigma \cdot
    \Delta^4_{-8 \, -2 \, -1 \, 7} \, \sigma
    \nonumber \\
    & \qquad \qquad \qquad \qquad \qquad \qquad \quad \, \cdot
    (1^{\ot 4} \ot \Phi \ot \Phi \ot \Phi \ot \Phi \ot 1^{\ot 6})\big)
    \nonumber \\
    & =
    \frac{1}{V_H^4 \, \mu^{6}} \left(\int_{H^4} \ot \int_{H^{12}}\right) \, 
    \big(
    \Delta^4_{-16 \, 4 \, -1 \, 6} \, \delh\mone \cdot
    \Delta^4_{-15 \, 4 \, 2 \, 9} \, \delh\mone \cdot
    \Delta^4_{-14 \, 4 \, -3 \, 12} \, \delh\mone \cdot
    \Delta^4_{-11 \, 3 \, 2 \, 10} \, \delh\mone \cdot
    \Delta^4_{-13 \, 3 \, -1 \, 5} \, \delh\mone 
    \nonumber \\
    & \qquad \qquad \qquad \qquad \qquad \,\, \cdot
    \Delta^4_{-8 \, -2 \, -1 \, 7} \, \delh\mone \cdot
    (1^{\ot 4} \ot \Phi \ot \Phi \ot \Phi \ot \Phi)\big)
\end{align}
In the fourth line we used the tensor notation \eqref{ij_CoProduct_Antipode} to encode the co-product and the inverse antipodes.

\subsection{Feynman diagram amplitude}
\label{App_PartitionFunction}

We now have all the ingredients to derive the Feynman amplitudes \eqref{GeneratingFunction_Triangulation}, \eqref{GeneratingFunction_DualComplex} and \eqref{GeneratingFunction_Double} following Prop. \ref{Prop_GeneratingFunction}.

\paragraph{Triangulation picture. }
The definition of Feynman amplitude of a given triangulation in the $A$ representation is straightforward.
As the amplitude of a single tetrahedron is associated to the kernel \eqref{VertexAmplitude}, the amplitude of a triangulation $\Gamma$ made of $M$ tetrahedra is simply given as a combination of $M$ (reduced) amplitudes \eqref{GeneratingFunction_1-Tetra} glued together by the proper propagator amplitudes \eqref{PropagatorAmplitude}. The details of the gluing of several tetrahedra through the propagators are provided in Prop. \ref{Prop_Tetrahedra_Gluing} and \ref{Prop_InternalEdge}. 
Following the gluing procedure explained in these propositions, one obtains the expression \eqref{GeneratingFunction_Triangulation}.

\paragraph{Dual complex picture. }
The amplitude of the Feynman diagrams of Hopf algebra field theory in the dual complex is more subtle. We consider the combination of an arbitrary number of tetrahedron amplitudes \eqref{VertexAmplitude} merged through the propagator amplitudes \eqref{PropagatorAmplitude} such that they share a single edge, as in Prop. \ref{Prop_InternalEdge}. We isolate the contribution of the links in the bulk of such tetrahedra. They form a closed loop around the internal edge. The resulting amplitude encodes the local curvature around the edge, and the overall Hopf algebra field theory Feynman amplitude of a dual complex is given by the sum over all the edges of the triangulation, of such local curvatures.
\smallskip \\
The vertex amplitude is given in terms of the co-product $\sDe$ defined in \eqref{Underline_CoProduct_Simplex_3d} and the propagator amplitude in terms of the three dimensional co-product \eqref{d-d_CoProduct}. 
The gluing that we are interested in, involves only one loop of solid blue and red links of Fig. \ref{Fig_VertexAmplitude} for each graph dual to a tetrahedron (one four dimensional co-product in $\sDe$ in the kernel \eqref{VertexAmplitude}) and a single co-product $\Delta$ for each propagator amplitude. 
We thus multiply a number of partial vertex amplitudes of the type
\be
    \coV = \frac{1}{\mu \, V_H} \left(\id \ot \int_H^2 \ot \id\right) \, \Delta^4_{-1 \, 2 \, -3 \, 4} \delh\mone \,,  
\ee
with a number of partial propagator amplitudes of the type
\be
    \coK = \frac{1}{\sqrt{\mu}} \, (S \ot \id) \circ \Delta \delh
    =
    \frac{1}{\sqrt{\mu}} \, \Delta_{-1 \, 2} \delh \,.
\ee
Let us first merge two (partial) vertex amplitudes and denote the resulting amplitude $\coV_2$.
It is given by the product
\be
    \coV_2 =
    \frac{1}{\sqrt{\mu^{5}} \, V_H^2} \left(\id \ot \int_{H^{2}} \ot \id^{\ot 2} \ot \int_{H^{2}} \ot \id\right) \,
    \big(
    \Delta^4_{-1 \, 2 \, -3 \, 4} \delh\mone \cdot 
    \Delta_{-5 \, 4} \delh \cdot 
    \Delta^4_{-5 \, 6 \, -7 \, 8} \delh\mone
    \big) \,.
\ee
Use the property \eqref{Delta_CoProduct_Property} of the delta function, between the propagator amplitude and the first vertex amplitude. 
\be
    \coV_2 =
    \frac{1}{\sqrt{\mu^{5}} \, V_H^2} \left(\id \ot \int_{H^{2}} \ot \id^{\ot 2} \ot \int_{H^{2}} \ot \id\right) \,
    \big(
    \Delta_{-5 \, 4} \delh \cdot
    \Delta^4_{-1 \, 2 \, -3 \, 5} \delh\mone \cdot 
    \Delta^4_{-5 \, 6 \, -7 \, 3} \delh\mone
    \big) \,.
\ee
We get rid of the propagator amplitude by taking the (normalized) co-integral on the fourth tensor space, so that we can use the left invariance of the co-integral on $H$ and then the normalization property of the delta function.
\begin{align}
    \left(\id \ot \int_H \ot \id^{\ot 2}\right) \, \coV_2
    & =
    \frac{1}{\sqrt{\mu^{5}} \, V_H^2} \left(\id \ot \int_{H^{3}} \ot \id \ot \int_{H^{2}} \ot \id\right) \,
    \big(
    (1^{\ot 3} \ot \delh \ot 1^{\ot 4}) \cdot
    \Delta^4_{-1 \, 2 \, -3 \, 5} \delh\mone \cdot
    \Delta^4_{-5 \, 6 \, -7 \, 8} \delh\mone
    \big)
    \nonumber \\
    & =
    \frac{1}{\mu^2 \, V_H^2} \left(\id \ot \int_{H^{2}} \ot \id \ot \int_{H^{2}} \ot \id\right) \,
    \big(
    \Delta^4_{-1 \, 2 \, -3 \, 4} \delh\mone \cdot
    \Delta^4_{-4 \, 5 \, -6 \, 7} \delh\mone
    \big) \,.
\end{align}
The last step consists in merging the two co-products. 
The easiest way to proceed is to use the definition of delta function to write them as skew co-pairing elements and use the properties \eqref{SkewEl_CoProduct-id} and \eqref{SkewEl_CoProduct-id_inv}.
\begin{align}
    \Delta^4_{-1 \, 2 \, -3 \, 4} \delh \cdot
    \Delta^4_{-4 \, 5 \, -6 \, 7} \delh
    & =
    \frac{1}{\mu} \left(\id^{\ot 7} \ot \int_{A^2}\right) \,
    \big(
    (\Delta^4_{-1 \, 2 \, -3 \, 4} \ot \id) \, \sigma \cdot 
    (\Delta^4_{-4 \, 5 \, -6 \, 7} \ot \id^{\ot 2}) \, \sigma\big)
    \nonumber \\
    & =
    \frac{1}{\mu} \left(\id^{\ot 7} \ot \int_{A^2}\right) \,
    \big(
    \sigma_{1 \, 8}\mone \cdot 
    \sigma_{2 \, 8} \cdot
    \sigma_{3 \, 8}\mone \cdot
    \sigma_{4 \, 8} \cdot
    \sigma_{4 \, 9}\mone \cdot
    \sigma_{5 \, 9} \cdot
    \sigma_{6 \, 9}\mone \cdot
    \sigma_{7 \, 9}
    \big) \,.
\end{align}
Now we would like to group the above elements in a co-product so that we can use the left invariance of the co-integral on $A$. To do this, we multiply for the proper combination of skew co-pairing $\sigma$ and convolution inverses $\sigma\mone$ and use the properties \eqref{SkewEl_CoProduct-id} and \eqref{SkewEl_CoProduct-id_inv}.
\begin{align}
    \Delta^4_{-1 \, 2 \, -3 \, 4} \delh \cdot
    \Delta^4_{-4 \, 5 \, -6 \, 7} \delh
    =
    \frac{1}{\mu} \left(\id^{\ot 7} \ot \int_{A^2}\right) \,
    &
    \big(
    (\id \ot S \ot \id^{\ot 5} \ot S \ot \id) \,
    (\sigma_{3 \, 9} \cdot
    \sigma_{2 \, 9} \cdot
    \sigma_{1 \, 9} \cdot
    (\sigma_{1 \, 9}\mone \cdot
    \sigma_{1 \, 8}\mone)
    \nonumber \\
    & \cdot
    (\sigma_{2 \, 9}\mone \cdot
    \sigma_{2 \, 8}\mone) \cdot
    (\sigma_{3 \, 9}\mone \cdot
    \sigma_{3 \, 8}\mone) \cdot
    (\sigma_{4 \, 8} \cdot
    \sigma_{4 \, 9})) \cdot 
    (\sigma_{5 \, 9} \cdot
    \sigma_{6 \, 9}\mone \cdot
    \sigma_{7 \, 9})
    \big)
    \nonumber \\
    =
    \frac{1}{\mu} \left(\id^{\ot 7} \ot \int_{A^2}\right) \,
    &
    \big(
    (\id \ot S \ot \id^{\ot 5} \ot S \ot \id) \,
    (\sigma_{3 \, 9} \cdot
    \sigma_{2 \, 9} \cdot
    \sigma_{1 \, 9} \cdot
    (\sigma_{1 \, 8}\mone \cdot \sigma_{2 \, 8}\mone \cdot \sigma_{3 \, 8}\mone \cdot \sigma_{4 \, 8}))
    \nonumber \\
    & \cdot
    (\sigma_{5 \, 9} \cdot
    \sigma_{6 \, 9}\mone \cdot
    \sigma_{7 \, 9})
    \big)
    \nonumber \\
    =
    \frac{1}{\mu} \left(\id^{\ot 7} \ot \int_{A^2}\right) \,
    &
    \big(
    (\id \ot S \ot \id^{\ot 5} \ot S \ot \id) \,
    (\sigma_{3 \, 9} \cdot
    \sigma_{2 \, 9} \cdot
    \sigma_{1 \, 9} \cdot
    (\Delta^4_{-1 \, -2 \, -3 \, 4} \ot \id) \, \sigma) 
    \nonumber \\
    & \cdot
    (\sigma_{5 \, 9} \cdot
    \sigma_{6 \, 9}\mone \cdot
    \sigma_{7 \, 9})
    \big) 
    \nonumber \\
    =
    \frac{1}{\sqrt{\mu}} \left(\id^{\ot 7} \ot \int_{A}\right) \,
    &
    \big(
    (\id \ot S \ot \id^{\ot 5} \ot S) \,
    (\sigma_{3 \, 8} \cdot
    \sigma_{2 \, 8} \cdot
    \sigma_{1 \, 8} \cdot
    \Delta^4_{-1 \, -2 \, -3 \, 4} \delh\mone) 
    \cdot
    (\sigma_{5 \, 8} \cdot
    \sigma_{6 \, 8}\mone \cdot
    \sigma_{7 \, 8})
    \big) \,.
\end{align}
In the last two lines we grouped all the skew co-pairings belonging to the eighth tensor space into a four dimensional co-product, and then we used the definition of delta function. 
Now taking the co-integral on the fourth tensor product, we can use its left invariance to remove the co-product on the delta function in the expression above and then its normalization property to get rid of it.
\begin{align}
    \left(\id \ot \int_{H^2} \ot \id\right) \, \coV_2
    & =
    \frac{\sqrt{\mu}}{\mu^3 \, V_H^2} \left(\id \ot \int_{H^{5}} \ot \id \ot \int_A\right) \,
    \big(
    (\id \ot S \ot \id^{\ot 5} \ot S) \,
    (\sigma_{3 \, 8} \cdot
    \sigma_{2 \, 8} \cdot
    \sigma_{1 \, 8} \cdot
    \Delta^4_{-1 \, -2 \, -3 \, 4} \delh\mone) 
    \nonumber \\
    & \qquad \qquad \qquad \qquad \qquad \qquad \quad \,\,\,\, \cdot
    (\sigma_{5 \, 8} \cdot
    \sigma_{6 \, 8}\mone \cdot
    \sigma_{7 \, 8})
    \big)
    \nonumber \\
    & =
    \frac{\sqrt{\mu}}{\mu^3 \, V_H^2} \left(\id \ot \int_{H^{5}} \ot \id \ot \int_A\right) \,
    \big(
    (1^{\ot 3} \ot \delh\mone \ot 1^{\ot 4}) \cdot
    (\sigma_{1 \, 8}\mone \cdot
    \sigma_{2 \, 8} \cdot
    \sigma_{3 \, 8}\mone \cdot
    (\sigma_{5 \, 8} \cdot
    \sigma_{6 \, 8}\mone \cdot
    \sigma_{7 \, 8})
    \big)
    \nonumber \\
    & =
    \frac{1}{\mu^2 \, V_H^2} \left(\id \ot \int_{H^{4}} \ot \id \ot \int_A\right) \,
    \big(
    (\sigma_{1 \, 7}\mone \cdot
    \sigma_{2 \, 7} \cdot
    \sigma_{3 \, 7}\mone \cdot
    (\sigma_{4 \, 7} \cdot
    \sigma_{5 \, 7}\mone \cdot
    \sigma_{6 \, 7})
    \big)
    \nonumber \\
    & =
    \frac{1}{\mu^2 \, V_H^2} \left(\id \ot \int_{H^{4}} \ot \id\right) \,
    \Delta^{6}_{-1 \, 2 \, -3 \, 4 \, -5 \, 6} \, \delh\mone \,.
\end{align}
Despite the cumbersome computation, the pattern of the gluing of two (partial) vertex amplitudes turned out to be very simple: take the tensor product of two (partial) vertex amplitudes and a (partial) propagator amplitude insert in the middle of them, connecting the last component of the first vertex with the first of the second. The co-integral of this expression reduces to a single (six dimensional) co-product defined in all the initial tensor spaces, except those that have been involved in the product with the propagator.
\begin{align}
    \left(\id \ot \int_{H^2} \ot \id\right) \, \coV_2 
    & = 
    \frac{1}{\sqrt{\mu^5} \, V_H^2} \left(\id \ot \int_{H^6} \ot \id\right) \,
    \big(\Delta^4_{-1 \, 2 \, -3 \, 4} \delh\mone \cdot
    \Delta_{-5 \, 4} \delh \cdot
    \Delta^4_{-5 \, 6 \, -7 \, 8} \delh\mone
    \big)
    \nonumber \\
    & =
    \frac{1}{\mu^2 \, V_H^2} \left(\id \ot \int_{H^{4}} \ot \id\right) \,
    \Delta^{6}_{-1 \, 2 \, -3 \, 4 \, -5 \, 6} \, \delh\mone \,. 
\end{align}
We can thus keep on gluing interaction terms by taking the co-integral of the alternate product of propagator and vertex amplitudes. Call $\coV_N$ the amplitude derived by gluing $N$ (partial) tetrahedron amplitudes through $N-1$ propagators.
\be
    \left(\id \ot \int_{H^{2N-2}} \ot \id\right) \, \coV_N
    = 
    \frac{1}{\mu^N \, V_H^{N}} \left(\id \ot \int_{H^{2N}} \ot \id\right) \,
    \Delta^{2N+2}_{-1 \, 2 \, -3 \, 4 \, \dots \, 2N \, -(2N+1) \, (2N+2)} \, \delh\mone \,.
\ee
In order to turn the above co-product into a \textit{loop} of links, we multiply it from the right by the (partial) propagator amplitude $\Delta_{-1 \, (2N+2)} \delh$, connecting the last and first tensor spaces. 
Again, after such gluing we take the co-integral over those tensor spaces and use its left invariance plus the normalization property of the delta function to remove the propagator amplitude. The result is a single loop of $2N$ links around an edge. 
Each of these links goes from the center of its tetrahedron to the center of the respective face. They can be repackaged as a product of $N$ links that connect the center of two tetrahedra.
The HAFT amplitude \eqref{GeneratingFunction_DualComplex} is %simply the sum over all the possible graphs $\Gamma_N^*$ 
given by combinations of all loops $\cL_N$ made of an arbitrary number $N$ of links, which span faces dual to the edges of the triangulation:
\be
    \cZ _{\Gamma^*}= %\sum_{\Gamma^*_{\{\cL\}}} 
    \prod_{\{\cL_N\}} \, \frac{1}{\mu^N \, V_H^{N}} \, \int_{H^{2 N}} \,
    \Delta^{2N}_{1 \, -2 \, 3 \, -4 \, \dots \, (2N-1) \, -2N} \, \delh\mone \,.
\ee
The expression of the Feynman amplitude in terms of the plane wave \eqref{GeneratingFunction_Double} can be derived by the above one simply using the definition of opposite delta function.
\be
    \cZ = %\sum_{\Gamma^*_{\{\cL\}}} 
    \prod_{\{\cL_N\}} \, \, \frac{1}{\sqrt{\mu^{2N+1}} \, V_H^{N}} \, \left(\int_{H^{2 N}} \ot \int_A\right) \,
    \big(\Delta^{2N}_{1 \, -2 \, 3 \, -4 \, \dots \, (2N-1) \, -2N} \ot \id\big) \, \sigma \,.
\ee

\subsection{Topological invariance}
\label{App_TopologicalInvariance}

Finally, following \cite{GirelliOecklPerez:2001PachnerSpinFoam}, we show that the fundamental building blocks of the Hopf algebra field theory Feynman diagrams are invariant under the action of the three dimensional Pachner moves \cite{Pachner1991Pachner}.
In three dimensions we have two Pachner moves, denoted $P_{(1,4)}$ and $P_{(2,3)}$. 
Let $\cA_{\cV^{\alpha}}$ be the amplitude of a combination of $\alpha$ tetrahedra, the two Pachner moves are defined as the maps
\be
    (P_{(1,4)} \, \cA_{\cV}) = \cA_{\cV^4}
    \,\,,\quad
    (P_{(2,3)} \, \cA_{\cV^2}) = \cA_{\cV^3} \,.
\ee
\begin{center}
    \textit{Amplitudes}
\end{center}
\smallskip
In order to prove that the Hopf algebra field theory Feynman diagram amplitudes are invariant under these moves, we first list the amplitudes of one, two, three and four tetrahedra. \\
For the amplitude of a single tetrahedron we consider the co-integral of the kernel \eqref{VertexAmplitude} in the $A$ representation.
\be
    \cA_{\cV}  =
    \frac{1}{\mu^6} \, \int_{A^{12}} \,
    \big(\uDe \dela\mone \cdot (\hat{\cC}_{1 \, 2 \, 3} \cdot \hat{\cC}_{4 \, 5 \, 6} \cdot \hat{\cC}_{7 \, 8 \, 9} \cdot \hat{\cC}_{10 \, 11 \, 12})\big) \,.
\ee
For simplicity, we used a shorthand tensor product notation for the closure constraints $\cC$.
Now, writing the co-product $\uDe$ as in \eqref{Underline_CoProduct_3d} and using Prop. \ref{Prop_ClosureConstraint_Gluing} for each of the co-products in its expression, we glue the four closure constraints.
The amplitude of a single tetrahedron is thus
\be
    \cA_{\cV} =
    \frac{1}{\mu^3} \, \int_{A^{6}} \,
    \big(\hat{\cC}_{-6 \, -4 \, -2} \cdot 
    \hat{\cC}_{-5 \, 6 \, -1} \cdot 
    \hat{\cC}_{-3 \, 4 \, 5} \cdot \hat{\cC}_{1 \, 2 \, 3}\big) \,.
    \label{GeneratingFunction_1-Tetra}
\ee
\smallskip \\
For the amplitude of two tetrahedra, we consider the tensor product of two copies of \eqref{GeneratingFunction_1-Tetra} and one propagator amplitude \eqref{VertexAmplitude} that connects the last face of the first tetrahedron (closure constraint $\hat{\cC}_{1\,2\,3}$) with the first of the second (closure constraint $\hat{\cC}_{-12\,-10\,-8}$).
\be
    \cA_{\cV^2} =
    \frac{1}{\mu^6 \sqrt{\mu^3}} \, \int_{A^{12}} \,
    \big(\hat{\cC}_{-6 \, -4 \, -2} \cdot 
    \hat{\cC}_{-5 \, 6 \, -1} \cdot 
    \hat{\cC}_{-3 \, 4 \, 5} \cdot 
    \hat{\cC}_{1 \, 2 \, 3}\big) \, 
    \Delta^{(3)}_{-1 \, -2 \, -3 \,;\, 10 \, 12 \, 8} \dela\mone \,
    \big(\hat{\cC}_{-12 \, -10 \, -8} \cdot 
    \hat{\cC}_{-11 \, 12 \, -7} \cdot 
    \hat{\cC}_{-9 \, 8 \, 11} \cdot 
    \hat{\cC}_{7 \, 10 \, 9}\big) \,.
\ee
Here we simply use Prop. \ref{Prop_Tetrahedra_Gluing} to merge the two amplitudes (by identifying the closure constraints) and remove the propagator
\begin{align}
    \cA_{\cV^2} & =
    \frac{1}{\mu^6} \, \int_{A^{9}} \,
    \big(\hat{\cC}_{-6 \, -4 \, -2} \cdot 
    \hat{\cC}_{-5 \, 6 \, -1} \cdot 
    \hat{\cC}_{-3 \, 4 \, 5}\big) \cdot 
    \hat{\cC}_{1 \, 2 \, 3} \cdot 
    \big(\hat{\cC}_{-11 \, 2 \, -7} \cdot 
    \hat{\cC}_{-9 \, 1 \, 11} \cdot 
    \hat{\cC}_{3 \, 9 \, 7}\big) 
    \nonumber \\
    & =
    \frac{1}{\mu^6} \, \int_{A^{9}} \,
    \big(\hat{\cC}_{-6 \, -4 \, -2} \cdot 
    \hat{\cC}_{-5 \, 6 \, -1} \cdot 
    \hat{\cC}_{-3 \, 4 \, 5} \cdot 
    \hat{\cC}_{1 \, 2 \, 3} \cdot 
    \hat{\cC}_{-9 \, 2 \, -7} \cdot 
    \hat{\cC}_{-8 \, 1 \, 9} \cdot 
    \hat{\cC}_{3 \, 8 \, 7}\big) \,,
    \label{GeneratingFunction_2-Tetra}
\end{align}
where in the second step we re-labelled the tensor spaces for simpicity.
\smallskip \\
Let us now move to the computation of the amplitude of three tetrahedra. The starting point is now a tensor product of three independent amplitudes of the type \eqref{GeneratingFunction_1-Tetra}, with three propagators used for identify the fourth and third triangles (closure constraints) of the first tetrahedron resp. to the first triangle of the second tetrahedron and to the second triangle of the fourth; plus a third propagator that identifies the fourth triangle of the second tetrahedron with the first triangle of the third.
\be
    \begin{aligned}
        \cA_{\cV^3} =
        \frac{1}{\mu^{9} \sqrt{\mu^9}} \, \int_{A^{18}} \,
        &
        \big(\hat{\cC}_{-6 \, -4 \, -2} \cdot 
        \hat{\cC}_{-5 \, 6 \, -1} \cdot 
        \hat{\cC}_{-3 \, 4 \, 5} \cdot 
        \hat{\cC}_{1 \, 2 \, 3}\big) \, 
        \Delta^{(3)}_{-1 \, -2 \, -3 \,;\, 10 \, 12 \, 8} \dela\mone \,\,
        \Delta^{(3)}_{-3 \, -4 \, -5 \,;\, 18 \, 17 \, 13} \dela\mone \,
        \\ 
        &
        \big(\hat{\cC}_{-12 \, -10 \, -8} \cdot 
        \hat{\cC}_{-11 \, 12 \, -7} \cdot 
        \hat{\cC}_{-9 \, 10 \, 11} \cdot 
        \hat{\cC}_{7 \, 8 \, 9}\big) \,
        \Delta^{(3)}_{-7 \, -8 \, -9 \,;\, 16 \, 18 \, 14} \dela\mone \,
        \\
        &
        \big(\hat{\cC}_{-18 \, -16 \, -14} \cdot 
        \hat{\cC}_{-17 \, 18 \, -13} \cdot 
        \hat{\cC}_{-15 \, 16 \, 17} \cdot 
        \hat{\cC}_{13 \, 14 \, 15}\big) \,.
    \end{aligned}
\ee
The pattern is similar to the computation of the amplitude of two tetrahedra \eqref{GeneratingFunction_2-Tetra}, with a subtle difference due to the fact that the three tetrahedra share one common edge. From a geometric prospective, the result is that this edge is internal to the combination of the tetrahedra (it is not part of the boundary) and hence one expects its contribution to disappear. From an algebraic point of view, this peculiarity is reflected from the fact that we can not use Prop. \ref{Prop_ClosureConstraint_Gluing} for gluing the closure constraints related to that edge. Instead, use first proposition \ref{Prop_InternalEdge} for gluing the closure constraints that involve the internal edge, the result is
\begin{align}
    \cA_{\cV^3} =
    \frac{\sqrt{\mu^3}}{\mu^{9} \sqrt{\mu^9}} \, \int_{A^{15}} \,
    &
    \big(\hat{\cC}_{-6 \, -4 \, -2} \cdot 
    \hat{\cC}_{-5 \, 6 \, -1} \cdot 
    % \hat{\cC}_{-3 \, 4 \, 5} \cdot 
    \hat{\cC}_{1 \, 2 \, 4 \, 5}\big) \, 
    \Delta^{(2)}_{-1 \, -2 \,;\, 10 \, 12} \dela\mone \,\,
    \Delta^{(2)}_{-4 \, -5 \,;\,17 \, 13} \dela\mone \,
    \nonumber \\ 
    &
    \big(\hat{\cC}_{-12 \, -10 \, -5 \, -4} \cdot 
    \hat{\cC}_{-11 \, 12 \, -7} \cdot 
    \hat{\cC}_{-9 \, 10 \, 11} \cdot 
    \hat{\cC}_{7 \, 4 \, 5 \, 9}\big) \,
    \Delta^{(2)}_{-7 \, -9 \,;\, 16 \, 14} \dela\mone \,
    \nonumber \\
    &
    \big(\hat{\cC}_{-5 \, -4 \, -16 \, -14} \cdot 
    \hat{\cC}_{-17 \, 4 \, 5 \, -13} \cdot 
    \hat{\cC}_{-15 \, 16 \, 17} \cdot 
    \hat{\cC}_{13 \, 14 \, 15}\big) \,.
\end{align}
Where the symbol $\hat{\cC}_{i_1 \, \dots \, i_n}$ stands for the closure constraint \eqref{ClosureConstraint} where the co-product $\Delta^3$ is re-placed by the $n^{th}$ co-product \eqref{1-d_CoProduct}.
Then, using Prop. \ref{Prop_Tetrahedra_Gluing} we compute the remaining propagators.
\begin{align}
    \cA_{\cV^3} =
    \frac{1}{\mu^{9}} \, \int_{A^{9}} \,
    &
    \big(\hat{\cC}_{-6 \, -4 \, -2} \cdot 
    \hat{\cC}_{-5 \, 6 \, -1} \cdot 
    \hat{\cC}_{1 \, 2 \, 4 \, 5}\big) \, 
    % \Delta^{(2)}_{-1 \, -2 \,;\, 10 \, 12} \dela\mone \,\,
    % \Delta^{(2)}_{-4 \, -5 \,;\,17 \, 13} \dela\mone \,
    \big(\hat{\cC}_{-1 \, -2 \, -5 \, -4} \cdot 
    \hat{\cC}_{-11 \, 2 \, -7} \cdot 
    \hat{\cC}_{-9 \, 1 \, 11} \cdot 
    \hat{\cC}_{7 \, 4 \, 5 \, 9}\big) \,
    % \Delta^{(2)}_{-7 \, -9 \,;\, 16 \, 14} \dela\mone \,
    \nonumber \\
    &
    \big(\hat{\cC}_{-5 \, -4 \, -7 \, -9} \cdot 
    \hat{\cC}_{-4 \, 4 \, 5 \, -5} \cdot 
    \hat{\cC}_{-15 \, 7 \, 4} \cdot 
    \hat{\cC}_{5 \, 9 \, 15}\big)
    \nonumber \\
    =
    \frac{1}{\mu^{9}} \, \int_{A^{9}} \,
    &
    \big(\hat{\cC}_{-5 \, -3 \, -2} \cdot 
    \hat{\cC}_{-4 \, 5 \, -1} \cdot 
    \hat{\cC}_{1 \, 2 \, 3 \, 4} \cdot 
    \hat{\cC}_{-8 \, 2 \, -6} \cdot 
    \hat{\cC}_{-7 \, 1 \, 8} \cdot 
    \hat{\cC}_{6 \, 3 \, 4 \, 7} \cdot
    \hat{\cC}_{-9 \, 6 \, 3} \cdot 
    \hat{\cC}_{4 \, 7 \, 9}\big) \,.
    \label{GeneratingFunction_3-Tetra}
\end{align}
In the last step we used \eqref{ClosureConstraint_Reduction} for removing the closure constraint $\hat{\cC}_{-1 \, -2 \, -5 \, -4}$ and the straightforward identity $\hat{\cC}_{-4 \, 4 \, 5 \, -5} = 1$. For the sake of clarity we also re-labelled some of the tensor spaces.
\smallskip \\
Last, let us compute the amplitude of four tetrahedra. The steps here are the same that we showed in the case of three tetrahedra. Hence, we start by the tensor product of four independent copies of the tetrahedron amplitude \eqref{GeneratingFunction_1-Tetra} where, multiplied by the propagators shown in the previous computation, plus three more that have the role of identifying the first, second and third triangles (closure constraints) of the fourth tetrahedron resp. with the fourth triangle of the third tetrahedron, the third of the second and the second of the first.
\be
    \begin{aligned}
        \cA_{\cV^4} =
        \frac{1}{\mu^{12} \sqrt{\mu^{18}}} \, \int_{A^{24}} \,
        &
        \big(\hat{\cC}_{-6 \, -4 \, -2} \cdot 
        \hat{\cC}_{-5 \, 6 \, -1} \cdot 
        \hat{\cC}_{-3 \, 4 \, 5} \cdot 
        \hat{\cC}_{1 \, 2 \, 3}\big) \, 
        \Delta^{(3)}_{-1 \, -2 \, -3 \,;\, 10 \, 12 \, 8} \dela\mone \,\,
        \Delta^{(3)}_{-3 \, -4 \, -5 \,;\, 18 \, 17 \, 13} \dela\mone \,
        \\ 
        &
        \big(\hat{\cC}_{-12 \, -10 \, -8} \cdot 
        \hat{\cC}_{-11 \, 12 \, -7} \cdot 
        \hat{\cC}_{-9 \, 10 \, 11} \cdot 
        \hat{\cC}_{7 \, 8 \, 9}\big) \,
        \Delta^{(3)}_{-7 \, -8 \, -9 \,;\, 16 \, 18 \, 14} \dela\mone \,
        \\
        &
        \big(\hat{\cC}_{-18 \, -16 \, -14} \cdot 
        \hat{\cC}_{-17 \, 18 \, -13} \cdot 
        \hat{\cC}_{-15 \, 16 \, 17} \cdot 
        \hat{\cC}_{13 \, 14 \, 15}\big) \,
        \Delta^{(3)}_{-5 \, -6 \, -1 \,;\, 22 \, 21 \, 23} \dela\mone \,
        \Delta^{(3)}_{-9 \, -10 \, -11 \,;\, 24 \, 23 \, 19} \dela\mone
        \\
        &
        \Delta^{(3)}_{-13 \, -14 \, -15 \,;\, 22 \, 24 \, 20} \dela\mone \,
        \big(\hat{\cC}_{-24 \, -22 \, -20} \cdot 
        \hat{\cC}_{-23 \, 24 \, -19} \cdot 
        \hat{\cC}_{-21 \, 22 \, 23} \cdot 
        \hat{\cC}_{19 \, 20 \, 21}\big) \,.
    \end{aligned}
\ee
The pattern is precisely the same of the computation of the amplitude of three tetrahedra \eqref{GeneratingFunction_3-Tetra}: first use (four times) proposition \ref{Prop_InternalEdge} for gluing the closure constraints that involve the four internal edges, and then proposition \ref{Prop_Tetrahedra_Gluing} to compute the remaining propagators.
\begin{align}
    \cA_{\cV^4} =
    \frac{\mu^6}{\mu^{12} \sqrt{\mu^{18}}} \, \int_{A^{9}} \,
    &
    \big(\hat{\cC}_{-6 \, -4 \, -2} \cdot 
    % \hat{\cC}_{-5 \, 6 \, -1} \cdot 
    % \hat{\cC}_{-3 \, 4 \, 5} \cdot 
    \hat{\cC}_{-5 \, 6 \, 2 \, 4 \, 5}\big) \, 
    \Delta_{-2 \,;\, 12} \dela\mone \,\,
    \Delta_{-4 \,;\, 17} \dela\mone \,
    \big(\hat{\cC}_{-12 \, -6 \, 5 \, -5 \, -4} \cdot 
    \hat{\cC}_{-11 \, 12 \, -7} \cdot 
    % \hat{\cC}_{-9 \, -5 \, 6 \, 11} \cdot 
    \hat{\cC}_{7 \, 4 \, 6 \, 11 \, -9 \, 9}\big) \,
    \nonumber \\
    &
    \Delta_{-7 \,;\, 16} \dela\mone \,
    \big(\hat{\cC}_{9 \, -11 \, -6 \, -4 \, -16 \, -9} \cdot 
    \hat{\cC}_{-17 \, 4 \, 6 \, 11 \, -9 \, 9 \, -11 \, -6} \cdot 
    \hat{\cC}_{-15 \, 16 \, 17} \cdot 
    \hat{\cC}_{6 \, 11 \, -9 \, 9 \, 15}\big) \,
    \Delta_{-6 \,;\, 21} \dela\mone \,
    \nonumber \\
    &
    \Delta_{-11 \,;\, 19} \dela\mone \,
    \Delta_{-15 \,;\, 20} \dela\mone \,
    \big(\hat{\cC}_{-9 \, 9 \, -11 \, -6 \, -20} \cdot 
    \hat{\cC}_{-6 \, 6 \, 11 \, -9 \, 9 \, -19} \cdot 
    \hat{\cC}_{-21 \, 5 \, -5 \, 6} \cdot 
    \hat{\cC}_{19 \, 20 \, 21}\big)
    \nonumber \\
    =
    \frac{1}{\mu^{12} \sqrt{\mu^6}} \, \int_{A^{9}} \,
    &
    \big(\hat{\cC}_{-6 \, -4 \, -2} \cdot 
    \hat{\cC}_{6 \, 2 \, 4}\big) \, 
    \Delta_{-2 \,;\, 12} \dela\mone \,\,
    \Delta_{-4 \,;\, 17} \dela\mone \,
    \big(\hat{\cC}_{-12 \, -6 \, -4} \cdot 
    \hat{\cC}_{-11 \, 12 \, -7} \cdot 
    \hat{\cC}_{7 \, 4 \, 6 \, 11}\big) \,
    \Delta_{-7 \,;\, 16} \dela\mone \,
    \nonumber \\
    &
    \big(\hat{\cC}_{-11 \, -6 \, -4 \, -16} \cdot 
    \hat{\cC}_{-17 \, 4} \cdot 
    \hat{\cC}_{-15 \, 16 \, 17} \cdot 
    \hat{\cC}_{6 \, 11 \, 15}\big) \,
    \Delta_{-6 \,;\, 21} \dela\mone \,
    \Delta_{-11 \,;\, 19} \dela\mone \,
    \Delta_{-15 \,;\, 20} \dela\mone \,
    \nonumber \\
    &
    \big(\hat{\cC}_{-11 \, -6 \, -20} \cdot 
    \hat{\cC}_{11 \, -19} \cdot 
    \hat{\cC}_{-21 \, 6} \cdot 
    \hat{\cC}_{19 \, 20 \, 21}\big)
    \nonumber \\
    =
    \frac{1}{\mu^{12}} \, \int_{A^{9}} \,
    &
    \big(\hat{\cC}_{6 \, 2 \, 4} \cdot 
    \hat{\cC}_{-2 \, -6 \, -4} \cdot 
    \hat{\cC}_{-11 \, 2 \, -7} \cdot 
    \hat{\cC}_{7 \, 4 \, 6 \, 11} \cdot 
    \hat{\cC}_{-11 \, -6 \, -4 \, -7} \cdot 
    \hat{\cC}_{-4 \, 4} \cdot 
    \hat{\cC}_{-15 \, 7 \, 4} \cdot 
    \hat{\cC}_{6 \, 11 \, 15} \cdot
    \hat{\cC}_{-11 \, -6 \, -15}
    \nonumber \\
    & \cdot 
    \hat{\cC}_{11 \, -11} \cdot 
    \hat{\cC}_{-6 \, 6} \cdot 
    \hat{\cC}_{11 \, 15 \, 6}\big)
    \nonumber \\
    =
    \frac{1}{\mu^{12}} \, \int_{A^{6}} \,
    &
    \big(\hat{\cC}_{-1 \, -3 \, -2} \cdot  
    \hat{\cC}_{-5 \, 1 \, -4} \cdot 
    \hat{\cC}_{4 \, 2 \, 3 \, 5} \cdot
    \hat{\cC}_{-6 \, 4 \, 2} \cdot 
    \hat{\cC}_{3 \, 5 \, 6}\big) \,.
    \label{GeneratingFunction_4-Tetra}
\end{align}
In the computation above, we have also used the identities \eqref{ClosureConstraint_Reduction}, \eqref{ClosureConstraint_CyclicPermutation+Inversion} and $\hat{\cC}_{i \, -i} = 1$, and at the end we re-labelled the tensor space for simplicity.
We now have all the ingredients to compute the Pachner moves.

\begin{center}
    \textit{Pachner move $P_{(1,4)}$}
\end{center}
\smallskip
The proof of the invariance of the amplitudes of Hopf algebra field theory under this move is particularly simple. Consider the amplitudes of one and four tetrahedra \eqref{GeneratingFunction_1-Tetra} and \eqref{GeneratingFunction_4-Tetra}
\be
    \begin{aligned}
        \cA_{\cV} & =
        \frac{1}{\mu^3} \, \int_{A^{6}} \,
        \big(\hat{\cC}_{-6 \, -4 \, -2} \cdot 
        \hat{\cC}_{-5 \, 6 \, -1} \cdot 
        \hat{\cC}_{-3 \, 4 \, 5} \cdot \hat{\cC}_{1 \, 2 \, 3}\big)
        \,,\\
        \cA_{\cV^4} & =
        \frac{1}{\mu^{12}} \, \int_{A^{6}} \,
        \big(
        \hat{\cC}_{-1 \, -3 \, -2} \cdot
        \hat{\cC}_{-5 \, 1 \, -4} \cdot 
        \hat{\cC}_{4 \, 2 \, 3 \, 5} \cdot
        \hat{\cC}_{-6 \, 4 \, 2} \cdot 
        \hat{\cC}_{3 \, 5 \, 6}\big) \,.
    \end{aligned}
\ee
We represent the transformation in Fig. \ref{Fig_Pachner(1,4)}. Notice that the closure constraint $\hat{\cC}_{4 \, 2 \, 3 \, 5}$ is not independent: in the picture on the left, it represents a loop of four edges on the boundary; it can be indeed removed by using the definition \eqref{ClosureConstraint} for the closure constraint $\hat{\cC}_{-5 \, 1 \, -4}$, which is thus expressed as the co-product of a delta function so that one can use the property \eqref{Delta_CoProduct_Property}
\be
    \hat{\cC}_{-5 \, 1 \, -4} \cdot \hat{\cC}_{4 \, 2 \, 3 \, 5} =
    \hat{\cC}_{-5 \, 1 \, -4} \cdot \hat{\cC}_{-5 \, 1 \, 2 \, 3 \, 5} =
    \hat{\cC}_{-5 \, 1 \, -4} \cdot \hat{\cC}_{1 \, 2 \, 3} \,.
\ee
We then remove the closure constraint $\hat{\cC}_{1 \, 2 \, 3}$ using \eqref{ClosureConstraint_Reduction}.
Finally, recalling that the tensor spaces are independent, we consider the identification below of the tensor space of the amplitude of four tetrahedra
\be
    1 \,\to\, 4
    \,\,,\quad
    2 \,\to\, 6 
    \,\,,\quad
    3 \,\to\, 2 
    \,\,,\quad
    4 \,\to\, -5
    \,\,,\quad
    5 \,\to\, 3
    \,\,,\quad
    6 \,\to\, 1 \,.
\ee
Under this identification (which is just a re-naming of the tensor spaces), it is straightforward to check that the two amplitudes are proportional each other. The action of the Pachner move on the amplitude of a single tetrahedron, is thus
\be
    (P_{(1,4)} \, \cA_{\cV}) 
    = 
    \cA_{\cV^4} =
    \frac{1}{\mu^9} \,\, \cA_{\cV} \,.
\ee
\smallskip 
\begin{center}
    \textit{Pachner move $P_{(2,3)}$}
\end{center}
\smallskip
The proof of the invariance of the Hopf algebra field theory amplitudes under the second Pachner move is equivalent to the first one. Consider the amplitudes of two \eqref{GeneratingFunction_2-Tetra} and three \eqref{GeneratingFunction_3-Tetra} tetrahedra
\be
    \begin{aligned}
        \cA_{\cV^2} & =
        \frac{1}{\mu^{6}} \, \int_{A^{9}} \,
        \big(\hat{\cC}_{-6 \, -4 \, -2} \cdot 
        \hat{\cC}_{-5 \, 6 \, -1} \cdot 
        \hat{\cC}_{-3 \, 4 \, 5} \cdot 
        \hat{\cC}_{1 \, 2 \, 3} \cdot 
        \hat{\cC}_{-9 \, 2 \, -7} \cdot 
        \hat{\cC}_{-8 \, 1 \, 9} \cdot 
        \hat{\cC}_{3 \, 8 \, 7}\big)
        \,,\\
        \cA_{\cV^3} & =
        \frac{1}{\mu^{9}} \, \int_{A^{9}} \,
        \big(\hat{\cC}_{-5 \, -3 \, -2} \cdot 
        \hat{\cC}_{-4 \, 5 \, -1} \cdot 
        \hat{\cC}_{1 \, 2 \, 3 \, 4} \cdot 
        \hat{\cC}_{-8 \, 2 \, -6} \cdot 
        \hat{\cC}_{-7 \, 1 \, 8} \cdot 
        \hat{\cC}_{6 \, 3 \, 4 \, 7} \cdot 
        \hat{\cC}_{-9 \, 6 \, 3} \cdot 
        \hat{\cC}_{4 \, 7 \, 9}\big) \,.
    \end{aligned}
\ee
In Fig. \ref{Fig_Pachner(2,3)} (on the left) we represented the gluing of two tetrahedra. 
Once again, we notice that the closure constraints $\hat{\cC}_{1 \, 2 \, 3 \, 4}$ and $\hat{\cC}_{6 \, 3 \, 4 \, 7}$ in the amplitude of three tetrahedra are not independent. In the picture, the represent
two loops of four edges of the boundary. Similarly to the first move, they can be removed by using the definition of closure constraint \eqref{ClosureConstraint} and then the property \eqref{Delta_CoProduct_Property} of the delta function, for the proper combination of the terms $\hat{\cC}_{-4 \, 5 \, -1}$, $\hat{\cC}_{-5 \, -3 \, -2}$ and $\hat{\cC}_{-7 \, 1 \, 8}$:
\be
    \begin{aligned}
        \hat{\cC}_{1 \, 2 \, 3 \, 4} &
        \quad\to\quad
        \hat{\cC}_{1 \, 2 \, 3 \, 5 \, -1} = \hat{\cC}_{2 \, 3 \, 5}
        \\
        \hat{\cC}_{6 \, 3 \, 4 \, 7} &
        \quad\to\quad
        \hat{\cC}_{6 \, 3 \, -3 \, -2 \, -1 \, 7} = \hat{\cC}_{6 \, -2 \, -1 \, 7}
        \quad\to\quad
        \hat{\cC}_{6 \, -2 \, -1 \, 1 \, 8} = \hat{\cC}_{6 \, -2 \, 8} \,. 
    \end{aligned}
\ee
The two resulting constraints can be r\eqref{ClosureConstraint_Reduction}.
In a similar way, we also remove the closure constraint $\hat{\cC}_{1 \, 2 \, 3}$ in the amplitude of two tetrahedra. This is associated to the triangle shared by the two tetrahedra in Fig. \ref{Fig_Pachner(2,3)} on the left, which is thus not independent on the other closures. To remove it, we use the same procedure explained above, with the closures $\hat{\cC}_{-3 \, 4 \, 5}$ and $\hat{\cC}_{-5 \, 6 \, 1}$.
The two resulting amplitudes then become
\be
    \begin{aligned}
        \cA_{\cV^2} & =
        \frac{1}{\mu^{6}} \, \int_{A^{9}} \,
        \big(\hat{\cC}_{-6 \, -4 \, -2} \cdot 
        \hat{\cC}_{-5 \, 6 \, -1} \cdot 
        \hat{\cC}_{-3 \, 4 \, 5} \cdot
        \hat{\cC}_{-9 \, 2 \, -7} \cdot 
        \hat{\cC}_{-8 \, 1 \, 9} \cdot 
        \hat{\cC}_{3 \, 8 \, 7}\big)
        \,,\\
        \cA_{\cV^3} & =
        \frac{1}{\mu^{9}} \, \int_{A^{9}} \,
        \big(\hat{\cC}_{-5 \, -3 \, -2} \cdot 
        \hat{\cC}_{-4 \, 5 \, -1} \cdot 
        \hat{\cC}_{-8 \, 2 \, -6} \cdot 
        \hat{\cC}_{-7 \, 1 \, 8} \cdot 
        \hat{\cC}_{-9 \, 6 \, 3} \cdot 
        \hat{\cC}_{4 \, 7 \, 9}\big) \,.
    \end{aligned}
\ee
As we did for the first move, consider the identification below of the tensor spaces of the amplitude of three tetrahedra
\be
    \begin{aligned}
        &
        1 \,\to\, -9
        \,,\\
        &
        4 \,\to\, 8
        \,,\\
        &
        7 \,\to\, 7
        \,,
    \end{aligned}
    \qquad
    \begin{aligned}
        &
        2 \,\to\, 6
        \,,\\
        &
        5 \,\to\, 1
        \,,\\
        &
        8 \,\to\, 2
        \,,
    \end{aligned}
    \qquad
    \begin{aligned}
        &
        3 \,\to\, 5
        \,,\\
        &
        6 \,\to\, 4
        \,,\\
        &
        9 \,\to\, 3 \,.
    \end{aligned}
\ee
Under this re-naming of tensor spaces, up to some permutations \eqref{ClosureConstraint_CyclicPermutation+Inversion} of the closure constraints, it is straightforward to check that the two amplitudes are proportional each other and the second Pachner move thus gives
\be
    (P_{(2,3)} \, \cA_{\cV^2}) 
    = 
    \cA_{\cV^3} =
    \frac{1}{\mu^3} \,\, \cA_{\cV^2} \,.
\ee

%% file: AppDerivationQuantumDouble.tex
Here we propose a procedure to derive skew symmetric bi-algebras, as in Def. \ref{Def_SkewSymmetricBi-Algebras}. 

\paragraph{Identification of the problem.}
Given any pair of bi-algebras $H$ and $A$, we would like to derive the skew pairing map $\sigma$ that satisfies the properties \eqref{SkewMap_Multiplication-CoProduct} to construct the generalized quantum double, or dually the skew co-pairing element $\sigma$ that satisfies the properties \eqref{SkewEl_id-CoProduct}, \eqref{SkewEl_CoProduct-id} to construct the dual of the quantum double. 
However, this derivation is highly non trivial, since the explicit expressions of both the map and the element $\sigma$ strongly depend on the choice of coordinates for each bi-algebra and even on the normal order of the bi-algebra generators.

\paragraph{Solution: reverse problem.}
Here we consider a possible solution to this problem.
The idea is to consider first a pair of co-algebras, to choose the  coordinates and the normal order for both of them. Then we define the skew pairing map and the skew co-pairing element, and derive the product of the two co-algebras in such that the map and element $\sigma$ satisfy the required properties. 
\smallskip \\
More specifically, consider a pair of co-algebras $C_A$ and $C_H$, and denote $a \in C_A$ and $h \in C_H$ their respective basis. Consider further the map $\sigma : C_H \ot C_A \to \bbC$, such that
\be
    \sigma(a_i \,,\, h_j) = \delta_{ij} \,.
\ee
We use this map $\sigma$ to define the  maps 
\be
    \begin{aligned}
        &
        \cL_h : C_H \to \bbC
        \,\,,\qquad
        \cL_h \equiv \sigma(\,\,\,,\,h) \,,\\
        &
        \cL_a : C_A \to \bbC
        \,\,,\qquad
        \cL_a \equiv \sigma(a\,,\,\,\,) \,.
    \end{aligned}   
\ee
Consider further an element $\sigma \in C_H \ot C_A$\footnote{By abuse of notation, we use the same symbol $\sigma$ for the map.} such that
\be
    \sigma(\tau \circ \sigma\mone) = 1 \,.
\ee
Last, let us consider two more maps, called \textit{star products}
\be
    \star : C_A \ot C_A \to C_A
    \,\,,\quad
    \ast : C_H \ot C_H \to C_H \,,
\ee
defined as
\be
    \begin{aligned}
        &
        a \star b = \big(\cL_a \ot \cL_b \ot \text{id}) \, (\tau \circ \Delta_H \ot \text{id}) \, \sigma
        \,,\\
        &
        h \ast g = \big(\text{id} \ot \cL_h \ot \cL_g) \, (\text{id} \ot \Delta_A) \, \sigma \,,
    \end{aligned} 
    \label{GenDef_StarProducts}
\ee
for all $h,g \in C_H$ and $a,b \in C_A$. The two star products are associative by co-associativity of the co-products. \\

\begin{proposition}[Bi-algebra compatibility]
    Consider two co-algebras $C_H$ and $C_A$ equipped with the star products \eqref{GenDef_StarProducts} as above. Denote the two resulting structures resp. $H$ and $A$, they are bi-algebras.
\end{proposition}
\begin{proof}
    Let us give a short proof of this proposition. In order for $H$ and $A$ to be bi-algebras, the two star products have to be compatible with the respective co-product
    \be
        \Delta_A (a \star b) = \Delta_A a \star \Delta_A b
        \,\,,\quad
        \Delta_H (h \ast g) = \Delta_H h \ast \Delta_H g \,.
    \ee
    Let us prove the first identity, the second is completely analogous. Let us suppose that the canonical element $\sigma \in C_H \ot C_A$ can be decomposed as $\sigma = \sum \, \sigma_H \ot \sigma_A$. Then, using the Sweedler notation, the property \eqref{SkewEl_id-CoProduct} can be expressed as
    \be
        (\id \ot \Delta_A) \, \sigma := 
        \sum \, \sigma_H \ot \sigma_{A_{(1)}} \ot \sigma_{A_{(2)}} =
        \sum \, \sigma_{H_1} \sigma_{H_2} \ot \sigma_{A_2} \ot \sigma_{A_1} :=
        \sigma_{13} \sigma_{12} \,.
    \ee
    Using this identity and the star product of $H$, the compatibility condition for the Hopf algebra $A$ follows as a direct computation.
    \begin{align}
        \Delta_A a \star \Delta_A b : & =
        (\cL_{\Delta_A a} \ot \cL_{\Delta_A b} \ot \id^{\ot 2}) \, (\tau \circ \Delta_{H \, 13} \ot \tau \circ \Delta_{H \, 24} \ot \id^{\ot 2}) \, (\sigma \ot \sigma) 
        \nonumber \\
        & =
        (\cL_{a_{(1)}} \ot \cL_{a_{(2)}} \ot \cL_{b_{(1)}} \ot \cL_{b_{(2)}} \ot \id^{\ot 2}) \, \big((\Delta_{H \, 31} \ot \id) \, \sigma) \, (\Delta_{H \, 42} \ot \id) \, \sigma)\big) 
        \nonumber \\
        & =
        \sum \, 
        \sigma(a_{(1)} \,,\, \sigma_{H_{1 \, (2)}}) \,
        \sigma(a_{(2)} \,,\, \sigma_{H_{2 \, (2)}}) \,
        \sigma(b_{(1)} \,,\, \sigma_{H_{1 \, (1)}}) \,
        \sigma(b_{(2)} \,,\, \sigma_{H_{2 \, (1)}}) \,
        \sigma_{A_1} \ot \sigma_{A_2}
        \nonumber \\
        & =
        \sum \, 
        \sigma(a \,,\, \sigma_{H_{2 \, (2)}} \ast \sigma_{H_{1 \, (2)}}) \,
        \sigma(b \,,\, \sigma_{H_{2 \, (1)}} \ast \sigma_{H_{1 \, (1)}}) \,
        \sigma_{A_{(1)}} \ot \sigma_{A_{(2)}}
        \nonumber \\
        & =
        (\cL_{a} \ot \cL_{b} \ot \Delta_A) \, (\tau \circ \Delta \ot \id) \, \sigma
        := \Delta_A (a \star b) \,.
    \end{align}
\end{proof}
\noindent
The two bi-algebras $H$ and $A$ are automatically skew paired by the map $\sigma$, which by construction of the star products, satisfies the axioms \eqref{SkewMap_Multiplication-CoProduct}, and they are also skew co-paired by the element $\sigma$, which by construction of the star products, satisfies the axioms \eqref{SkewEl_id-CoProduct}, \eqref{SkewEl_CoProduct-id}.
\smallskip \\
Note that, taking $C_A$ to be a trivial co-algebra (with a primitive co-product) the star product $\ast$ becomes a commutative pointwise product. As a particular example, the bi-algebras $H$ and $A$ can be taken respectively the bi-algebras $F(G)$ and $\cU(\g)$, with $G$ being a Lie group and $\g$ its Lie algebra. In this case the map $\cL_a$ can be seen as the ordinary Lie derivative on $G$.
% If the star products $\star$ and $\ast$ are compatible respectively with the co-products $\Delta_A$ and $\Delta_H$, the two co-algebras $C_A$ and $C_H$, equipped with the respective star products, become bi-algebras $A,H$. Such two bi-algebras are automatically skew paired by the map $\sigma$, which by construction of the star products, satisfies the axioms \eqref{SkewMap_Multiplication-CoProduct}. The two bi-algebras are also skew co-paired by the element $\sigma$, which by construction of the star products, satisfies the axioms \eqref{SkewEl_id-CoProduct}, \eqref{SkewEl_CoProduct-id}.
% \smallskip \\
% Note that, taking $C_A$ to be a trivial co-algebra (with a primitive co-product) the star product $\ast$ becomes a commutative pointwise product. As a particular example, the bi-algebras $H$ and $A$ can be taken   respectively  as the bi-algebras $F(G)$ and $\cU(\g)$, with $G$ being a Lie group and $\g$ its Lie algebra. In this case the map $\cL_a$ can be seen as the ordinary Lie derivative on $G$.
%
\begin{proposition}[Generalized quantum double of $\cU_q(\su(2))$ and its dual]
    Consider the $\cU_q(\su(2)) \cong F(\AN_q(2))$ and $\cU_q(\an(2)) \cong F(\SU_q(2))$ co-algebras given in \eqref{Bi-Algebra_Uq(su)} and \eqref{Bi-Algebra_F(SUq)}. 
    The full bi-algebra structures as well as the generalized quantum double of $\cU_q(\su(2))$ in Prop. \ref{Prop_qDeformedGenQuantumDouble} and its dual in Prop. \ref{Prop_qDeformedDualGenQuantumDouble} can be deduced by the skew pairing map \eqref{qDeformedSkewMap}
    \be
        \sigma\big(X_-^j H^i X_+^k \,,\, \varphi_+^b \phi^a \varphi_-^c\big)
        = i^{a+b+c} \,
        \delta_{ai} \delta_{bj} \delta_{ck} \,
        a![b]_{q^{2}}![c]_{q^{-2}}! \,,
    \ee
    and the skew co-pairing element \eqref{qDeformedSkewEL}
     \be
        \sigma
        =
        e_{\star \, q^{2}}^{i\varphi_+ \ot X_-} \,
        e_{\star}^{i\phi \ot H} \,
        e_{\star \, q^{-2}}^{i\varphi_- \ot X_+} \,.
    \ee
\end{proposition}
\begin{proof}
    In the propositions \ref{Prop_qDeformedGenQuantumDouble} and \ref{Prop_qDeformedDualGenQuantumDouble} we have already proved that the multiplication on $F(\SU_q(2))$ (in \eqref{Bi-Algebra_F(SUq)}) can be derived as a $\ast$-product, from the co-algebra sector of $\cU_q(\su(2))$ in \eqref{Bi-Algebra_Uq(su)}. 
    It remains to show that the multiplication of $\cU_q(\su(2))$ (in \eqref{Bi-Algebra_Uq(su)}) can be derived as a $\star$-product from the co-algebra sector of $F(\SU_q(2))$ in \eqref{Bi-Algebra_F(SUq)}.
    We show how to derive the most complicated commutator
    \be
        [X_+ , X_-] = \ell^2 q\mone \, \frac{\sinh(\lambda H)}{\sinh(\ell \lambda)} \,.
    \ee
    The computation of the commutators $[H,X_{\pm}]$ trivially follows in a similar way. 
    According to the definition \eqref{GenDef_StarProducts}, we consider
    \be
        \begin{aligned}
            [X_+ \,\overset{\star}{,}\, X_-] := &
            (\cL_{X_+} \wedge \cL_{X_-} \ot \id) \, (\Delta^{op} \ot \id) \, \sigma \\
            = &
            \big(\sigma(X_+ \wedge X_- \,,\,\,\,) \ot \id\big)\, (\Delta^{op} \ot \id) \, \sigma \,.
        \end{aligned}
    \ee
    Note that, as we evaluate the co-product of the skew co-pairing element in $\sigma(X_+ \,,\,\,\,)$ and $\sigma(X_- \,,\,\,\,)$, only the terms with a single power of $\varphi_- \ot \varphi_+$ or $\varphi_+ \ot \varphi_-$ survive. Hence, in the following computation we will use the symbol $\approx$ to discard all such non-relevant contributions.
    The co-product of the skew co-pairing element is
    \be
        \begin{aligned}
            (\Delta^{op} \ot \id) \, \sigma & =
            e_{\star \, q^{2}}^{i\Delta^{op}\varphi_+ \ot X_-} \,
            e_{\star}^{i\Delta^{op}\phi \ot H} \,
            e_{\star \, q^{-2}}^{i\Delta^{op}\varphi_- \ot X_+} 
            \\
            & =
            \sum_{u,v,w=0}^{\infty} \, \frac{i^{u+v+w}}{u![v]_{q^2}![w]_{q^{-2}}!} \,
            \big((\Delta^{op} \varphi_+)^v \, (\Delta^{op} \phi)^u \, (\Delta^{op} \varphi_-)^w\big) \ot X_-^v H^u X_+^w \,.
        \end{aligned}
    \ee
    The co-products of $\varphi_{\pm}$ are trivially computed using the generalized binomial theorem:
    \begin{align}
        (\Delta^{op} \varphi_+)^v & =
        \sum_{i=0}^v \, 
        \begin{bmatrix}
            v \\ i
        \end{bmatrix}_{q^{-2}} \, 
        (\varphi_0 e^{-i\ell \phi} \ot \varphi_+)^i 
        (\varphi_+ \ot e^{i\ell \phi}\varphi_0)^{v-i}
        \approx
        \delta_{v,0} 1 \ot 1 + \delta_{v,1}(\varphi_+ \ot 1 + 1 \ot \varphi_+) \,,\\
        (\Delta^{op} \varphi_+)^w & =
        \sum_{j=0}^w \, 
        \begin{bmatrix}
            w \\ j
        \end{bmatrix}_{q^{2}} \, 
        (e^{i\ell \phi}\varphi_0 \ot \varphi_-)^{j}
        (\varphi_- \ot \varphi_0 e^{-i\ell \phi})^{w-j} 
        \approx
        \delta_{w,0} 1 \ot 1 + \delta_{w,1}(\varphi_- \ot 1 + 1 \ot \varphi_-) \,.
    \end{align}
    The computation for the co-product of $\phi$ is more involved. Let us first simplify it by discarding all the unnecessary contributions.
    \begin{align}
        \Delta^{op} \phi & =
        \frac{i}{\ell} 
        \log\bigg(
        \frac{1}{\Delta^{op} \sqrt{1-q\mone\ell^2 \varphi_-\varphi_+}} 
        \big(\varphi_0 e^{-i\ell \phi} \ot \varphi_0 e^{-i\ell \phi} - \ell^2 \varphi_+ \ot \varphi_-\big)\bigg)
        \nonumber \\
        & \approx
        \frac{i}{\ell} 
        \log\bigg(
        \frac{1}{\sqrt{1 \ot 1 -q\mone\ell^2 (\varphi_- \ot e^{-i\ell \phi} + e^{i\ell \phi} \ot \varphi_-) (\varphi_+ \ot e^{i\ell \phi} + e^{-i\ell \phi} \ot \varphi_+)}}
        \big(e^{-i\ell \phi} \ot e^{-i\ell \phi} - \ell^2 \varphi_+ \ot \varphi_-\big)\bigg)
        \nonumber \\
        & \approx
        \frac{i}{\ell} 
        \log\bigg(
        \big(1 \ot 1 + \frac{1}{2}q\mone\ell^2 
        (e^{-i\ell \phi} \varphi_- \ot \varphi_+ e^{-i\ell \phi} + 
        \varphi_+ e^{-i\ell \phi} \ot e^{-i\ell \phi} \varphi_-)\big)
        \big(e^{-i\ell \phi} \ot e^{-i\ell \phi} - \ell^2 \varphi_+ \ot \varphi_-\big)\bigg)
        \nonumber \\
        & \approx
        \frac{i}{\ell} 
        \log\bigg(
        e^{-i\ell \phi} \ot e^{-i\ell \phi} + \frac{1}{2}\ell^2 \big(
        e^{-2i\ell \phi} \varphi_- \ot \varphi_+ e^{-2i\ell \phi} - 
        \varphi_+ \ot \varphi_-)
        \big)\bigg)
        \nonumber \\
        & \approx
        \frac{i}{\ell} 
        \log\Big(
        e^{-i\ell \phi} \ot e^{-i\ell \phi} + \frac{1}{2}\ell^2
        \varphi_- \wedge \varphi_+
        \Big) \,.
    \end{align}
    For simplicity, call $\Phi = \phi \ot 1 + 1 \ot \phi$ and $\psi = \frac{1}{2}\ell^2 (\varphi_- \wedge \varphi_+)$, and notice that 
    \be
        \Phi^n \psi^m = \psi^m (\Phi -2i\lambda m)^n
        \,\,,\quad
        e^{i\ell \Phi} \psi = q^{2} \, \psi e^{i\ell \Phi} \,.
    \ee    
    The co-product of $\phi$ then reduces to
    \begin{align}
        \Delta^{op} \phi & =
        \frac{i}{\ell} 
        \log\big(1 \ot 1 + (e^{-i\ell \Phi} - 1 \ot 1 + \psi)\big) =
        \frac{i}{\ell} \sum_{n=1}^{\infty} \, \frac{(-1)^{n+1}}{n} \, 
        (e^{-i\ell \Phi} - 1 \ot 1 + \psi)^n
        \nonumber \\
        & =
        \frac{i}{\ell} 
        \sum_{n=1}^{\infty} \, \frac{(-1)^{n+1}}{n} \, 
        \sum_{a=0}^n \, \binom{n}{a} (-1)^{n-a} 
        (e^{-i\ell \Phi} + \psi)^a
        \nonumber \\
        & =
        \frac{i}{\ell} 
        \sum_{n=1}^{\infty} \, \frac{(-1)^{n+1}}{n} \, 
        \sum_{a=0}^n \, \binom{n}{a} (-1)^{n-a} \,
        \sum_{k=0}^a \, 
        \begin{bmatrix}
            a \\ k    
        \end{bmatrix}_{q^{-2}} \,
        e^{-i(a-k)\ell \Phi} \psi^k
        \nonumber \\
        & \approx
        \frac{i}{\ell} 
        \sum_{n=1}^{\infty} \, \frac{(-1)^{n+1}}{n} \, 
        \sum_{a=0}^n \, \binom{n}{a} (-1)^{n-a} \,
        \big(e^{-ia\ell \Phi} + [a]_{q^{-2}} e^{-i(a-1)\ell \Phi} \psi \big)
        \nonumber \\
        & =
        \frac{i}{\ell} 
        \sum_{n=1}^{\infty} \, \frac{(-1)^{n+1}}{n} \, 
        \sum_{a=0}^n \, \binom{n}{a} (-1)^{n-a} \,
        e^{-ia\ell \Phi} 
        \Big(1 + \frac{1-q^{-2a-2}}{1-q^{-2}} e^{i\ell \Phi} \psi \Big)
        \nonumber \\
        & =
        \Phi \Big(1 + e^{i\ell \Phi} \psi\Big) +
        \frac{2i\lambda}{q^2-1} e^{i\ell \Phi} \psi
        \approx
        \Phi + \frac{2i\lambda}{q^2-1} e^{i\ell \Phi} \psi \,.
    \end{align}
    According to the product rule between $\Phi$ and $\psi$, we finally obtain
    \begin{align}
        (\Delta^{op} \phi)^u & =
        \Big(\Phi + \frac{2i\lambda}{q^2-1} e^{i\ell \Phi} \psi\Big)^u 
        \nonumber \\
        & =
        \Phi^u + 
        \Big(\Phi^{u-1} \frac{2i\lambda}{q^2-1} \psi + \Phi^{u-2} \frac{2i\lambda}{q^2-1} \psi \Phi + \cdots \Big) + 
        \Big(\Phi^{u-2} \Big(\frac{2i\lambda}{q^2-1}\psi\Big)^2 + \Phi^{u-3} \Big(\frac{2i\lambda}{q^2-1}\psi\Big)^2 \Phi + \cdots \Big) + \cdots
        \nonumber \\
        & \approx
        \Phi^u + \frac{2i\lambda}{q^2-1} \sum_{a=0}^{u-1} \, 
        \Phi^{u-1-a} \psi \Phi^a 
        \nonumber \\
        & =
        \Phi^u + \frac{2i\lambda}{q^2-1} 
        \sum_{a=0}^{u-1} \sum_{\alpha=0}^{u-1-a} \sum_{\beta=0}^{a} \,
        \binom{u-1-a}{\alpha} \binom{a}{\beta} \, (-i\lambda)^{u-1-a-\alpha} (i\lambda)^{a-\beta}
        :\Phi^{\alpha + \beta} \psi: 
        \nonumber \\
        & \approx
        \Phi^u + \frac{2i\lambda}{q^2-1} 
        \sum_{a=0}^{u-1} \, (-i\lambda)^{u-1} (-1)^{a} \, \psi
        =
        \Phi^u - 2\frac{(-i\lambda)^u}{q^2-1} \frac{1-(-1)^u}{2} \psi \,.
    \end{align}
    Above we used the symbol $:\,:$ to address at the normal order of the basis.
    Plugging the results in the $[X_+ \,\overset{\star}{,}\, X_-]$ commutator, we get
    \begin{align}
        [X_+ \,\overset{\star}{,}\, X_-] & =
        \delta_{v,1}\delta_{w,1}\delta_{u,0} \,
        \sigma\big(X_+ \wedge X_- \,,\,
         (\varphi_+ \ot \varphi_- + \varphi_- \ot \varphi_+)\big) \, X_-^v H^u X_+^w 
        \nonumber \\
        & \quad 
        - \delta_{v,0}\delta_{w,0} \sum_{u=0}^{\infty} \frac{i^u}{u!} \frac{(-i\lambda)^u\ell^2}{q^2-1} \frac{1-(-1)^u}{2} \,
        \sigma\big(X_+ \wedge X_- \,,\, \varphi_- \wedge \varphi_+\big) \, X_-^v H^u X_+^w
        \nonumber \\
        & = \ell^2 q\mone \frac{\sinh(\lambda H)}{\sinh(\ell\lambda)} \,.
    \end{align}
\end{proof}

%% file: AppExamples.tex
In order to make contact with more familiar versions of group field theories, we consider here some explicit realizations of Hopf algebra field theory. 
In the first example we pick the Hopf algebras associated to a finite group; this choice reduces the Hopf algebra field theory to ordinary group field theory based on finite groups.
In the second example we show how to derive ordinary three dimensional group field theories \cite{Boulatov:1992vp,Oriti:2006GFT} based on the $\SU(2)$ Lie group. This example can be alternatively derived as the limiting case of the group field theory based on the $q$-deformation of $\SU(2)$ presented in Sec. \ref{Sec_Examples}.

\subsection{Finite Groups}
\label{Sub_FiniteGroup}

Let $G$ a finite group of cardinality $|G|$, $F(G)$ the algebra of functions on $G$ and $\mathbb{C}[G]$ the group algebra constructed as the vector space generated by the span of the group elements $g \in G$. % \subset \mathbb{C}[G]$.
Since the generators of the algebra $\mathbb{C}[G]$ are group elements $g$, any element of the group algebra can be written as the linear combination $\bbC[G] \ni \hat{f} = \sum_u \hat{f} (g) g$, where the coefficients $\hat{f}(g)$ are regarded as functions on $G$. The $\bbC[G]$ bi-algebra structure can either be given in terms of the group element $g \in G$
\be
    \begin{aligned}
        \text{Co-product:} 
        & \quad
        \Delta g = g \ot g \,,\\
        \text{Co-unit:}
        & \quad 
        \varepsilon g = e \,,
        \label{Bi-Algebra_C[G]}
    \end{aligned}
\ee
or alternatively in terms of the coefficients $\hat{f}(u)$:
\be
    \begin{aligned}
        \text{Product:} 
        & \quad
        (\hat{f}_1 \cdot \hat{f}_2) (g) = \sum_v \hat{f}_1(v)\hat{f}_2(v^{-1} g) \,,\\
        \text{Co-product:} 
        & \quad
        \Delta \hat{f} (g,v) = \delta_g(v) \hat{f}(g) \,,\\
        \text{Co-unit:}
        & \quad 
        \varepsilon \hat{f} = \sum_g \hat{f}(g) \,.
    \end{aligned}
\ee
A basis for the algebra of functions $F(G)$ is given instead by the Kronecker delta on the group $\delta_g$, and the bi-algebra structure, given in terms of functions $f$ on the group, is
\be
    \begin{aligned}
        \text{Product:} 
        & \quad
        (f_1 \cdot f_2) (g) = f_1(g) \, f_2(g) \,,\\
        \text{Co-product:} 
        & \quad
        \Delta f(u,v) = f(u\cdot v) \,,\\
        \text{Co-unit:}
        & \quad 
        \varepsilon f = f(e) \,.
    \end{aligned}
    \label{Bi-Algebra_F(G)}
\ee
The co-integrals on the algebras $H$ and $A$ are the maps
\be
    \int_{F(G)} \equiv \frac{1}{|G|} \sum_g g : F(G) \rightarrow \mathbb{K}
    \,\,,\qquad
    \int_{\mathbb{C}[G]} \equiv \delta_e : \mathbb{C}[G] \rightarrow \mathbb{K} \,.
\ee
We used the symbol $e$ as the  identity in the group $G$.
\begin{proposition}[Generalized quantum double of a finite group]
    The generalized quantum double of the finite group $G$, $\cD(F(G),\bbC[G],\sigma)$, is given by the bi-algebras $F(G)$ and $\bbC[G]$ skew paired by the map
    \be
        \sigma(u \,,\, f) = f(u)
        \,\,,\quad
        \forall f \in F(G) , u \in G \subset \bbC[G] \,,
    \ee
    with mutual actions \eqref{MatchedPair}
    \be
        \begin{aligned}
            (\hat{f} \lhd f)(u) & = \varepsilon(f) \hat{f}(u) \,,\\
            (\hat{f} \rhd f)(u) & = \sum_v f(v\mone u v) \hat{f}(v) \,.
        \end{aligned}
    \ee
\end{proposition}
\begin{proof}
    Let us prove that the skew pairing $\sigma(u \,,\, f) = f(u)$ satisfies the relations \eqref{SkewMap_Multiplication-CoProduct}.
    We carry out the proof just using the bi-algebra structures \eqref{Bi-Algebra_F(G)} of $F(G)$ and \eqref{Bi-Algebra_C[G]} $\bbC[G]$. 
    Consider first
    \be
        \sigma(\Delta u \,,\, f_1 \ot f_2) := 
        f_1(u) f_2(u) = (f_1 \cdot f_2)(u) :=
        \sigma(u \,,\, (f_1 \cdot f_2)) \,. 
    \ee
    The second identity is instead
    \be
        \sigma(u_1 \ot u_2 \,,\, \Delta f) := 
        \Delta f(u_1 \ot u_2) = f(u_1 \cdot u_2) :=
        \sigma((u_1 \cdot u_2) \,,\, f) \,. 
    \ee
    Note that also the unit identities hold
    \be
        \sigma(e \,,\, f) = f(e) = 1
        \,\,,\quad
        \sigma(u \,,\, 1_f) = 1_f(u) = 1 \,,
    \ee
    where we defined $1_f(g) = 1, \forall g \in G$ as the constant function on the group (identity element of $F(G)$) and we normalized $f(e)$ to the identity.
\end{proof}
\noindent
\begin{proposition}[Dual of the generalized quantum double of a finite group]
    The dual of the generalized quantum double $\cD^*(\bbC[G],F(G),\sigma)$ of a finite group $G$ is given by the bi-algebras $\bbC[G]$ and $F(G)$ skew co-paired by the element
    \be
        \sigma = \sum_u \, \delta_u \ot u \,,
    \ee
    with mutual co-actions \eqref{QuantumDouble_CoActions}
    \be
        \begin{aligned}
            \alpha(f) & = f \ot 1 \,\\
            \beta(u) & = \sum_{v} \, \delta_v \ot v u v\mone \,.
        \end{aligned}
    \ee
\end{proposition}
\begin{proof}
    Let us prove that the skew co-pairing $\sigma = \sum_u \, \delta_u \ot u$ satisfies the relations \eqref{SkewEl_CoProduct-id}, \eqref{SkewEl_id-CoProduct}. 
    We carry out the proof just using the bi-algebra structures \eqref{Bi-Algebra_F(G)} of $F(G)$ and \eqref{Bi-Algebra_C[G]} $\bbC[G]$. 
    Consider first
    \begin{align}
        (\sigma_{13} \, \sigma_{23}) (v_1,v_2;f) & :=
        \sum_{u_1,u_2} \, \delta_{u_1}(v_1) \delta_{u_2}(v_2) \ot f(u_1 \cdot u_2) =
        \sum_{u_1,u_2} \, \delta_{u_1 \cdot u_2\mone}(v_1) \delta_{u_2}(v_2) \ot f(u_1) 
        \nonumber \\
        & \, =
        \sum_{u} \, \delta_{u}(v_1 \cdot v_2) \ot f(u) :=
        (\Delta \ot \id) \, \sigma (v_1,v_2;f) \,.
    \end{align}
    The second identity is instead
    \begin{align}
        (\sigma_{13} \, \sigma_{12}) (v,;f_1,f_2) & :=
        \sum_{u_1,u_2} \, (\delta_{u_1} \, \delta_{u_2})(v) \ot f_1(u_1) f_2(u_2) =
        \sum_{u_1,u_2} \, \delta_{u_1}(v) \delta_{u_2}(v) \ot f_1(u_1) f_2(u_2)
        \nonumber \\
        & \, =
        \sum_{u} \, \delta_{u}(v) \ot f_1(u) f_2(u) :=
        (\id \ot \Delta) \, \sigma (v;f_1,f_2) \,.
    \end{align}
    Note that also the co-unit identities hold
    \be
        (\varepsilon \ot \id) \, \sigma =
        \sum_u \, \varepsilon \delta_u \ot u =
        \sum_u \, \delta_u(e) \ot u = e
        \,\,,\quad
        (\id \ot \varepsilon) \, \sigma =
        \sum_u \, \delta_u \ot \varepsilon u =
        \delta_e \,.
    \ee
\end{proof}
\smallskip
\noindent
Using the dual of the quantum double of a finite group as the symmetry bi-algebra for the Hopf algebra field theory, the gauge projection \eqref{GaugeProj_Field} for the field $\Phi \in {F(G)}^3$ turns out to be
\be
    (\mathcal{P}_L \, \Phi)(u_1,u_2,u_3)
    =
    \frac{1}{|G|} \sum_h \, \Phi(h \cdot u_1, h \cdot u_2 , h \cdot u_3) \,.
\ee
The closure constraint \eqref{ClosureConstraint} is instead
\be
    \hat{\mathcal{C}} 
    =
    \frac{1}{|G|} \sum_u \, u \ot u \ot u \,.
\ee
Lastly, the amplitude \eqref{GeneratingFunction_Double} for a finite group is
\be
    \begin{aligned}
        \cA_{\Gamma^*} & = 
        %\sum_{\{\Gamma_N\}} 
        \prod_{\{\cL_N\}} 
        \sum_h \sum_{h_1, \dots , h_{2N}} \, 
        \delta_h(h_1 \cdot h_2\mone \dots h_{2N-1}\cdot h_{2N}) \ot \delta_e(h) \\
        & =
        \prod_{\{\cL_N\}} 
        \sum_{h_{i,i+1} ; i=1}^N \, 
        \delta_e(h_{1,2} \cdot h_{2,3} \dots h_{N,1}) \,,
    \end{aligned}
\ee
where we defined the composed links $h_{i,i+1} = h_{2i-1} \, h_{2i}\mone \in G$ between tetrahedron $i$ and tetrahedron $i+1$.

\subsection{$\SU(2)$ group field theory}
\label{SubSec_GFT}

%We call \textit{flat} case the Hopf algebra field theory sub-case when the Hopf algebras $H$ and $A$ are taken to be the (commutative Hopf) algebra of functions on a Lie group $F(G)$ and the respective (co-commutative Hopf) enveloping algebra $\cU(\g)$. 
We consider the Lie group $\SU(2)$ and the associated Hopf algebras $F(SU(2))$ and $F_{\star}(\mathbb{R}^3) \cong \mathcal{U}(\mathfrak{su}(2))$.
The Hopf algebra $F_{\star}(\mathbb{R}^3)$  is specified by
\be
    \begin{aligned}
        \text{Product:} 
        & \quad
        [x_0 , x_{\pm}]_\star = \pm 2\ell x_{\pm} \,,
        & \quad
        [x_+ , x_-]_\star = \ell x_0 \,,\\
        \text{Co-product:} 
        & \quad
        \Delta x_a = x_a \ot 1 + 1 \ot x_a \,,\\
        \text{Co-unit:}
        & \quad 
        \varepsilon(x_a) = 0 \,,\\
        \text{Antipode:}
        & \quad 
        S(x_a) = - x_a \,,
    \end{aligned}
    \label{Bi-Algebra_U(su)}
\ee
with $x_a = x_0,x_{\pm}$. 
Here the symbol $\star$ stands for the star product of the universal enveloping algebra $\cU(\su(2))$, which implies that the functions on it (such as the coordinates $x_a$) satisfy a non commutative product, which antisymmetric contribution is given in \eqref{Bi-Algebra_U(su)}. 

The Hopf algebra $F(\SU(2))$ is specified by
\be
    \begin{aligned}
        \text{Product:} 
        & \quad
        [p_a , p_b] = 0 \,,\\
        \text{Co-product:} 
        & \quad
        \Delta p_3 = p_3 \ot p_0 + p_0 \ot p_3 -i\ell p_+ \wedge p_- \,,\\
        & \quad
        \Delta p_{\pm} = p_{\pm} \ot p_0 + p_0 \ot p_{\pm} \pm i\ell p_{\pm} \wedge p_3 \,,\\
        & \quad
        \Delta p_0 = p_0 \ot p_0 + -\frac{1}{2}\ell^2 (2 p_3 \ot p_3 + p_+ \ot p_- + p_- \ot p_+) \,,\\
        \text{Co-unit:}
        & \quad 
        \varepsilon(p_a) = 0 \,,\\
        \text{Antipode:}
        & \quad 
        S(p_a) = - p_a \,,
    \end{aligned}
    \label{Bi-Algebra_F(SU)}
\ee
with $p_0 = \sqrt{1 - \ell^2 |p|^2}$ and $p_a = \{p_3,p_{\pm}\}$.
Note that by choosing $p_0 > 0$ the coordinates $p_a$ only cover half of the $\SU(2)$ group. To be rigorous, one should take both the coordinate patches spanned by $p_a$ with $p_0 > 0$ and $p_0 < 0$ \cite{Joung:2008mr}.
The co-integrals on $F(\SU(2))$ and $\cU(\su(2)) \cong \bbR^3_{\star}$ are resp. given by the usual Haar measure on $\SU(2)$ and by the Lebesgue measure on $\bbR^3$.
\medskip \\
Boulatov group field theory is a specific example of Hopf algebra field theory with the choices $H = F(SU(2))$ and $A= F_{\star}(\mathbb{R}^3) \cong \mathcal{U}(\mathfrak{su}(2))$. Since we use $F_{\star}(\mathbb{R}^3)$ to decorate the triangulation, the triangulation is "flat".
The dual complex is  decorated by the elements of $F(SU(2))$. 
\begin{proposition}[Group field theory: generalized quantum double]
    The generalized quantum double (Def. \ref{Def_GenQuantumDouble}) of $\SU(2)$ group field theory is the Drinfeld double of $\mathcal{U}(\mathfrak{su})$:
    \be
        \mathcal{D}(\mathbb{R}^3_{\star}, F(SU(2)), \sigma) 
        =
        \mathcal{D}(\mathcal{U}(\mathfrak{su}(2))) \cong
        F(SU(2)) \rtimes_{\sigma} \mathbb{R}^3_{\star} \,,
    \ee
    with pairing 
    \be
        \sigma(x_-^j x_0^i x_+^k \,,\, p_+^b p_3^a p_-^c) = i^{a+b+c} \, 
        \delta_{ia} \delta_{jb} \delta_{kc} \, a!b!c! \,,
    \ee
    and actions \eqref{QuantumDouble_Actions}
    \be
        \begin{aligned}
            &
            x_0 \rhd p_3 = 0 
            \,\,,\quad
            x_0 \rhd p_{\pm} = \mp 2 \ell p_{\pm} \,,\\
            &
            x_{\pm} \rhd p_3 = \pm 2 \ell p_{\pm} 
            \,\,,\quad
            x_{\pm} \rhd p_{\pm} = 0 
            \,\,,\quad
            x_{\pm} \rhd p_{\mp} = \mp 2 \ell p_3 \,,\\
        \end{aligned}
        \qquad
        x_a \lhd p_b = 0 \,.
    \ee    
\end{proposition}
\begin{proof}
    Let us prove that the skew pairing $\sigma(x_0^i x_+^j x_-^k \,,\, p_3^a p_-^b p_+^c) = i^{a+b+c} \, \delta_{ia} \delta_{jb} \delta_{kc} \, a!b!c!$ satisfies the relations \eqref{SkewMap_Multiplication-CoProduct}.
    We carry out the proof just using the bi-algebra structures \eqref{Bi-Algebra_U(su)} of $\cU(\su(2))$ and \eqref{Bi-Algebra_F(SU)} of $F(\SU(2))$.
    Consider first
    \be
        \begin{aligned}
            \sigma(\Delta (x_-^j x_0^i x_+^k) \,,\, p_+^b p_3^a p_-^c \ot p_+^s p_3^r p_-^t) & :=
            \sum_{u,v,w=0}^{i,j,k} \, 
            \binom{i}{u} \binom{j}{v} \binom{k}{w} \,
            \sigma(x_-^v x_0^u x_+^w \ot x_-^{j-v} x_0^{i-u} x_+^{k-w} \,,\, p_+^b p_3^a p_-^c \ot p_+^s p_3^r p_-^t) 
            \\
            & \, =i^{a+b+c+r+s+t} \,
            \delta_{i,a+r} \delta_{j,b+s} \delta_{k,c+t} \,
            \binom{a+r}{a} \binom{b+s}{b} \binom{c+t}{c} \, (a!b!c!r!s!t!) \,,\\ 
            \sigma(x_-^j x_0^i x_+^k \,,\, p_+^b p_3^a p_-^c \cdot p_+^s p_3^r p_-^t) & :=
            \sigma(x_-^j x_0^i x_+^k \,,\, p_+^{b+s} p_3^{a+r} p_-^{c+t}) \\ 
            & \, = i^{a+b+c+r+s+t} \,
            \delta_{i,a+r} \delta_{j,b+s} \delta_{k,c+t} \,
            (a+r)!(b+s)!(c+t)! \,.
        \end{aligned}
    \ee
    The second identity is more involved, since it requires the use of the star product on $\cU(\su(2))$. 
    We refer to App. \ref{App_Bi-Alg_Derivation} which provides a general setting for such type of computations.
    Note that also the unit identities hold
    \be
        \sigma(1 \,,\, p_+^b p_3^a p_-^c) = 
        \delta_{0a} \delta_{0b} \delta_{0c} = 1
        \,\,,\quad
        \sigma(x_-^j x_0^i x_+^k \,,\, 1) = 
        \delta_{i0} \delta_{j0} \delta_{k0} = 1 \,.
    \ee 
\end{proof}
\noindent
\begin{proposition}[Group field theory: dual of the generalized quantum double]
    The dual of the generalized quantum double (Def. \ref{Def_DualGenQuantumDouble}) of the standard group field theory is the dual of the Drinfeld double of $\cU(\su(2))$:
    \be
        \mathcal{D}(F(SU(2)), \mathbb{R}^3_{\star}, \sigma) = \mathbb{R}^3_{\star} \semicoprodl_{\sigma} F(SU(2)) \,,
    \ee
    with skew co-pairing element given by the star exponential
    \be
        \sigma = e_{\star}^{i p_+ \ot x_-} \, e_{\star}^{i p_3 \ot x_3} \, e_{\star}^{i p_- \ot x_+} \,,
    \ee
    with co-actions \eqref{QuantumDouble_CoActions}
    \be
        \alpha(p_a) = p_a \ot 1
        \,\,,\quad
        \begin{aligned}
            &
            \beta(x_0) = 1 \ot x_0 + 2i\ell (p_- \ot x_+ - p_+ \ot x_-)
            + \cO(\ell^2) \,,\\
            &
            \beta(x_{\pm}) = 1 \ot x_{\pm} \pm i\ell p_{\pm} \ot x_0 
            + \cO(\ell^2) \,.
        \end{aligned}
    \ee
    Since the full expression of second co-action involves the star product on $\cU(\su(2))$, for simplicity we truncated its expansion in $\ell$ at the first order.
\end{proposition}
\begin{proof}
    Let us prove that the skew co-pairing $\sigma = e_{\star}^{i p_+ \ot x_-} \, e_{\star}^{i p_3 \ot x_3} \, e_{\star}^{i p_- \ot x_+}$ satisfies the relations \eqref{SkewEl_id-CoProduct}, \eqref{SkewEl_CoProduct-id}. We carry out the proof just using the bi-algebra structures \eqref{Bi-Algebra_U(su)} of $\cU(\su(2))$ and \eqref{Bi-Algebra_F(SU)} of $F(\SU(2))$.
    Consider first
    \begin{align}
        (\id \ot \Delta) \, \sigma & :=
        e_{\star}^{i p_+ \ot \Delta x_-} \, e_{\star}^{i p_3 \ot \Delta x_3} \, e_{\star}^{i p_- \ot \Delta x_+} =
        e_{\star}^{i p_+ \ot 1 \ot x_-} e_{\star}^{i p_+ \ot x_- \ot 1} \,
        e_{\star}^{i p_3 \ot 1 \ot x_3} e_{\star}^{i p_3 \ot x_3 \ot 1} \,
        e_{\star}^{i p_- \ot 1 \ot x_+} e_{\star}^{i p_- \ot x_+ \ot 1} 
        \\ 
        & \, =
        e_{\star}^{i p_+ \ot 1 \ot x_-} e_{\star}^{i p_3 \ot 1 \ot x_3} e_{\star}^{i p_- \ot 1 \ot x_+} \,\,
        e_{\star}^{i p_+ \ot x_- \ot 1} e_{\star}^{i p_3 \ot x_3 \ot 1} e_{\star}^{i p_- \ot x_+ \ot 1} :=
        \sigma_{13} \, \sigma_{12} \,.
    \end{align}
    The second identity is more involved as it requires the definition of star exponential for $\cU(\su(2))$. 
    We refer to App. \ref{App_Bi-Alg_Derivation} which provides a general setting for such type of computations. 
    Note that also the co-unit relations hold
    \be
        (\varepsilon \ot \id) \, \sigma = 
        e_{\star}^{i 0 \ot x_-} \, e_{\star}^{i 0 \ot x_3} \, e_{\star}^{i 0 \ot x_+} = 1
        \,\,,\qquad
        (\id \ot \varepsilon) \, \sigma =
        e_{\star}^{i p_+ \ot 0} \, e_{\star}^{i p_3 \ot 0} \, e_{\star}^{i p_- \ot 0} = 1 \,.
    \ee
\end{proof}
\smallskip
\noindent
Let us now derive the fundamental ingredients of ordinary three dimensional group field theory using the dual of the generalized quantum double above. The gauge projection \eqref{GaugeProj_Field} for the field $\Phi \in F(\SU(2)^{\times 3})$ turns out to be
\be
    (\cP_L \, \Phi)(p^1,p^2,p^3) = \int_{\SU(2)} [\dd k^3] \,
    \Phi(k \op p^1, k \op p^2, k \op p^3) \,,
\ee
where $[\dd k^3]$ is the Haar measure on $\SU(2)$ in the coordinates $\{k_3,k_{\pm}\}$ and  $(k \op p^i)_a$ is the coordinate $a=3,\pm$ of the product of two group elements, where the left element is parametrized by the coordinates $k_a(h)$ and the right one by the coordinates $p_a(g)$. 
%One can write the coordinates $p_a,k_a$ in a more compact form, by packaging them in the $\SU(2)$ group elements $g_i$ parametrized by the coordinates $p_a^i$ and $h$, parametrized by the coordinates $k_a$. 
Omitting the coordinates, the gauge averaging assumes the more familiar form
\be
    (\cP_L \, \Phi)(g_1,g_2,g_3) = \int_{\SU(2)} [\dh] \,
    \Phi(hg_1,hg_2,hg_3) \,.
\ee
Noting $x_a^i \equiv 1 \ot \cdots 1 \ot x_a \ot 1 \cdots \ot 1$ for $a=0,\pm$, the closure constraint as given in \eqref{ClosureConstraint} would be
\be
    \hat{\cC}(x^1,x^2,x^3) =  \delta(x_1 + x_2 + x_3) \,.  
\ee

The Feynman diagram amplitude \eqref{GeneratingFunction_Double} is 
\be
    \cA_{\Gamma^*} = \prod_{\{\cL_N\}} \, \int [\dd p^3]^{N} [\dd x^3] \,
    e_{\star}^{i \Delta^{N} p_+ \ot x^e_-} \,
    e_{\star}^{i \Delta^{N} p_3 \ot x^e_3} \,
    e_{\star}^{i \Delta^{N} p_- \ot x^e_+} \,
    \,,
    \label{GeneratingFunction_GroupFieldTheory}
\ee
where we normalized the volume of the group $\SU(2)$ to the identity. 
The element living in the $n^{th}$ tensor space of the co-product $\Delta^N p_a$ is the coordinate $a=3,\pm$ of the $\SU(2)$ group element that decorates the link that goes from the center of tetrahedron $n$ to the center of tetrahedron $n+1$, where tetrahedron $N+1$ is identified to the first tetrahedron.
As explained in Sec. \ref{sec:pw-disc}, this expression is associated to the discretization of $BF$ theory without any cosmological constant, where the $x^e_a$ are the coordinates of the variable associated to the frame field discretized on an edge $e$ of the triangulation $\Gamma$, while the $N^{th}$ co-product gives the coordinates of the curvature $F(A)$ around such edge, where $A$ is the connection whose holonomy is discretized on the links of the dual complex.